\documentclass[11pt,a4paper]{article}            
 \usepackage[skins,theorems]{tcolorbox}
\tcbset{highlight math style={enhanced,
  colframe=red,colback=white,arc=0pt,boxrule=1pt}}
  \usepackage[bookmarksopen, bookmarksnumbered, bookmarksopenlevel=2]{hyperref}
  \usepackage{tikz}
  \usepackage{tikz-3dplot}
 \usetikzlibrary{calc}
 \usetikzlibrary{decorations} %
 \usepackage[UKenglish]{babel}
 \usepackage[toc,page]{appendix}
 \usepackage{amsmath}
 \usepackage{amssymb}
 \usepackage{amsthm}
 \usepackage{graphicx}
 \usepackage{hhline}
 \usepackage[bf]{caption}
\usepackage{cite}
\usepackage[vcentermath]{youngtab}
\usepackage{geometry}
\usepackage{changepage}
\usepackage{slashed}
\usepackage{color}
\usepackage{stackrel}
\usepackage{tikz-cd} 
\usepackage{mathtools}
\usepackage{cancel} 
\usepackage{multirow}
\usepackage{cleveref}
\usepackage[margin=0pt,font=small,labelfont=normalfont,skip=22pt]{subcaption}
\usepackage{diagbox}
\usepackage[table]{xcolor}
\definecolor{coneYellow}{RGB}{255,240,0}
\definecolor{conePink}{RGB}{245,175,175}
\colorlet{conePinkPinklight}{conePink!50!white}
\definecolor{coneBlue}{RGB}{135,125,225}
\colorlet{coneBluelight}{coneBlue!50!white}
\colorlet{conePinklight}{coneBluelight!50!white}

\newcommand{\Y}{\cellcolor{conePinkPinklight}}
\newcommand{\Pnk}{\cellcolor{coneBluelight}}
\newcommand{\Pk}{\cellcolor{conePinklight}}

\usepackage{empheq}
\usepackage{arydshln}
\usepackage{yhmath}
\usepackage{setspace}
\newcommand{\bqa}{\begin{eqnarray}}
\newcommand{\eqa}{\end{eqnarray}}

\newcommand{\Z}{\mathbb{Z}}

\newcommand{\NE}{\overline{\mathrm{NE}}}

\newcommand{\Nef}{\mathrm{Nef}}

\newcommand{\Cone}{\operatorname{Cone}}

\newenvironment{eqn*}{\begin{equation*}\begin{aligned}}{\end{aligned}\end{equation*}\noindent}
\hypersetup{
    pdftitle={},
    pdfauthor={},
    pdfsubject={}
}
\numberwithin{equation}{section}
\numberwithin{table}{section}

\makeatletter

\def\punct{:}

\newtheoremstyle{dotless}{3pt}{3pt}{\itshape}{}{\bfseries}{\punct}{.5em}{}
\theoremstyle{dotless}

\newtheorem{theorem}{Theorem}

\newtheorem{definition}{Definition}
\newtheorem*{definition*}{Definition}

\newtheorem{proposition}{Proposition}
\newtheorem{conj}{Conjecture}
\newtheorem*{conj*}{Conjecture}

\newcommand{\be}{\begin{equation}}
\newcommand{\ee}{\end{equation}}
\newcommand{\beq}{\begin{equation}}
\newcommand{\eeq}{\end{equation}}
\newcommand{\ba}{\begin{aligned}}
\newcommand{\ea}{\end{aligned}}

\newcommand{\bea}{\begin{eqnarray}}
\newcommand{\eea}{\end{eqnarray}}

\newcommand{\cO}{\mathcal{O}}

\newcommand{\cE}{\mathcal{E}}

\newcommand{\cC}{\mathcal{C}}

\newcommand{\cL}{\mathcal{L}}

\newcommand{\cN}{\mathcal{N}}

\newcommand{\cA}{\mathcal{A}}
\newcommand{\cH}{\mathcal{H}}

\newcommand{\cF}{\mathcal{F}}

\newcommand{\cS}{\mathcal{S}}

\newcommand{\cV}{\mathcal{V}}

\newcommand{\cM}{\mathcal M}

\newcommand{\volume}{\text{{\small} vol}\, }

\newcommand\bi{\begin{itemize}}
\newcommand\ei{\end{itemize}}

\def\bra#1{{\langle{#1}|}}

\def\ket#1{{|{#1}\rangle}}

\def\unit{{1\kern-.65ex {\rm l}}}
\def\1{{1\kern-.65ex {\rm l}}}

\def\pl{{\mathrm{pl}}}

\def\dd{{\mathrm{d}}}

\newcount\hour \newcount\minute
\hour=\time \divide \hour by 60
\minute=\time
\def\now{%
\ifnum \hour<13
  \ifnum \hour=0 \advance \hour by 12 \number\hour:\else \number\hour:\fi%
     \ifnum \minute<10 0\fi%
     \number\minute%
\ A.M.%
\else \advance \hour by -12 \number\hour:%
  \ifnum \minute<10 0\fi%
  \number\minute%
  \ P.M.%
\fi%
}

\makeatother

\begin{document}

\begin{titlepage}
\begin{center}
\rightline{\small }

\vskip 15 mm

{\LARGE \bf
On BPS Branes
} 
\vskip 11 mm

Cumrun Vafa, David H. Wu, and Kai Xu
\vskip 11 mm
\small 
{\it Jefferson Physical Laboratory, Harvard University, Cambridge, MA 02138, USA}

\end{center}
\vskip 17mm

\begin{abstract}
We study supersymmetric BPS branes (BPS-B) in supergravity theories. Some of these states are anticipated by BPS black-brane (BPS-BB) solutions of supergravity. In particular, we define and distinguish the cone generated by BPS branes from the subcone of charges that admit BPS black-brane attractor solutions in the infrared limit of the supergravity effective field theory. We denote these cones by $\cC_{\rm BPS-B}$ and $\cC_{\rm BPS-BB}$, respectively.
We conjecture that, in any supersymmetric theory of quantum gravity, every integrally charged state lying in $\cC_{\rm BPS-BB}$ is realized by a BPS state in the spectrum. Furthermore, we conjecture and present evidence that when $\cC_{\rm BPS-B}$ is moduli independent, it can be determined as the cone dual to the $\cC_{\rm BPS-BB}$ under the electric-magnetic pairing.
\end{abstract}

\vfill
\end{titlepage}

\newpage

\tableofcontents

\setcounter{page}{1}
\section{Introduction}

One of the central features of a consistent quantum theory of gravity is the completeness~\cite{Polchinski:2003bq,Banks:2010zn} of charge lattices for all gauge symmetries (including higher form gauge symmetries).  This is related to the lack of global symmetries~\cite{Banks:2010zn} in quantum gravity theories, or more broadly the triviality of cobordism groups~\cite{McNamara:2019rup} for quantum gravity theories.  However, this feature by itself does not predict the masses or tensions of charged states.  The weak gravity conjecture (WGC) postulates~\cite{Arkani-Hamed:2006emk} that a finite index sublattice is generated by charged states whose tension (or mass) is less than the corresponding extremal black brane of the same charge
\begin{equation}
T(q)\leq T^{\rm Ext}_{\rm Black\, Brane}(q)\,,
\end{equation}
where $q$ denotes the charge.
There is also a sublattice version of the WGC~\cite{Heidenreich:2016aqi} which proposes that this relation holds not only for the generators of the sublattice but for all elements of the sublattice.

While these are general expectations for the tensions of these objects, there are sharper expectations in the supersymmetric cases.  Exploration of what the precise expectations for supersymmetric cases are is the main aim of this work.  The basic ingredient in the supersymmetric case is the existence of BPS branes, whose mass is protected by supersymmetry.  When there is an extremal black brane, we would have
\begin{equation}
T_{\rm BPS}(q)=T^{\rm Ext}_{\rm Black\, Brane}(q)\,,
\end{equation}
saturating the WGC bound. One important question in this context is to ask for which charge states $q$ are there supersymmetric BPS states, which will be the main focus of this work.

One can study BPS objects in the supergravity context with fixed black brane charges as well. In such cases, the asymptotic values of scalar fields (which are viewed as compactification moduli) flow and approach a particular value at the horizon of the black brane called the attractor value.  This minimizes the tension of the brane (in the Einstein frame), as a function of moduli.  When the attractor value is not singular, (i.e., lead to no additional massless modes in the EFT)
the supergravity solution predicts the existence of a BPS charged state. For a sufficiently large charge $q\rightarrow N q$, such a solution would be a reliable prediction of the existence of a black brane object.  The question would be whether this is valid for arbitrary charge $q$ which may not be large, where EFT is inapplicable.  We conjecture that despite the lack of reliable argument from EFT, this is still the case, and we provide strong evidence in support of this.

BPS states do not all correspond to black branes.  We can define a cone by considering arbitrary positive combinations of BPS charges.  This would be natural to consider only when the BPS states preserve the same supersymmetry. In the present work, we consider supersymmetric theories and BPS branes which are moduli independent. This is, for example, the case for BPS 1-branes in 6d ${\cal N}=(1,0)$ and also for BPS 1-branes, BPS 0-branes in 5d ${\cal N}=1$. 
In these cases,
we propose and provide evidence that the BPS electric black branes, are dual to BPS magnetic cone and BPS electric cone is dual to BPS magnetic black branes, when restricted to the same preserved supersymmetry sector and when the BPS branes can be defined independently of the moduli space. With the cone of BPS black branes defined, we also conjecture that there is a BPS state in the spectrum of the quantum gravitational theory for every integral charge site inside this cone.

The paper is organized as follows: In \cref{sec:def}, we introduce the definitions of central concepts in this paper, namely BPS compatibility, BPS occupancy, BPS branes (BPS-B), and BPS black branes (BPS-BB). Using these definitions, we introduce the associated charge cones of BPS-B and BPS-BB; In \cref{sec:conj}, we formulate the main conjectures about BPS completeness within the BPS-BB cone, the duality between BPS compatible cones via their elctric-magnetic pairings, and the classification of tensions of BPS branes across the BPS-B cone; In \cref{sec:QG with 32}, we show that M2-branes and M5-branes in 11d M-theory and D3-branes in 10d type IIB string theory are evidences for our conjectures; In \cref{sec:ECalabi--Yau threefolds} and \cref{sec:M-theory on Calabi--Yau threefold}, we provide more evidence for this conjecture in non-trivial cases using concrete string theory compactifications, namely F-theory and M-theory compactified on Calabi--Yau threefolds; and the paper concludes in \cref{sec:conclusions} with a summary and future directions.

\section{Definitions}
\label{sec:def}

In this section, we define the BPS states of interest in a given theory and clarify the correlation between the orientation of the worldvolume of the BPS brane and the supersymmetry preserved for such a BPS brane. We define BPS occupancy for a given charge in the charge lattice of a theory. The existence of a BPS state in the spectrum can be predicted from the weak gravity conjecture. Upon establishing these, we define the BPS brane (BPS-B) cone, $\cC_{\rm BPS-B}$, which is the positive real cone over the integral BPS charge lattices associated to a fixed supersymmetric sector of a BPS $p$-brane of any given theory.  For a fixed $\cC_{\rm BPS-B}$, we also define the BPS black brane (BPS-BB) cone, $\cC_{\rm BPS-BB}$, as a subcone of $\cC_{\rm BPS-B}$.

\subsection{BPS compatibility and supersymmetry}
\label{sec:BPS compatibility}

In $d$-dimensional Minkowski spacetime, the spacetime symmetry algebra contains the $d$-dimensional Poincar\'e algebra $\mathfrak{iso}(1,d-1)$. A supersymmetric extension can be obtained by adjoining the algebra with supercharges which transform as spinors of the Lorentz group, $Spin(1,d-1)$. The amount of supersymmetry is determined by the real spinor representations of $Spin(1,d-1)$. Furthermore, by unitarity, for any such state $\ket{\psi}$ in the theory, we have
\begin{equation}
    \bra{\psi}\{f(\eta)\,,f(\eta)^\dagger\}\ket{\psi}\geq0\,,
\end{equation}
where $f(\eta)=\eta^\alpha_IQ_\alpha^I$ denotes any complex linear combination of supercharges $Q_\alpha^I$ with $\alpha$ and $I$ labeling the spinor index and the independent copies of the minimal spinor representation, respectively. From this along with the supersymmetry algebra, the tension of extended objects such as $p$-branes are bounded below by
\begin{equation}
\label{eq:BPS bound}
    T^{(p)}\geq |Z^{(p)}|\,,
\end{equation}
where $|Z^{(p)}|$ denotes the norm of the $p$-form extension of the supersymmetry algebra. Then, a \textit{BPS state} is a state in the theory that saturates the above inequality. Correspondingly, there exists a non-zero linear combination of supercharges such that
\begin{equation}
    f(\eta)\ket{\psi_{\rm BPS}}=0\,.
\end{equation}
When such a state is found in the spectrum, we refer to these BPS states as BPS $p$-branes.

Consider a BPS brane where the supersymmetry satisfies
\begin{equation}
    \Gamma_{\rm brane}\epsilon=\epsilon\,.
\end{equation}
where $\Gamma_{\rm brane}$ denotes the supersymmetry projectors relevant to the brane configuration in spacetime while $\epsilon$ denotes the relevant spinors in a given dimension preserved along the BPS brane. We can also consider the preserved supersymmetry subspace associated to these branes as
\begin{equation}
    S_i=\{\epsilon\in S_{\rm vac}\,\vert\, \Gamma_{\rm brane}\epsilon=\epsilon\}\,,
\end{equation}
where $S_{\rm vac}$ denotes the space of supersymmetry parameters in the vacuum theory, without any insertion of branes.

Relevant for our later discussions, we introduce the notion of \textit{BPS compatibility} among BPS branes, which in the case of BPS branes with the same dimensions is equivalent to preserving the same supersymmetry. Namely, for BPS branes of the same dimensions to be BPS compatible, this requires $S_i=S_j$ for a pair of BPS $p_i$-branes and BPS $p_j$-branes. This would, for example, in the cases of BPS 0-branes in 4d $\cN=2$ require the associated central charges of BPS 0-branes to have the same complex phase. In heterotic theories on $T^n$, BPS states preserving the same supersymmetry is a fixed choice of $\hat{p}_R=\vec{p}_R/|\vec{p}_R|\in S^{n-1}\subset \mathbb{R}^n$. 

For BPS branes with different worldvolume dimensions, the branes are BPS compatible when the following three conditions are satisfied between a pair of BPS $p_i$-branes and BPS $p_j$-branes with $p_i\neq p_j$: 1) the BPS brane worldvolume theories preserve the same amount of supersymmetry, $\dim S_i=\dim S_j$; 2) their associated preserved supersymmetry subspace must have a non-trivial intersection, $\dim (S_i \cap S_j)>0$; and 3) the charge orientation of the BPS branes must also be aligned.  
In this present paper, the dualities between cones realized via the electric-magnetic pairing which we discuss in \cref{sec:duality between cones}, we restrict to BPS branes in theories where the choice of orientation is $\mathbb{Z}_2$, namely either the branes are BPS or anti-BPS. In other words, we focus on theories where preserving the same supersymmetry is automatic (up to charge conjugation) as we illustrate in the class of theories below.

Let us consider 5d $\cN=1$ supergravity theories. The supersymmetry algebra is 
\begin{equation}
\label{eq:5d n=1 susy algebra}
    \big\{Q^A_\alpha,(Q_\beta^B)^\dagger\big\}=2\big[\delta^{AB}(\Gamma^\mu\Gamma^0)_{\alpha\beta}P_\mu+(\sigma_3)^{AB}(i\Gamma^0)_{\alpha\beta}Z-(\sigma_1)^{AB}(\Gamma_\mu\Gamma^0)_{\alpha\beta}Z^\mu\big]\,,
\end{equation}
where $\{\Gamma^\mu,\Gamma^\nu\}=2\eta^{\mu\nu}$ with $\mu,\nu=0,1,\dots,4$, $\sigma_1,\sigma_3$ denote the relevant Pauli matrices, $A,B=1,2$ denotes the $SU(2)_R$ index, and $\alpha,\beta=1,\dots,4$ denotes the 5d spinor index. Let us start with the 0-branes in this theory. In the rest frame of a massive 0-brane with charge $q$ and mass $m$, the above supersymmetry algebra becomes $\{Q^A_\alpha,(Q^B_\beta)^\dagger\}=2\big[m\delta^{AB}\delta_{\alpha\beta}+(\sigma_3)^{AB}(i\Gamma^0)_{\alpha\beta}Z_q\big]$. In the $\eta^{\mu\nu}=(-,+,\dots,+)$ signature, $(i\Gamma^0)^2=1$ and the eigenvalues of $m+(i\Gamma^0\sigma_3)Z_q$ are $m\pm Z_q$. Therefore, the BPS bound on the mass of massive 0-branes is $m\geq |Z_q|$ where this is saturated for BPS 0-branes. For a BPS state with $m_{\rm BPS}=|Z_q|$ where $Z_qe^{i\theta_q}=|Z_q|$, the annihilating supersymmetries obey $\Pi_q\epsilon=\frac12\big(1-e^{i\theta_q}i\Gamma^0\sigma_3\big)\epsilon=\epsilon$. Then, the supercharges that annihilate the BPS 0-brane is\footnote{Note that the projector satisfies $\Pi_q^\dagger=\Pi_q$ and $\Pi_q^2=\Pi_q$, then we have 
\[
\big\{\cA^A_\alpha,(\cA^B_\beta)^\dagger\big\}=2|Z_q|(\Pi_q)_{\alpha~~C}^{~\gamma A}(\delta_\gamma^{~\rho}\delta^C_{~D}+e^{i\theta_q}(i\Gamma^0)_\gamma^{~\rho}(\sigma_3)^C_{~D})(\Pi_q)_{\rho~~D}^{~\beta B}=0\,,
\]
}
\begin{equation}
    \cA^A_\alpha(q)=(\Pi_q)_{\alpha~~B}^{~\beta A}Q_\beta^B\,.
\end{equation}
Therefore, the preserved supersymmetry subspace associated to this BPS state is
\begin{equation}
    S_q=\big\{c^\alpha_A(\Pi_q)_{\alpha~~B}^{~\beta A}Q^B_\beta\,\vert\, c^\alpha_A\in\mathbb{R}\big\}\,,
\end{equation}
such that again for any $Q\in S_q$, the combination of supercharges annihilates the BPS state $Q\ket{\psi_{\rm BPS}}=0$. The BPS 0-branes are BPS compatible among each other when $\theta_q=\theta_q'$ for any pair of BPS 0-branes with charges $q,q'$. This is equivalent to BPS 0-branes preserving the same supersymmetry.

Next, let us consider the 1-branes in this theory. The supersymmetry algebra in the rest frame of the massive 1-branes stretched along $x^1$ with charge $p$ and tension $T$ is $\{Q_\alpha^A,(Q_\beta^B)^\dagger\}=2(T\delta^{AB}\delta_{\alpha\beta}+(\sigma_1)^{AB}(\Gamma^{01})_{\alpha\beta}Z^1_p)$. Similarly, in the same signature, we have $(\Gamma^{01})^2=1$ and the eigenvalues of $T+\sigma_1\Gamma^{01}Z_p^1$ are $T\pm Z_p^1$. Therefore, the BPS bound on the tension of massive 1-branes is $T\geq |Z^1_p|$. The saturation again only occurs for BPS states. Then, for a BPS 1-brane, the supercharges that annihilate its BPS state is
\begin{equation}
    \cA^A_\alpha=(\Pi_p)_{\alpha~~B}^{~\beta A}Q_\beta^B\,,
\end{equation}
where $\Pi_p=\frac12(1-e^{i\theta_p}\sigma_1\Gamma^{01})$ with $e^{i\theta_p}=\pm 1$. Here, the convention is similar to the 0-brane case where $Z_p^1e^{i\theta_p}=|Z_p^1|$. Therefore, the preserved supersymmetry subspace of the BPS 1-brane with charge $p$ is
\begin{equation}
    S_p=\big\{\tilde{c}^\alpha_A(\Pi_p)_{\alpha~~B}^{~\beta A}Q_\beta^B\,\vert\, \tilde{c}^\alpha_A\in\mathbb{R}\big\}\,.
\end{equation}
Therefore, BPS 1-branes are BPS compatible only when they are extended along the same spatial directions. Additionally, they can either be parallel to each other $\theta_p=\theta_{p'}$, or anti-parallel to each other $\theta_p=\theta_{p'}+\pi$. Only when the BPS 1-branes are parallel to each other do they preserve the same supersymmetry.

Therefore, for either the BPS 0-brane or the BPS 1-brane, the same supersymmetry is preserved up to charge conjugation, namely a choice of BPS branes and anti-BPS branes. 
Let us now consider the intersection between the preserved supersymmetry subspaces associated to BPS 0-branes and BPS 1-branes. The common preserved supersymmetry subspace is 
\begin{equation}
    S_{q\vert p}:=S_q\cap S_p=\left\{\epsilon \in S_{\rm vac}\,\big\vert\, i\sigma_3\Gamma^0\epsilon=-e^{i\theta_q}\epsilon\,,\sigma_1\Gamma^{01}\epsilon=-e^{i\theta_p}\epsilon\right\}\,.
\end{equation}
As $[i\sigma_3\Gamma^0,\sigma_1\Gamma^{01}]=0$, we obtain $\dim S_{q\vert p}=2$. We can see that the above non-vanishing of $S_{q\vert p}$ is independent of the choice of $\theta_q$ and $\theta_p$. However, as these are BPS branes of different dimensions, the last condition imposed by BPS compatibility requires that the BPS compatible 0-branes and 1-branes must be aligned, namely $\theta_q=\theta_p$. Hence, in 5d $\cN=1$ theories, BPS 0-branes and BPS 1-branes are BPS compatible, while anti-BPS 0-branes and anti-BPS 1-branes are BPS compatible.

\subsection{BPS occupancy and the Weak Gravity Conjecture}

For a given charge $q$ associated to a $p$-brane in the charge lattice of a theory, we can examine its associated Hilbert space $\cH_q$ and its BPS subspace $\cH^{\rm BPS}_{q,S}\subseteq \cH(q)$. The BPS Hilbert space can be defined as
\begin{equation}
    \cH^{\rm BPS}_{q,S}:=\big\{\ket{\psi}\in \cH_q\,\big\vert\,f(\eta)_q\ket{\psi}=0\text{ for all }f(\eta)_q\in S\big\}\,,
\end{equation}
where $S$ is the preserved supersymmetry subspace along the worldvolume theory of the $p$-brane of interest with charge $q$.
Then, we arrive at the following definition.
\begin{definition}
\label{def:BPS occupied 1}
    The charge $q$, associated to a BPS-B whose preserved supersymmetry subspace is $S$, within the charge lattice of a given supersymmetric theory is \textit{occupied} by BPS states if 
    \begin{equation}
    \label{eq:BPS occupancy}
        \dim\cH^{\rm BPS}_{q,S}>0\,.
    \end{equation}
\end{definition}
For example, a single D$p$-brane, along with any such number of $N$ coincident parallel D$p$-branes in 10d, has $\dim\cH_N^{\rm BPS}>0$. Therefore, along the ray within the charge lattice $\mathbb{R}[Dp]$, all integral sites are occupied by BPS states. It is worth noting that the above definition is the coarsest notion of BPS occupancy. Namely, it cannot detect if such a BPS configuration forms a bound state. Nevertheless, with the above definition, integral charges occupied by mutually BPS compatible states naturally form a semi-group under addition. 

The above definition necessarily implies that all BPS occupied charges $q$ must satisfy $q\in\cC_{\rm BPS}$. A refinement of the above BPS occupancy can be made when we introduce appropriate BPS indices that can directly compute the BPS occupancy for a given charge. This is, for example, the Gopakumar--Vafa invariants for curves in 5d $\cN=1$ theories~\cite{Gopakumar:1998jq} which we will compute in the examples in \cref{sec:5d examples}. If such a BPS index is defined, then a non-zero BPS index implies BPS occupancy at that given charge. 

What we would like to especially highlight in this section is the utility of the Weak Gravity Conjecture (WGC) in predicting the existence of a BPS state. To start, let us briefly review the statement of the Weak Gravity Conjecture.
\begin{conj*}[Weak Gravity Conjecture~\cite{Arkani-Hamed:2006emk}]
    For any such $U(1)$ gauge theory coupled to a consistent theory of quantum gravity, there must exists a particle of mass $m$ and charge $q\sim \mathcal O(1)$ such that
    \begin{equation}
        \frac{m}{|q|}\leq \frac{M}{|Q|}\bigg\vert_{\rm Ext}\,,
    \end{equation}
    is satisfied. Here, $\frac{M}{|Q|}\big\vert_{\rm Ext}$ denotes the mass-to-charge ratio of an extremal black hole.
\end{conj*}
There is a natural extension of the above WGC to have implications on $p$-branes. In this case, we instead consider $(p+1)$-form gauge fields coupling to quantum gravity. Then, the conjecture states that there must exists a $p$-brane with charge $q$ and tension $T$ such that
\begin{equation}
    \frac{T}{|q|}\leq \frac{T}{|q|}\bigg\vert_{\rm Ext}\,,
\end{equation}
where $\frac{T}{|q|}\big\vert_{\rm Ext}$ is the tension-to-charge ratio of an extremal black $p$-brane.
A sharpened version of this conjecture~\cite{Ooguri:2016pdq} further states that the inequality from the WGC is saturated if and only if the theory is supersymmetric and the particle/state under consideration is BPS.

Another natural lower bound that occurs is due to supersymmetry, namely \cref{eq:BPS bound}. In particular, this bound remains as a moduli-dependent, supersymmetrically protected lower bound on the tension of the $p$-brane in the spectrum. Therefore, from supersymmetry, we also have the global bound on the tension of branes of charge $q$
\begin{equation}
    |Z_{\rm min}(q)|\leq T(q)\,,
\end{equation}
where $Z_{\rm min}(q)$ denotes the global minimum of the central charge across all moduli space.

Combining the above, we arrive at the following natural bound on the tension of any $p$-brane in the spectrum of a quantum gravitational theory
\begin{equation}
\label{eq:BPS predict}
    \big\vert Z_{\rm min}(q)\big\vert\leq T(q)\leq T^{\rm Ext}_{\text{Black Brane}}(q)\,,
\end{equation}
where $T^{\rm Ext}_{\text{Black Brane}}$ denotes the extremal tension of the black brane evaluated at its attractor point.
The saturation above can only occur if $q$ corresponds to a BPS black $p$-brane state in the theory. Therefore, we arrive at a criterion for BPS occupancy at a charge $q$ in the spectrum of theory.
\begin{definition}
\label{def:BPS occupied 2}
    In a quantum gravitational theory, the charge $q$ of a $p$-brane within the spectrum of the theory is occupied by BPS states if
    \begin{equation}
        |Z_{\rm min}(q)|=T^{\rm Ext}_{\rm{Black\, Brane}}(q)\,,
    \end{equation}
\end{definition}

As pointed out in~\cite{Alim:2021vhs}, when the two bounds above agree (which is a nontrivial assumption), the sublattice weak gravity conjecture predicts a tower of BPS states in the charge direction $q$. Reference~\cite{Long:2021lon} discusses interesting examples where the inequalities are strict.  We will see more of these examples in \cref{sec:ECalabi--Yau threefolds} and \cref{sec:M-theory on Calabi--Yau threefold}.

\subsection{The BPS brane cone}
\label{sec:BPS brane}

From the definition of BPS occupancy, for any two given charges $\gamma_1,\gamma_2\in \cS_S$, then we must have $\gamma_1+\gamma_2\in \cS_S$ where $\cS_S$ denotes the set of BPS compatible charges (with respect to the subspace $S$) within the full charge lattice $\cS_{S}\subset \Lambda$ of the theory. Therefore, BPS states (preserving the same supersymmetry) form a semi-group within the lattice of charges.
With this, we can define the BPS brane cone.
\begin{definition}
\label{def:BPS cone}
    For a supersymmetric theory, the positive real span of the set of BPS compatible charges $\cS^{(p)}_S$ for a given $p$-brane appearing in the spectrum of the theory is the \textit{BPS ($p$-)brane cone}
    \begin{equation}
        \cC^{(p)}_{\rm BPS}:=\mathrm{Cone}(\cS^{(p)}_S)=\bigg\{\sum_{i=1}^na_i\gamma_i\bigg\vert a_i\in\mathbb{R}_{\geq 0}\,,\gamma_i\in\cS^{(p)}_S\bigg\}\,.
    \end{equation}
\end{definition}

A peculiar phenomenon emerges among the BPS branes that has partially motivated the \textit{elementary constituent conjecture} in~\cite{Nevoa:2025xiq}. For example, consider BPS branes in type II string theories. The tensions for various D-branes in the Einstein frame are
\begin{equation}
\label{eq:D-brane tensions}
    T_{\mathrm{D}p}=\frac{g_s^{(p-3)/4}}{(2\pi)^p\alpha'{}^{(p+1)/2}}\,.
\end{equation}
and the tension of F1-strings and NS5-branes in Einstein frame are
\begin{equation}
\label{eq:F1 NS5 tension}
    T_{\rm F1}=\frac{g_s^{1/2}}{2\pi\alpha'}\,,\qquad T_{\rm NS5}=\frac{g_s^{-1/2}}{(2\pi)^{5}\alpha'{}^3}\,.
\end{equation}
The moduli space of type IIA string theory is the dilaton $g_s=e^{\phi}$ while the moduli space of type IIB string theory is parameterized by the axio-dilaton $\tau=C_0+ie^{-\phi}$.
For type IIA string theory, all D$p$-branes having vanishing tension in either the strong-coupling limit $g_s\to\infty$ or the weak-coupling limit $g_s\to0$, including the F1-strings and NS5-branes.
In type IIB string theory, again, most of the D$p$-brane tensions vanish in the limit of either $g_s\to 0$ or $g_s\to \infty$ again. However, for D3-branes, the tension is independent of the moduli and remains positive throughout the entire moduli space. 

Therefore, it seems natural that there are two categories of BPS branes within the BPS-B cone: 1) those with minimum positive tension achieved in the interior of the moduli space;\footnote{Note that the conifold points of Calabi--Yau threefolds that can be flopped should be viewed as the interior of moduli space. However, upon inclusion of all cones associated to the birational-family of Calabi--Yau threefolds, the flop curve no longer appears in the $\cC_{\rm BPS-B}$ as the Mori cone becomes restricted as the K\"ahler cone is extended.} and 2) the rest. In particular, the first type is intimately related to extremal black brane solutions with smooth horizons (that are reliable solutions within the leading IR description of the supergravity theory) which we will now turn to focus on.

As it turns out, the classification of the remaining BPS branes can be refined. The rest can have two possibilities: either the minimum tension is on the boundary of the black brane regions and is positive, or the minimum tension goes to zero. More details will be discussed in \cref{sec:tension properties}.

\subsection{The BPS black brane cone}
\label{sec:BPS black brane}

We have defined the BPS brane cones and its definition is independent of gravity. In the presence of gravity, certain BPS branes, when carrying a sufficiently large charge, can create large gravitational backreactions such that the configuration forms black branes. This is exemplified by the D3-brane discussion above. Namely, in the presence of a large $N$ stack of D3-branes, the near-horizon geometry of the $N$ D3-branes becomes $AdS_5\times S^5$ where $R_{AdS_5}=R_{S^5}=(4\pi g_s N \alpha'{}^2)^{1/4}$. These BPS branes in the BPS-B cone are distinguishable from the rest as they admit solutions to the attractor mechanism in supergravity. Hence, we can define the sub-cone of BPS branes within the BPS-B cone as the \textit{BPS black $p$-brane cone}, which we denote as $\cC_{\mathrm{BPS-BB}}$, in the following way.
\begin{definition}
\label{def:BPS black cone}
    In a supersymmetric quantum gravitational theory, the BPS black $p$-brane cone is the minimal real cone that contains all integral charge sites in the BPS-B cone which lead to supersymmetric BPS black $p$-branes that have smooth (near-)horizon geometries.  The boundary of the cone includes points where the attractor moduli is at the boundary, which may either have positive or vanishing tension.
\end{definition}
\noindent In particular, the BPS branes along the boundaries of $\cC_{\rm BPS-BB}$ which have vanishing tension occur when they exist along $\partial \cC_{\rm BPS-BB}\cap \partial \cC_{\rm BPS-B}$.

A unique property of BPS black branes in the interior of $\cC_{\rm BPS-BB}$ is that as we vary the moduli fields in the theory, their minimum tension obeys $T_{\rm min}>0$. Therefore, an equivalent definition of the BPS-BB cone is as follows.
\begin{definition}
\label{def:BPS black cone 2}
    In supersymmetric quantum gravities, the interior of the BPS black $p$-brane cone is the interior of the minimal real cone which contains all integral charge sites where the associated $p$-brane tension has a non-zero minimum in the strict interior of the moduli space.
\end{definition}

With either definition, we can systematically deduce charges that lead to a state with a smooth horizon in supergravity, which we will turn to in the following specific supergravity settings.

\section{Conjectures}
\label{sec:conj}

In this section, we present all our conjectures concerning properties of $\cC_{\rm BPS-B}$ and $\cC_{\rm BPS-BB}$. We propose two conjectures: 1) the BPS completeness in $\cC_{\rm BPS-BB}$ presented in \cref{conj:BPS completeness}, and 2) the electric-magnetic duality between $\cC_{\rm BPS-B}$ and $\cC_{\rm BPS-BB}$, which preserve the same supersymmetry, presented in \cref{conj:cone dual conjecture}. Lastly, we discuss the four classes of BPS branes in $\cC_{\rm BPS-B}$ categorized by the behavior of their tensions as we vary the moduli fields in the theory.

\subsection{BPS completeness in the BPS black brane cone}

The completeness of spectrum hypothesis~\cite{Banks:2010zn} states that in a consistent theory of quantum gravity, all possible representations of the gauge group appear in the spectrum. This is intimately related to the no global symmetry conjecture~\cite{Banks:2010zn,Polchinski:2003bq,Harlow:2018tng} and the triviality of cobordism groups of quantum gravity~\cite{McNamara:2019rup}.
However, no such analog of the completeness hypothesis is found for BPS states in the supersymmetric charge lattices. Furthermore, the naive extrapolated BPS completeness hypothesis for quantum gravity has known counter-examples.
A set of known examples are in 5d $\cN=1$ quantum gravity theories constructed from M-theory compactified on a Calabi--Yau threefold. 
In the BPS 1-brane cone from M5-branes wrapping 4-cycles in a Calabi--Yau threefold, there are many non-BPS sites populating the BPS-B cone as discussed in~\cite{Gendler:2026uux}. Additionally, there are also similar non-BPS holes in the 6d $(1,0)$ quantum gravity theories which have been discussed in~\cite{Kim:2024eoa}.\footnote{We will revisit both of these examples in \cref{sec:5d examples} and \cref{sec:6d examples} in light of our conjectures.}

\begin{figure}
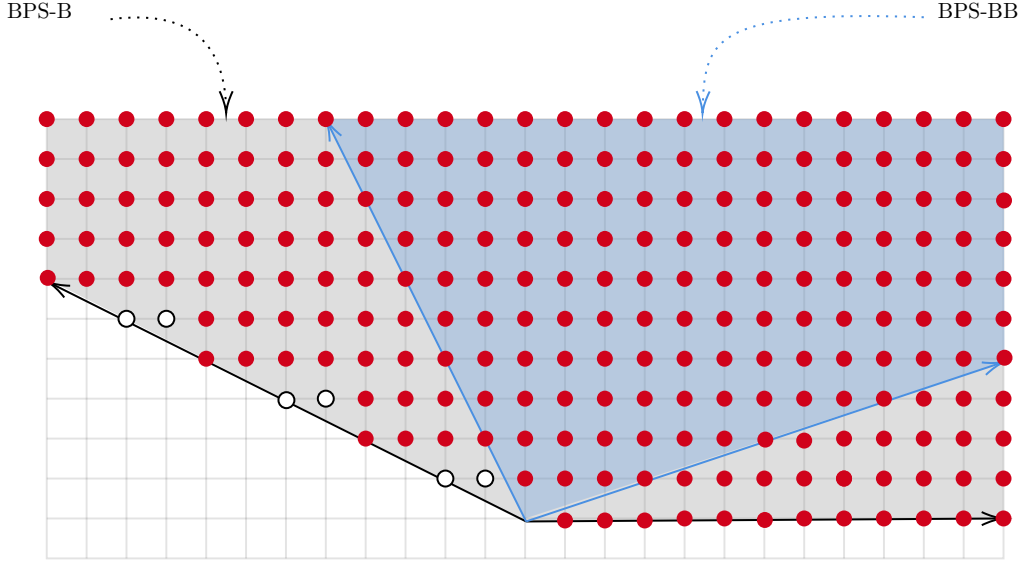

    \centering

\tikzset{every picture/.style={line width=0.75pt}} 

\tikzset{every picture/.style={line width=0.75pt}} 


    \caption{A schematic illustration of the BPS occupancy within $\cC_{\rm BPS-BB}$ and $\cC_{\rm BPS-B}$. Here, the red dots indicates the integral charge site has a BPS state while the white dot indicates an empty site with no BPS states at that given charge in the spectrum of the theory. Additionally, we also illustrate that the BPS states are closed under a BPS algebra.}
    \label{fig:BPS occupancy}
\end{figure}

Motivated by these counter-examples and the necessity of a BPS completeness in quantum gravity, we make the following conjecture concerning the analogue of the completeness hypothesis for BPS states in the BPS-BB cone.
\begin{tcolorbox}
\begin{conj}
\label{conj:BPS completeness}
    In any consistent quantum gravitational theories with supersymmetry, \underline{all} integer sites sitting at the intersection between the BPS-BB cone of the theory and $\mathbb{Z}^{\dim(\cC_{\rm BPS\, BB})}$ are occupied by BPS states, in the sense of \cref{def:BPS occupied 1}.
\end{conj}
\end{tcolorbox}
\noindent The above conjecture is illustrated in \cref{fig:BPS occupancy} where we do not expect BPS completeness in the full $\cC_{\rm BPS-B}$ while we do expect it within $\cC_{\rm BPS-BB}$, including the extremal rays of $\cC_{\rm BPS-BB}$.

\subsection{Duality among BPS cones}
\label{sec:duality between cones}

\begin{figure}
    \centering

\tikzset{every picture/.style={line width=0.75pt}} 

\begin{tikzpicture}[x=0.75pt,y=0.75pt,yscale=-1,xscale=1]

\draw  [fill={rgb, 255:red, 155; green, 155; blue, 155 }  ,fill opacity=0.39 ] (498,119) -- (538,122) -- (556,168) -- (494,277) -- (443,127) -- cycle ;
\draw  [color={rgb, 255:red, 74; green, 144; blue, 226 }  ,draw opacity=1 ][fill={rgb, 255:red, 74; green, 144; blue, 226 }  ,fill opacity=0.11 ] (462,91) -- (549,86) -- (632,105) -- (494,277) -- (376,107) -- cycle ;
\draw  [fill={rgb, 255:red, 155; green, 155; blue, 155 }  ,fill opacity=0.17 ] (186,77) -- (273,145) -- (116,274) -- (39,152) -- (61,81) -- cycle ;
\draw  [color={rgb, 255:red, 0; green, 0; blue, 0 }  ,draw opacity=1 ][fill={rgb, 255:red, 74; green, 144; blue, 226 }  ,fill opacity=0.36 ] (177,102) -- (224,139) -- (116,274) -- (64,127) -- (69,106) -- cycle ;
\draw    (116,274) -- (22.07,124.69) ;
\draw [shift={(21,123)}, rotate = 57.82] [color={rgb, 255:red, 0; green, 0; blue, 0 }  ][line width=0.75]    (10.93,-3.29) .. controls (6.95,-1.4) and (3.31,-0.3) .. (0,0) .. controls (3.31,0.3) and (6.95,1.4) .. (10.93,3.29)   ;
\draw    (116,274) -- (182.88,175.65) ;
\draw [shift={(184,174)}, rotate = 124.22] [color={rgb, 255:red, 0; green, 0; blue, 0 }  ][line width=0.75]    (10.93,-3.29) .. controls (6.95,-1.4) and (3.31,-0.3) .. (0,0) .. controls (3.31,0.3) and (6.95,1.4) .. (10.93,3.29)   ;
\draw    (116,274) -- (297.45,126.26) ;
\draw [shift={(299,125)}, rotate = 140.85] [color={rgb, 255:red, 0; green, 0; blue, 0 }  ][line width=0.75]    (10.93,-3.29) .. controls (6.95,-1.4) and (3.31,-0.3) .. (0,0) .. controls (3.31,0.3) and (6.95,1.4) .. (10.93,3.29)   ;
\draw [color={rgb, 255:red, 0; green, 0; blue, 0 }  ,draw opacity=1 ]   (116,274) -- (199.33,39.88) ;
\draw [shift={(200,38)}, rotate = 109.59] [color={rgb, 255:red, 0; green, 0; blue, 0 }  ,draw opacity=1 ][line width=0.75]    (10.93,-3.29) .. controls (6.95,-1.4) and (3.31,-0.3) .. (0,0) .. controls (3.31,0.3) and (6.95,1.4) .. (10.93,3.29)   ;
\draw [color={rgb, 255:red, 0; green, 0; blue, 0 }  ,draw opacity=1 ]   (116,274) -- (47.55,32.92) ;
\draw [shift={(47,31)}, rotate = 74.15] [color={rgb, 255:red, 0; green, 0; blue, 0 }  ,draw opacity=1 ][line width=0.75]    (10.93,-3.29) .. controls (6.95,-1.4) and (3.31,-0.3) .. (0,0) .. controls (3.31,0.3) and (6.95,1.4) .. (10.93,3.29)   ;
\draw    (39,152) -- (61,81) ;
\draw    (61,81) -- (186,77) ;
\draw    (186,77) -- (273,145) ;
\draw    (273,145) -- (165,203) ;
\draw    (39,152) -- (165,203) ;
\draw [color={rgb, 255:red, 74; green, 144; blue, 226 }  ,draw opacity=1 ]   (116,274) -- (37.67,52.89) ;
\draw [shift={(37,51)}, rotate = 70.49] [color={rgb, 255:red, 74; green, 144; blue, 226 }  ,draw opacity=1 ][line width=0.75]    (10.93,-3.29) .. controls (6.95,-1.4) and (3.31,-0.3) .. (0,0) .. controls (3.31,0.3) and (6.95,1.4) .. (10.93,3.29)   ;
\draw [color={rgb, 255:red, 74; green, 144; blue, 226 }  ,draw opacity=1 ]   (116,274) -- (278.74,72.56) ;
\draw [shift={(280,71)}, rotate = 128.93] [color={rgb, 255:red, 74; green, 144; blue, 226 }  ,draw opacity=1 ][line width=0.75]    (10.93,-3.29) .. controls (6.95,-1.4) and (3.31,-0.3) .. (0,0) .. controls (3.31,0.3) and (6.95,1.4) .. (10.93,3.29)   ;
\draw [color={rgb, 255:red, 74; green, 144; blue, 226 }  ,draw opacity=1 ]   (116,274) -- (126.85,128.99) ;
\draw [shift={(127,127)}, rotate = 94.28] [color={rgb, 255:red, 74; green, 144; blue, 226 }  ,draw opacity=1 ][line width=0.75]    (10.93,-3.29) .. controls (6.95,-1.4) and (3.31,-0.3) .. (0,0) .. controls (3.31,0.3) and (6.95,1.4) .. (10.93,3.29)   ;
\draw [color={rgb, 255:red, 74; green, 144; blue, 226 }  ,draw opacity=1 ]   (64,127) -- (69,106) ;
\draw [color={rgb, 255:red, 74; green, 144; blue, 226 }  ,draw opacity=1 ]   (69,106) -- (177,102) ;
\draw [color={rgb, 255:red, 74; green, 144; blue, 226 }  ,draw opacity=1 ]   (177,102) -- (224,139) ;
\draw [color={rgb, 255:red, 74; green, 144; blue, 226 }  ,draw opacity=1 ]   (224,139) -- (123,174) ;
\draw [color={rgb, 255:red, 74; green, 144; blue, 226 }  ,draw opacity=1 ]   (64,127) -- (123,174) ;
\draw [color={rgb, 255:red, 74; green, 144; blue, 226 }  ,draw opacity=1 ]   (116,274) -- (69.54,107.93) ;
\draw [shift={(69,106)}, rotate = 74.37] [color={rgb, 255:red, 74; green, 144; blue, 226 }  ,draw opacity=1 ][line width=0.75]    (10.93,-3.29) .. controls (6.95,-1.4) and (3.31,-0.3) .. (0,0) .. controls (3.31,0.3) and (6.95,1.4) .. (10.93,3.29)   ;
\draw [color={rgb, 255:red, 74; green, 144; blue, 226 }  ,draw opacity=1 ]   (116,274) -- (176.33,103.88) ;
\draw [shift={(177,102)}, rotate = 109.53] [color={rgb, 255:red, 74; green, 144; blue, 226 }  ,draw opacity=1 ][line width=0.75]    (10.93,-3.29) .. controls (6.95,-1.4) and (3.31,-0.3) .. (0,0) .. controls (3.31,0.3) and (6.95,1.4) .. (10.93,3.29)   ;
\draw [color={rgb, 255:red, 74; green, 144; blue, 226 }  ,draw opacity=1 ]   (494,277) -- (653.75,76.56) ;
\draw [shift={(655,75)}, rotate = 128.56] [color={rgb, 255:red, 74; green, 144; blue, 226 }  ,draw opacity=1 ][line width=0.75]    (10.93,-3.29) .. controls (6.95,-1.4) and (3.31,-0.3) .. (0,0) .. controls (3.31,0.3) and (6.95,1.4) .. (10.93,3.29)   ;
\draw [color={rgb, 255:red, 74; green, 144; blue, 226 }  ,draw opacity=1 ]   (494,277) -- (559.44,52.92) ;
\draw [shift={(560,51)}, rotate = 106.28] [color={rgb, 255:red, 74; green, 144; blue, 226 }  ,draw opacity=1 ][line width=0.75]    (10.93,-3.29) .. controls (6.95,-1.4) and (3.31,-0.3) .. (0,0) .. controls (3.31,0.3) and (6.95,1.4) .. (10.93,3.29)   ;
\draw [color={rgb, 255:red, 74; green, 144; blue, 226 }  ,draw opacity=1 ]   (494,277) -- (456.33,54.97) ;
\draw [shift={(456,53)}, rotate = 80.37] [color={rgb, 255:red, 74; green, 144; blue, 226 }  ,draw opacity=1 ][line width=0.75]    (10.93,-3.29) .. controls (6.95,-1.4) and (3.31,-0.3) .. (0,0) .. controls (3.31,0.3) and (6.95,1.4) .. (10.93,3.29)   ;
\draw [color={rgb, 255:red, 74; green, 144; blue, 226 }  ,draw opacity=1 ]   (494,277) -- (353.14,72.65) ;
\draw [shift={(352,71)}, rotate = 55.42] [color={rgb, 255:red, 74; green, 144; blue, 226 }  ,draw opacity=1 ][line width=0.75]    (10.93,-3.29) .. controls (6.95,-1.4) and (3.31,-0.3) .. (0,0) .. controls (3.31,0.3) and (6.95,1.4) .. (10.93,3.29)   ;
\draw [color={rgb, 255:red, 74; green, 144; blue, 226 }  ,draw opacity=1 ]   (462,91) -- (549,86) ;
\draw [color={rgb, 255:red, 74; green, 144; blue, 226 }  ,draw opacity=1 ]   (376,107) -- (462,91) ;
\draw [color={rgb, 255:red, 74; green, 144; blue, 226 }  ,draw opacity=1 ]   (549,86) -- (632,105) ;
\draw    (494,277) -- (463.44,139.95) ;
\draw [shift={(463,138)}, rotate = 77.43] [color={rgb, 255:red, 0; green, 0; blue, 0 }  ][line width=0.75]    (10.93,-3.29) .. controls (6.95,-1.4) and (3.31,-0.3) .. (0,0) .. controls (3.31,0.3) and (6.95,1.4) .. (10.93,3.29)   ;
\draw    (494,277) -- (422.64,65.89) ;
\draw [shift={(422,64)}, rotate = 71.32] [color={rgb, 255:red, 0; green, 0; blue, 0 }  ][line width=0.75]    (10.93,-3.29) .. controls (6.95,-1.4) and (3.31,-0.3) .. (0,0) .. controls (3.31,0.3) and (6.95,1.4) .. (10.93,3.29)   ;
\draw    (494,277) -- (500.94,60) ;
\draw [shift={(501,58)}, rotate = 91.83] [color={rgb, 255:red, 0; green, 0; blue, 0 }  ][line width=0.75]    (10.93,-3.29) .. controls (6.95,-1.4) and (3.31,-0.3) .. (0,0) .. controls (3.31,0.3) and (6.95,1.4) .. (10.93,3.29)   ;
\draw    (494,277) -- (537.45,123.92) ;
\draw [shift={(538,122)}, rotate = 105.85] [color={rgb, 255:red, 0; green, 0; blue, 0 }  ][line width=0.75]    (10.93,-3.29) .. controls (6.95,-1.4) and (3.31,-0.3) .. (0,0) .. controls (3.31,0.3) and (6.95,1.4) .. (10.93,3.29)   ;
\draw    (494,277) -- (574.01,136.74) ;
\draw [shift={(575,135)}, rotate = 119.7] [color={rgb, 255:red, 0; green, 0; blue, 0 }  ][line width=0.75]    (10.93,-3.29) .. controls (6.95,-1.4) and (3.31,-0.3) .. (0,0) .. controls (3.31,0.3) and (6.95,1.4) .. (10.93,3.29)   ;
\draw    (468,167) -- (443,127) ;
\draw    (443,127) -- (498,119) ;
\draw    (498,119) -- (538,122) ;
\draw    (538,122) -- (556,168) ;
\draw    (468,167) -- (556,168) ;
\draw [color={rgb, 255:red, 74; green, 144; blue, 226 }  ,draw opacity=1 ]   (491,205) -- (632,105) ;
\draw [color={rgb, 255:red, 74; green, 144; blue, 226 }  ,draw opacity=1 ]   (376,107) -- (491,205) ;
\draw [color={rgb, 255:red, 74; green, 144; blue, 226 }  ,draw opacity=1 ]   (494,277) -- (491.06,175) ;
\draw [shift={(491,173)}, rotate = 88.35] [color={rgb, 255:red, 74; green, 144; blue, 226 }  ,draw opacity=1 ][line width=0.75]    (10.93,-3.29) .. controls (6.95,-1.4) and (3.31,-0.3) .. (0,0) .. controls (3.31,0.3) and (6.95,1.4) .. (10.93,3.29)   ;
\draw [color={rgb, 255:red, 74; green, 144; blue, 226 }  ,draw opacity=1 ] [dash pattern={on 0.84pt off 2.51pt}]  (235,18) .. controls (173.31,17.01) and (147.26,46.7) .. (145.03,120.88) ;
\draw [shift={(145,122)}, rotate = 271.53] [fill={rgb, 255:red, 74; green, 144; blue, 226 }  ,fill opacity=1 ][line width=0.08]  [draw opacity=0] (12,-3) -- (0,0) -- (12,3) -- cycle    ;
\draw [color={rgb, 255:red, 0; green, 0; blue, 0 }  ,draw opacity=1 ] [dash pattern={on 0.84pt off 2.51pt}]  (59,14) .. controls (106.76,14) and (134.72,17.96) .. (133.03,91.88) ;
\draw [shift={(133,93)}, rotate = 271.53] [fill={rgb, 255:red, 0; green, 0; blue, 0 }  ,fill opacity=1 ][line width=0.08]  [draw opacity=0] (12,-3) -- (0,0) -- (12,3) -- cycle    ;
\draw [color={rgb, 255:red, 74; green, 144; blue, 226 }  ,draw opacity=1 ] [dash pattern={on 0.84pt off 2.51pt}]  (590,23) .. controls (528.31,22.01) and (517.11,31.9) .. (515.03,105.88) ;
\draw [shift={(515,107)}, rotate = 271.53] [fill={rgb, 255:red, 74; green, 144; blue, 226 }  ,fill opacity=1 ][line width=0.08]  [draw opacity=0] (12,-3) -- (0,0) -- (12,3) -- cycle    ;
\draw [color={rgb, 255:red, 0; green, 0; blue, 0 }  ,draw opacity=1 ] [dash pattern={on 0.84pt off 2.51pt}]  (403,38) .. controls (454.74,38.99) and (485.69,62.76) .. (490.92,138.85) ;
\draw [shift={(491,140)}, rotate = 266.28] [fill={rgb, 255:red, 0; green, 0; blue, 0 }  ,fill opacity=1 ][line width=0.08]  [draw opacity=0] (12,-3) -- (0,0) -- (12,3) -- cycle    ;

\draw (58,31.4) node [anchor=north west][inner sep=0.75pt]  [xscale=0.75,yscale=0.75]  {$e_{1}$};
\draw (209,26.4) node [anchor=north west][inner sep=0.75pt]  [xscale=0.75,yscale=0.75]  {$e_{2}$};
\draw (301,128.4) node [anchor=north west][inner sep=0.75pt]  [xscale=0.75,yscale=0.75]  {$e_{3}$};
\draw (182,189.4) node [anchor=north west][inner sep=0.75pt]  [xscale=0.75,yscale=0.75]  {$e_{4}$};
\draw (15,146.4) node [anchor=north west][inner sep=0.75pt]  [xscale=0.75,yscale=0.75]  {$e_{5}$};
\draw (21,54.4) node [anchor=north west][inner sep=0.75pt]  [color={rgb, 255:red, 74; green, 144; blue, 226 }  ,opacity=1 ,xscale=0.75,yscale=0.75]  {$\tilde{e}_{1}$};
\draw (70,83.4) node [anchor=north west][inner sep=0.75pt]  [color={rgb, 255:red, 74; green, 144; blue, 226 }  ,opacity=1 ,xscale=0.75,yscale=0.75]  {$\tilde{e}_{2}$};
\draw (185,83.4) node [anchor=north west][inner sep=0.75pt]  [color={rgb, 255:red, 74; green, 144; blue, 226 }  ,opacity=1 ,xscale=0.75,yscale=0.75]  {$\tilde{e}_{3}$};
\draw (282,74.4) node [anchor=north west][inner sep=0.75pt]  [color={rgb, 255:red, 74; green, 144; blue, 226 }  ,opacity=1 ,xscale=0.75,yscale=0.75]  {$\tilde{e}_{4}$};
\draw (129,130.4) node [anchor=north west][inner sep=0.75pt]  [color={rgb, 255:red, 74; green, 144; blue, 226 }  ,opacity=1 ,xscale=0.75,yscale=0.75]  {$\tilde{e}_{5}$};
\draw (341,92.4) node [anchor=north west][inner sep=0.75pt]  [color={rgb, 255:red, 74; green, 144; blue, 226 }  ,opacity=1 ,xscale=0.75,yscale=0.75]  {$m_{1}$};
\draw (436,69.4) node [anchor=north west][inner sep=0.75pt]  [color={rgb, 255:red, 74; green, 144; blue, 226 }  ,opacity=1 ,xscale=0.75,yscale=0.75]  {$m_{2}$};
\draw (561,65.4) node [anchor=north west][inner sep=0.75pt]  [color={rgb, 255:red, 74; green, 144; blue, 226 }  ,opacity=1 ,xscale=0.75,yscale=0.75]  {$m_{3}$};
\draw (634,101.4) node [anchor=north west][inner sep=0.75pt]  [color={rgb, 255:red, 74; green, 144; blue, 226 }  ,opacity=1 ,xscale=0.75,yscale=0.75]  {$m_{4}$};
\draw (498,178.4) node [anchor=north west][inner sep=0.75pt]  [color={rgb, 255:red, 74; green, 144; blue, 226 }  ,opacity=1 ,xscale=0.75,yscale=0.75]  {$m_{5}$};
\draw (400,66.4) node [anchor=north west][inner sep=0.75pt]  [color={rgb, 255:red, 0; green, 0; blue, 0 }  ,opacity=1 ,xscale=0.75,yscale=0.75]  {$\tilde{m}_{1}$};
\draw (474,54.4) node [anchor=north west][inner sep=0.75pt]  [color={rgb, 255:red, 0; green, 0; blue, 0 }  ,opacity=1 ,xscale=0.75,yscale=0.75]  {$\tilde{m}_{2}$};
\draw (509,122.4) node [anchor=north west][inner sep=0.75pt]  [color={rgb, 255:red, 0; green, 0; blue, 0 }  ,opacity=1 ,xscale=0.75,yscale=0.75]  {$\tilde{m}_{3}$};
\draw (564,151.4) node [anchor=north west][inner sep=0.75pt]  [color={rgb, 255:red, 0; green, 0; blue, 0 }  ,opacity=1 ,xscale=0.75,yscale=0.75]  {$\tilde{m}_{4}$};
\draw (472,143.4) node [anchor=north west][inner sep=0.75pt]  [color={rgb, 255:red, 0; green, 0; blue, 0 }  ,opacity=1 ,xscale=0.75,yscale=0.75]  {$\tilde{m}_{5}$};
\draw (241,9) node [anchor=north west][inner sep=0.75pt]  [xscale=0.75,yscale=0.75] [align=left] {\textcolor[rgb]{0.29,0.56,0.89}{BPS-BB}};
\draw (6,5) node [anchor=north west][inner sep=0.75pt]  [color={rgb, 255:red, 0; green, 0; blue, 0 }  ,opacity=1 ,xscale=0.75,yscale=0.75] [align=left] {BPS-B};
\draw (595,17) node [anchor=north west][inner sep=0.75pt]  [xscale=0.75,yscale=0.75] [align=left] {\textcolor[rgb]{0.29,0.56,0.89}{BPS-B}};
\draw (152,256) node [anchor=north west][inner sep=0.75pt]  [xscale=0.75,yscale=0.75] [align=left] {Electric cone};
\draw (537,256) node [anchor=north west][inner sep=0.75pt]  [xscale=0.75,yscale=0.75] [align=left] {Magnetic cone};
\draw (350,33) node [anchor=north west][inner sep=0.75pt]  [color={rgb, 255:red, 0; green, 0; blue, 0 }  ,opacity=1 ,xscale=0.75,yscale=0.75] [align=left] {BPS-BB};

\end{tikzpicture}
    \caption{A schematic illustration of the duality between cones. Here, we have depicted the extremal rays of the relevant cones. In particular, between the electric and magnetic BPS cones, the blue cones are dual to each other, while the gray cones are dual to each other as well all via the electric/magnetic duality. Namely, $e_i.\tilde{m}_j\geq 0$ and $\tilde{e}_1.m_j\geq 0$.
    }
    \label{fig:dual cones}
\end{figure}
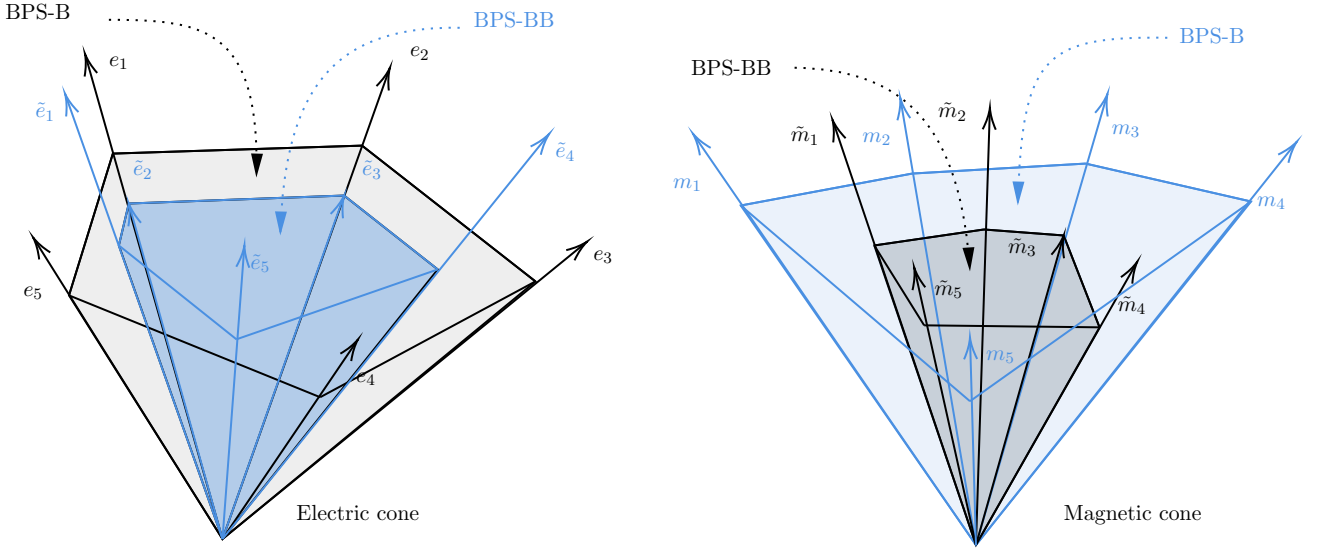

In the supergravity description, the BPS-BB cones are naturally conical hulls of the corresponding scalar moduli spaces by the attractor mechanism. For example, in 6d $\cN=(1,0)$ supergravity theories, the BPS black 1-brane cone is the cone over the tensor branch cone. In 5d $\cN=1$ supergravity theories, the BPS black 0-brane cone is the cone over the vector-multiplet moduli space written in the dual attractor coordinates (also referred to as dual coordinates, consequently the cone of dual coordinates), while the BPS black 1-brane cone is the cone over the same moduli space written in the original very-special coordinates~\cite{Alim:2021vhs}.

From the supersymmetry algebra, if a BPS-BB with charge $q$ and an electric/magnetic dual BPS brane with charge $p$ are BPS compatible, the mass of brane $p$ evaluated at the attractor point of $q$ is given by the neutral electric-magnetic pairing $\langle q,p\rangle$ with a positive overall constant depending on $q$. This in particular implies that $\langle q,p\rangle\geq 0$. In other words, the natural pairing on $\cC_{\rm BPS-B}$ and $\cC_{\rm BPS-BB}$ is non-negative. From this, we know that $\cC_{\rm BPS-B}\subseteq \cC_{\rm BPS-BB}^\vee$ and $\cC_{\rm BPS-BB}\subseteq \cC_{\rm BPS-B}^\vee$. The converse direction comes from the physical interpretation of boundaries of moduli space.  The vector multiplet/tensor multiplet moduli space can end only because some BPS sector becomes massless or tensionless. Hence, the walls (codimension-one boundaries) of $\cC_{\rm BPS-BB}$ are exactly vertices (generators) of $\cC_{\rm BPS-B}$. This implies that these two cones are exactly dual to each other. Importantly, the duality is not concerned of the size of the extremal black branes. Hence, we arrive at our second conjecture.
\begin{tcolorbox}
\begin{conj}
\label{conj:cone dual conjecture}
    In any consistent $d$-dimensional supersymmetric quantum gravitational theory with $d\geq 4$, the BPS ($p$-)brane cone is dual to the BPS black ($(d-4-p)$-)brane cone preserving the same supersymmetry when these notions do not depend on the choice of moduli,
    \begin{equation}
        \cC^{(p)}_{\rm BPS-B}=\big(\cC^{(d-4-p)}_{\rm BPS-BB}\big)^{\vee}\,,
    \end{equation}
    and the BPS black ($p$-)brane cone is dual to the BPS ($(d-4-p)$-)brane cone
    \begin{equation}
        \cC^{(p)}_{\rm BPS-BB}=\big(\cC^{(d-4-p)}_{\rm BPS-B}\big)^{\vee}\,.
    \end{equation}
\end{conj}
\end{tcolorbox}
In the present work, we verify this conjecture in theories where preserving the same supersymmetry is automatic (up to charge conjugation), similar to what we have seen for the 5d $\cN=1$ supersymmetry discussed in \cref{sec:BPS compatibility}. Namely, we examine this duality either among BPS-Bs and BPS-BBs or among $\overline{\rm BPS}$-Bs and $\overline{\rm BPS}$-BBs. However, as the latter duality is necessarily implies by the former, in these theories, it is sufficient to examine the cone duals among BPS compatible BPS-Bs.
Furthermore, we restrict to theories and BPS-Bs of which the $\cC_{\rm BPS-B}$ defined among BPS-Bs preserving the same supersymmetry is moduli independent.


\subsection{Landscape of BPS branes in quantum gravity}
\label{sec:tension properties}

With the above examples, we observe four distinct behaviors of the tension of BPS branes across the moduli space of quantum gravitational theories that are related to the definitions of $\cC_{\rm BPS-B}$ (\cref{def:BPS cone}) and $\cC_{\rm BPS-BB}$ (more relevantly \cref{def:BPS black cone 2}).
\begin{enumerate}
    \item When $q\in \cC_{\rm BPS-BB}\backslash \partial \cC_{\rm BPS-BB}$, the tension of the BPS brane has a minimum in the interior of the moduli space.
    \item When $q\in \partial \cC_{\rm BPS-BB}$, the tension of the BPS brane is strictly positive in the interior of the moduli space and it is minimized along the boundary of the moduli space. Additionally, this minimum tension is a stable minimum as $|\nabla T(q)|=0$ where $\nabla$ is taken with respect to the metric on the moduli space. If $q\notin \partial \cC_{\rm BPS-B}$, the minimum tension is strictly positive.
    \item When $q\in \cC_{\rm BPS-B}\backslash \cC_{\rm BPS-BB}$ and $q\notin \partial \cC_{\rm BPS-B}\cup \partial \cC_{\rm BPS-BB}$, the tension of the BPS brane minimizes along the boundary of the moduli space, which we denote the location of the minimum tension as $\phi_{\rm min}=\mathrm{argmin}_\phi T(q;\phi)\in \partial \cM$. However, the stability of this minimum tension depends the behavior of the tensions for BPS branes on the boundary of $\cC_{\rm BPS-B}$. More explicitly, the minimum value of $T(q;\phi_{\rm min})$ is determined by the tension of a BPS-B with charge $q'\in \partial{\cC_{\rm BPS-BB}}$ evaluated along the same boundary of the moduli space where $T(q')$ also is minimized, namely $T(q;\phi_{\rm min})=T(q';\phi_{\rm min})\equiv T_{\rm min}(q')$. On the other hand, the gradient of the tension of this BPS brane is controlled by a BPS-B with charge $q''\in \partial \cC_{\rm BPS-B}$ whose tension vanishes along the same boundary of moduli space where $T(q)$ is minimized, namely when $\lim_{\phi\to\phi_{\rm min}}T(q'';\phi)=0$, then $\nabla T(q,\phi_{\rm min})=\nabla T(q',\phi_{\rm min})$.
    \item When $q\in \partial\cC_{\rm BPS-B}$, the tension for these BPS branes remains positive in the interior of moduli space and vanishes in the limit as we approach a boundary of moduli space.
\end{enumerate}

The first and third class of charge states are mutually exclusive, whereas the second and fourth class of charge states need not be. The boundaries of moduli spaces can be categorized into infinite-distance boundaries and finite-distance boundaries. By the distance conjecture~\cite{Ooguri:2006in}, when the BPS branes in the fourth class vanish along an infinite-distance boundary, their tensions scale as $\sim e^{-\alpha \Delta}$ where $\Delta$ denotes the traversed geodesic distance in moduli space and $\alpha\sim \cO(1)$. On the other hand, the BPS branes in the fourth class whose tension vanishes along a finite-distance boundary generally does not exhibit the behavior of $\nabla_{\phi\to\partial \cM}\nabla T\to 0$.
Therefore, a BPS brane existing on both the boundary of $\cC_{\rm BPS-BB}$ and on the boundary of $\cC_{\rm BPS-B}$ must have a minimum tension along an infinite-distance boundary of moduli space.

For BPS $p$-branes, the attractor solution exists if the tension of such a BPS $p$-brane has a minimum in moduli space. Therefore, why does the second and third class differ from BPS black brane solutions? The answer lies in the fact that all such minima for the second and third classes of BPS branes are located on the boundary of field space. As we have seen from the distance conjecture, when we approach the infinite-distance boundaries of field space, we expect to encounter a tower of light states emerging. This signals the breakdown of the EFT description of quantum gravity. Hence, the attractor solution is no longer reliable in these limits. 
We may also encounter finite-distance boundaries, for example those in F-theory and M-theory compactified on Calabi--Yau manifolds~\cite{Witten:1996qb}. In these cases, we similarly encounter massless degrees of freedom appearing along these finite-distance boundaries such as the flop wall, the CFT wall, or the $SU(2)$ wall in 5d $\cN=1$ quantum gravities~\cite{Alim:2021vhs,Gendler:2022ztv}. Hence, the EFT description of quantum gravity similarly breaks down leading to the conclusion that the attractor solution obtained at these values in the field space are no longer reliable as well. Nevertheless, from a top-down perspective, we can still compute the tensions of BPS branes by knowing the volume of the minimum-volume representatives of a calibrated cycle class as a function of the appropriate parameters parametrizing the moduli space of the lower-dimensional quantum gravity theory.

\section{Quantum gravity with 32 supercharges}
\label{sec:QG with 32}

In this section, we consider our conjectures in 11d M-theory and 10d type II string theories, of which the bulk spacetime preserves 32 supercharges. In 11d M-theory, the BPS branes we consider are the stack of coincident M2-branes or M5-branes. The supergravity solution corresponding to a large number of such BPS branes in spacetime lead to an extremal black brane solution with a smooth horizon of finite size. This is similar to D3-branes in the type IIB string theory. Thus, these BPS branes are the simplest examples to test our conjectures in \cref{sec:conj}, which will be discussed in this section.

\subsection{11d M-theory}
\label{sec:M-theory}

Let us first consider 11d M-theory. The supersymmetry algebra for maximal supersymmetry in 11d is
\begin{equation}
    \{Q_\alpha,Q_\beta\}=(C\Gamma^\mu)_{\alpha\beta}P^\mu+\frac12(C\Gamma^{\mu_1\mu_2})_{\alpha\beta} Z_{\mu_1\mu_2}+\frac{1}{5!}(C\Gamma^{\mu_1\dots\mu_5})_{\alpha\beta}Z_{\mu_1\dots\mu_5}\,,
\end{equation}
where we have two central extensions via a 2-form central charge $Z_{\mu_1\mu_2}$ and a 5-form central charge $Z_{\mu_1\dots\mu_5}$. These correspond to charges carried by the M2-branes and M5-branes in M-theory, respectively. 
To illustrate our definitions and conjectures, we will focus on stacks of coincident M2-branes, or M5-branes.
The preserved supersymmetry of the configurations of isolated M2-branes extending in the (012) directions and of isolated M5-branes extending in the (012345) directions separately satisfy
\begin{equation}
\label{eq:M2 M5 spinors}
    \Gamma^{012}\epsilon=\epsilon\,,\qquad \Gamma^{012345}\epsilon=\epsilon\,,
\end{equation}
where $\epsilon$ denotes the 32-component real Majorana spinor.
Hence, these correspond to $1/2$-BPS configurations in the bulk theory. In this case, the BPS brane cone of M2-branes and M5-branes are simply
\begin{equation}
    \cC_{\rm BPS-B}^{(\rm M2)}=\mathbb{R}_+\,,\qquad \cC_{\rm BPS-B}^{(\rm M5)}=\mathbb{R}_+\,,
\end{equation}
where the integral charges $N\in \cC_{\rm BPS-B}\cap \mathbb{Z}$ correspond to the number of such parallel BPS branes in spacetime. Additionally, there is a corresponding anti-BPS brane cone associated to the M2-branes and M5-branes where \cref{eq:M2 M5 spinors} become $\Gamma\epsilon=-\epsilon$. In this case, the anti-BPS brane cone of M2-branes and M5-branes are $\cC_{{\overline{\rm BPS}\rm -B}}^{(\rm M2)}=\mathbb{R}_-$ and $\cC_{{\overline{\rm BPS}\rm -B}}^{(\rm M5)}=\mathbb{R}_-$. The negativity of the cones for anti-BPS branes arises due to anti-BPS branes having opposite orientation to the BPS branes.

In 11d M-theory, there are no moduli. The tensions of an M2-brane and an M5-brane in Einstein frame are
\begin{equation}
    T_{\rm M2}=2^{1/3}\pi^{2/3}
    M_{\pl}^3\,,\qquad 
    T_{\rm M5}=\bigg(\frac{\pi}{2}\bigg)^{1/3}
    M_{\pl}^6\,,
\end{equation}
where $M_\pl$ denotes the 11d Planck mass.\footnote{Here, we use $M_\pl^9=(2\kappa_{11}^{2})^{-1}$ with $(2\kappa_{11}^2)=(2\pi)^8 l_\pl^9$.} The tension of $N$ parallel M2-branes and M5-branes are $NT_{\rm M2}$ and $NT_{\rm M5}$, respectively. Therefore, from \cref{def:BPS black cone 2}, we can identify all such choices of $N$ as the BPS black brane cone.

The above definition coincides with the BPS black cone identified via \cref{def:BPS black cone}. To see this, we can also study the extremal solutions of a stack of coincident M2-branes or M5-branes in 11d maximal supergravity. To start, for a stack of coincident $N$ M2-branes extending in the $(012)$ directions in spacetime, the extremal solution in Einstein frame is
\begin{equation}
\label{eq:M2brane metric}
    \dd s^2=\bigg(1+\frac{L^6}{r^6}\bigg)^{-2/3}\big(-\dd t^2+\dd y_1^2+\dd y_2^2\big)+\bigg(1+\frac{L^6}{r^6}\bigg)^{1/3}\big(\dd r^2+r^2\dd \Omega_7^2\big)\,,
\end{equation}
where $L^6=2^5\pi^2l_{11}^6 N$. The warp factors can be deduced by solving for the harmonic solution of $\nabla^2_{\mathbb{R}^8} H=0$. In this case, as the curvature is inversely proportional to $N$, as we take the large $N$ limit, we can reliably trust the above ans\"atz. In the near-horizon region $r\ll L$, the geometry becomes the familiar $AdS_4\times S^7$ where 
\begin{equation}
    2R_{AdS_4}=R_{S^7}=L\,.
\end{equation}
Hence, the BPS black brane cone of M2-branes is simply the ray
\begin{equation}
    \cC_{\rm BPS-BB}^{\rm (M2)}=\mathbb{R}_+\,.
\end{equation}

Similarly, for a stack of $N$ M5-branes, the metric becomes
\begin{equation}
    \dd s^2=\bigg(1+\frac{L^3}{r^3}\bigg)^{-1/3}\big(-\dd t^2+\dd\vec{y}_{5}^2\big)+\bigg(1+\frac{L^3}{r^3}\bigg)^{2/3}\big(\dd r^2+r^2\dd \Omega_4^2\big)\,,
\end{equation}
where $L^3=\pi N l_{11}^3$.
The curvatures become small in the limit of large $N$, and hence the solution is reliable. In the near-horizon limit where $r\ll L$, the geometry becomes $AdS_7\times S^4$ with
\begin{equation}
    R_{AdS_7}=2R_{S^4}=L\,.
\end{equation} 
Therefore, the BPS black brane cone of M5-branes is again
\begin{equation}
    \cC_{\rm BPS-BB}^{(\rm M5)}=\mathbb{R}_+\,.
\end{equation}

As emphasized in \cref{sec:BPS compatibility}, in the present setting, there is a choice of $\mathbb{Z}_2$ supersymmetry of which the M-branes preserve. This is associated to their charge orientation and corresponds to either a BPS M-brane or an anti-BPS M-brane. We can see that by the top-down construction of the BPS black brane solutions in $\cC_{\rm BPS-BB}$, regardless of which supersymmetry the BPS black brane preserves, at all such integral sites, namely $\cC_{\rm BPS-BB}\cap \mathbb{Z}$, there is a BPS state in the spectrum at that given charge. Hence, \cref{conj:BPS completeness} remains satisfied across all BPS black brane cones. 
The completeness of BPS states in $\cC_{\rm BPS-BB}$ is also naturally anticipated from the holographic correspondence where for each positive integer $N$, the corresponding near-horizon geometries $AdS_4\times S^7$ and $AdS_7\times S^4$ have dual descriptions in terms of the worldvolume theories of $N$ coincident M2-branes and M5-branes, respectively.

However, when discussing the duality among BPS cones, we can see that the choice of supersymmetry preserved along the BPS brane worldvolumes becomes relevant. Namely, the electric-magnetic duality between cones in \cref{conj:cone dual conjecture} is indeed only satisfied among cones preserving the same supersymmetry. In other words, among the cones for BPS branes, we have the following duality
\begin{equation}
    \cC_{\rm BPS-BB}^{(\rm M2)}=\big(\cC_{\rm BPS-B}^{(\rm M5)}\big)^\vee\,,\qquad \cC_{\rm BPS-B}^{(\rm M2)}=\big(\cC_{\rm BPS-BB}^{(\rm M5)}\big)^\vee\,.
\end{equation}
And identically for the anti-BPS branes, the duality becomes $\cC_{\overline{\rm BPS}\rm -BB}^{(\rm M2)}=\big(\cC_{\overline{\rm BPS}\rm -B}^{(\rm M5)}\big)^\vee$ and $\cC_{\overline{\rm BPS}\rm -B}^{(\rm M2)}=\big(\cC_{\overline{\rm BPS}\rm -BB}^{(\rm M5)}\big)^\vee$. Among the BPS compatible pairings, the electric-magnetic duality among these cones is realized as the self-duality of the corresponding ray, namely $\mathbb{R}_+$ in the BPS sector and $\mathbb{R}_-$ in the anti-BPS sector.

\subsection{10d Type II string theories}

Next, let us consider 10d type II string theories. In this case, we again only focus on the simplest case of a stack of coincident BPS branes of the same type in the theory, namely only of one type among all D$p$-branes, F1-strings, or NS5-branes. This is sufficient in illustrating our conjectures. 
Hence, the relevant BPS brane cones are
\begin{equation}
    \cC_{\rm BPS-B}^{(\rm F1)}=\cC_{\rm BPS-B}^{(\rm NS5)}=\cC_{\rm BPS-B}^{(\mathrm{D}p)}=\mathbb{R}_+\,,\qquad \cC_{\overline{\rm BPS}\rm -B}^{(\rm F1)}=\cC_{\overline{\rm BPS}\rm -B}^{(\rm NS5)}=\cC_{\overline{\rm BPS}\rm -B}^{(\mathrm{D}p)}=\mathbb{R}_-\,.
\end{equation}
Here, $\cC_{\rm BPS-B}^{(\mathrm{D}p)}$ denotes the BPS brane cone associated to a specific D$p$-brane with positively oriented charge, as opposed to $\cC_{\overline{\rm BPS}\rm -B}^{({\rm D}p)}$ which denotes the cone of $\overline{\mathrm{D}p}$-branes.
The dilaton $\phi$ couples to the relevant gauge field strengths, to which the D$p$-branes couple to electrically, as $e^{a_p\phi}F^2_{p+2}$ where $a_p=\frac{3-p}{2}$ in the 10d type II supergravity action in Einstein frame.
The single-charge extremal solutions sourced by F1-strings, NS5-branes, and all D$p$-branes, except the special case of D3-branes, all lead to singular black brane solutions where radius of the transverse sphere measured in Einstein frame scales as $R_{S^{8-p}}\sim r^{\frac{(p-3)^2}{16}}$ in the near-horizon limit and vanishes at the would-be horizon $r=0$.\footnote{In string frame, the near-horizon geometry of a large stack of parallel NS5-branes develops a linear dilaton profile. However, in Einstein frame, such an extremal solution nevertheless has a singular horizon. This is related to the divergence of the dilaton at the origin in the string frame.} For D3-branes, the 5-form field strength appears in the action without a coupling to the dilaton. Therefore, the effective black D3-brane potential is independent of the moduli. This would imply that the attractor mechanism $\partial_\phi V_3=0$ is trivially satisfied where $V_3\sim N^2$ for a stack of $N$ D3-branes. This is also related to the tension of D3-branes being moduli independent, which we have discussed previously in \cref{sec:BPS brane}.
More explicitly, to see the smooth horizon of D3-branes in type IIB supergravity, the extremal solution of a stack of parallel $N$ D3-branes extending in the $(0123)$ directions of spacetime in Einstein frame
\begin{equation}
    \dd s_{\rm D3}^2=\bigg(1+\frac{L^4}{r^4}\bigg)^{-1/2}\big(-\dd t^2+\dd y_1^2+\dd y_2^2+\dd y_3^2\big)+\bigg(1+\frac{L^4}{r^4}\bigg)^{1/2}\big(\dd r^2+r^2 \dd\Omega_5\big)\,,
\end{equation}
where $L^4=4\pi g_s N \alpha'{}^2$ for a stack of $N$ D3-branes. 
In the near horizon limit $r\ll L$, 
we encounter the familiar $AdS_5\times S^5$ geometry with
\begin{equation}
    R_{AdS_5}=R_{S^5}=L\,.
\end{equation}
Hence, from the supergravity perspective, a sufficiently large stack of D3-branes leads to an extremal black brane with a smooth horizon of finite size.
Hence, the BPS black brane cone for D3-branes is
\begin{equation}
    \cC_{\rm BPS-BB}^{(\rm D3)}=\mathbb{R}_+\,.
\end{equation}
Similar to the discussion above for M-branes in \cref{sec:M-theory}, there are also equivalent extremal solutions for the $\overline{\rm D3}$-branes where the BPS black brane cone is $\cC_{\overline{\rm BPS}\rm -BB}^{(D3)}=\mathbb{R}_-$. 
Again, from the top-down perspective of type IIB string theory, all such integral $N$ sitting at the intersection between $\cC_{\rm BPS-BB}^{(\rm D3)}\cap \mathbb{Z}$ are occupied by a BPS state consisting of $N$ parallel D3-branes extending along the $(0123)$ directions in spacetime. (The same is true for $\cC_{\overline{\rm BPS}\rm-BB}\cap \mathbb{Z}$.) This is similarly anticipated from the holographic correspondence. Hence, again, \cref{conj:BPS completeness} is satisfied for D3-branes in type IIB string theory.

Similar to the electric-magnetic pairing between M2-branes and M5-branes, the supersymmetry preserved along D3-branes is a $\mathbb{Z}_2$ choice of orientation, corresponding to a choice of positive or negative charge. As D3-branes are self-dual via electric-magnetic duality, the cone duality in \cref{conj:cone dual conjecture} is between the relevant BPS cones preserving the same supersymmetry. Namely, the cone duality is
\begin{equation}
    \cC_{\rm BPS-B}^{(\rm D3)}=\big(\cC_{\rm BPS-BB}^{(\rm D3)}\big)^\vee\,,\qquad \cC_{\overline{\rm BPS}\rm -B}^{(\rm D3)}=\big(\cC_{\overline{\rm BPS}\rm -BB}^{(\rm D3)}\big)^\vee\,.
\end{equation}

\section{F-theory on elliptic Calabi--Yau threefolds}
\label{sec:ECalabi--Yau threefolds}

In this section, we provide evidence for our conjectures for the specific case of F-theory on elliptic Calabi--Yau threefolds, leading to ${\cal N}=(1,0)$ supsersymmetric theories in $d=6$. In this geometric setup, we prove both our \cref{conj:BPS completeness} and \cref{conj:cone dual conjecture}. Lastly, we study a variety of base surfaces to illustrate the intricate behaviors of BPS branes.

To start, a general 6d $\cN=(1,0)$ supergravity contains one gravity multiplet, $n_V$ vector multiplets, $n_H$ hypermultiplets, and $n_T$ tensor multiplets. However, these are not independent values as the cancellation of the gravitational anomaly in the theory requires $n_H-n_V+29n_T=273$, with additional constraints arising from the gauge and mixed anomaly cancellation.
The scalar fields in the tensor multiplets can be organized into a vector $J^\alpha\in\mathbb{R}^{1,n_T}$ with $J^0>0$ and $\frac12J\cdot J=1$. They parameterize the tensor moduli space of the theory
\begin{equation}
    \cM_T=SO(1,n_T)/SO(n_T)\,.
\end{equation}
The metric on this moduli space, restricted to $\frac12J.J=1$, is
\begin{equation}
    g_{\alpha\beta}=2J_\alpha J_\beta-\Omega_{\alpha\beta}\,,
\end{equation}
where $\Omega_{\alpha\beta}$ is an $SO(1,n_T)$-invariant bilinear form and $J_\alpha=\Omega_{\alpha\beta}J^\beta$.
The theory also contains two-form gauge potentials $B^\alpha$. Thus, in this theory, we have both electric and magnetic strings that couple to $B^\alpha$. For a BPS string of charge $q^\alpha$, its tension is $T\sim q^\alpha J_\alpha$.

When the 6d $\cN=(1,0)$ theory is constructed from F-theory compactified on an elliptic Calabi--Yau threefold, the above quantities can be derived from the geometry. Let $\pi:Y\to B$ be an elliptic Calabi--Yau threefold with smooth base surface $B$.  Then, the vector $J^\alpha$ is can be identified with the K\"ahler form $J$ of $B$. Choosing a basis of divisor classes $[D_\alpha]\in H^2(B,\mathbb{Z})$, we have
\begin{equation}
\label{eq:kahler form 6d}
    J=\sum_{\alpha=1}^{h^{1,1}(B)}t^\alpha[D_\alpha]\,,
\end{equation}
where $t^\alpha$ are the K\"ahler parameters. Here, $[D]$ denotes the cohomology class associated to a divisor $D$. When the K\"ahler parameters take values within the K\"ahler cone of the base surface, $t^\alpha\in K(B)$, the volume of all effective divisors and curves must be non-negative. Using the K\"ahler form above, the volume of the base surface can be computed as
\begin{equation}
\label{eq:6d constant slice}
    \cV=\frac12 C_{\alpha\beta}t^\alpha t^\beta\,,
\end{equation}
where $C_{\alpha\beta}=\int_B[D_\alpha]\wedge[D_\beta]$ is the intersection form on divisor classes in the base. Then, the tensor moduli space of the quantum gravitational theory is identified with the constant-volume slice ($\cV=1$) within $K(B)$.

\subsection{BPS strings and the Mori cone}

The spectrum of the 6d $\cN=(1,0)$ theory consists of D3-branes wrapping curves within $B$. Hence, the six-dimensional tensor/string charge lattice is realized geometrically by $H_2(B,\Z)$, or equivalently by divisor classes on $B$ under Poincar\'e duality. Within $H_2(B,\mathbb{R})$, the natural BPS string cone is the effective cone of curves, namely the Mori cone,
\begin{equation}
    \cC_{\rm BPS-B}:=M(B)\subset N_1(B)\,,
\end{equation}
where $N_1(B)$ is the real vector space of curve classes in $B$ modulo numerical equivalence.
The intersection form on the base gives the natural lattice pairing. Namely, suppose $q_{C},q_{C'}\in \cC_{\rm BPS-B}\cap H_2(B,\Z)$, where $q_C$ and $q_{C'}$ are the charges associated with the strings from D3-branes wrapping the curves $C$ and $C'$, respectively, then we have
\begin{equation}
    q_C\cdot q_{C'}=C.C'\in \Z\,,
\end{equation}
where $C.C'=\int_B [C]\wedge [C']$ is the intersection between the curves $C,C'$ within $B$.
In 6d, as the strings can either couple to the two-form gauge potential electrically or magnetically, the electric and magnetic charge lattice of strings are self-dual, $\Lambda_e=\Lambda_m$. 
This implies that by the electric-magnetic duality, the natural six-dimensional consistency requires the full charge lattice of strings to be self-dual (with respect to the Dirac pairing). 

The tension of a D3-brane wrapping a curve $C\in H_2(B,\Z)$ can be computed from geometry as
\begin{equation}
    T=\volume(C)\geq \int_{[C]} J=q_at^a=:|Z|\,,
\end{equation}
where the tension is related to the volume of the curves it wraps and $q_a$ denotes the quantized charges, namely the decomposition of $[C]=q_a[C_a]$ into the basis of curve classes dual to the basis of divisor classes used in \cref{eq:kahler form 6d}. When $[C]$ is an effective curve class, the resulting string in 6d is supersymmetric and the above inequality is saturated as we have discussed in \cref{sec:BPS compatibility}. Then, as previously stated, the constraints on the K\"ahler cone arise from imposing all such BPS strings in 6d must have non-negative tensions everywhere in $K(B)$.

\subsection{Black strings and the nef cone}
 
The physical tensor branch is controlled by the K\"ahler cone, and the closure of the K\"ahler cone is the nef cone $\Nef(B)$.  With the tensor moduli space and the BPS spectrum identified, we can now study the corresponding attractor mechanism for BPS black strings in these theories. For a string with charge $q$, the relevant ans\"atz is
\begin{equation}
    \dd s_6^2=H(r)^{-1}\big(-\dd t^2+\dd y^2\big)+H(r)\big(\dd r^2+r^2\dd\Omega_3^2\big)\,,
\end{equation}
where the solution is 
\begin{equation}
    K^\alpha(r):=H(r)J^\alpha(r)=J_\infty^\alpha+\frac{q^\alpha}{r^2}\,,
\end{equation}
with $J_\infty^\alpha$ denoting the asymptotic value of the tensor modulus. Restricting to the constant-volume slice, the harmonic function $H(r)$ is
\begin{equation}
    H(r)=\sqrt{\frac12\Omega_{\alpha\beta}K^\alpha(r)K^\beta(r)}\,.
\end{equation}
The scalar flow follows from the harmonic function
\begin{equation}
    J^\alpha(r)=\frac{J_\infty^\alpha+q^\alpha/r^2}{\sqrt{\frac12\Omega_{\gamma\beta}\left(J^\gamma_\infty+q^\gamma/r^2\right)\left(J_\infty^\beta+q^\beta/r^2\right)}}\,.
\end{equation}
In the near-horizon limit $r\to 0$, the charge term $q^\alpha/r^2$ dominates, leading to the attractor value
\begin{equation}
\label{eq:6d attractor point}
    J_*^\alpha=\frac{q^\alpha}{\sqrt{\frac12\Omega_{\gamma\beta}q^\gamma q^\beta}}\,.
\end{equation}
Thus, a necessary condition for a black string with charge $q$ to have a smooth horizon is that the charges must satisfy $q\cdot q>0$.
Equivalently, we can introduce the corresponding black string potential for a black string with charge $q$ as
\begin{equation}
    V_{\rm BS}=g_{\alpha\beta}q^\alpha q^\beta=2Z^2-q^2\,,
\end{equation}
where $Z=q^\alpha J_\alpha$.
The BPS attractor equation can be similarly obtained by extremizing $Z$ along the tensor moduli space
\begin{equation}
    \partial_i Z=0\,,
\end{equation}
where $\partial_i$ denotes the derivatives tangent to the constant-volume slice. 
For the solution to describe a black string with a smooth horizon, the attractor point must lie within the K\"ahler cone. 
Evaluating the central charge at \cref{eq:6d attractor point}, or equivalently, evaluating $V_{\rm BS}$ at the BPS attractor point, givse the tension of the extremal BPS black string
\begin{equation}
\label{eq:6d attractor solution}
    T^{\rm Ext}\sim \sqrt{2q\cdot q}\,.
\end{equation}

From this perspective, a regular black-string charge is naturally represented by a class compatible with the tensor-branch positivity cone.  From the geometric perspective, this is encoded by effective nef (equivalently semi-ample, in the standard base situations of interest) divisor classes. Therefore, we identify the BPS-BB cone with the nef cone
\begin{equation}
    \cC_{\rm BPS-BB}:=\mathrm{Nef}(B)\,.
\end{equation}

The natural cone duality is
\begin{equation}
    M(B)^\vee=\Nef(B)\,,
\end{equation}
which is the usual Kleiman duality for a projective surface. This is also the geometric expression of the duality between BPS-Bs and BPS-BBs in 6d.

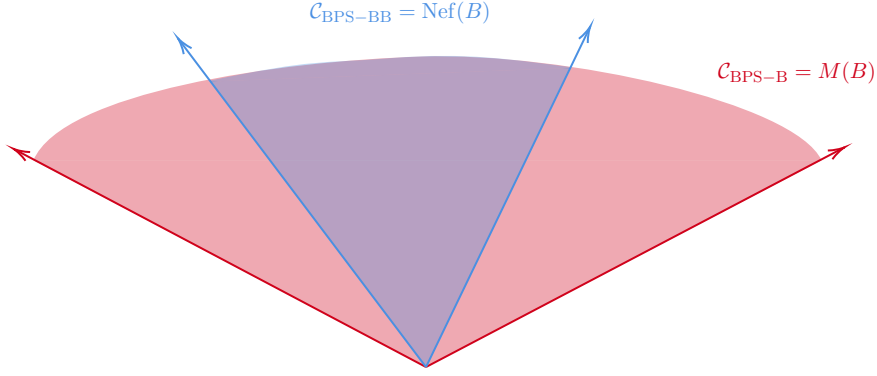
\begin{figure}
    
\centering
\tikzset{every picture/.style={line width=0.75pt}} 

\begin{tikzpicture}[x=0.75pt,y=0.75pt,yscale=-0.75,xscale=0.75]

\draw [color={rgb, 255:red, 208; green, 2; blue, 27 }  ,draw opacity=1 ]   (329,275) -- (54.77,129.94) ;
\draw [shift={(53,129)}, rotate = 27.88] [color={rgb, 255:red, 208; green, 2; blue, 27 }  ,draw opacity=1 ][line width=0.75]    (10.93,-3.29) .. controls (6.95,-1.4) and (3.31,-0.3) .. (0,0) .. controls (3.31,0.3) and (6.95,1.4) .. (10.93,3.29)   ;
\draw [color={rgb, 255:red, 208; green, 2; blue, 27 }  ,draw opacity=1 ]   (329,275) -- (607.23,128.93) ;
\draw [shift={(609,128)}, rotate = 152.3] [color={rgb, 255:red, 208; green, 2; blue, 27 }  ,draw opacity=1 ][line width=0.75]    (10.93,-3.29) .. controls (6.95,-1.4) and (3.31,-0.3) .. (0,0) .. controls (3.31,0.3) and (6.95,1.4) .. (10.93,3.29)   ;
\draw  [color={rgb, 255:red, 208; green, 2; blue, 27 }  ,draw opacity=0 ][fill={rgb, 255:red, 208; green, 2; blue, 27 }  ,fill opacity=0.34 ] (593,137) -- (329,275) -- (67,137) -- cycle ;
\draw [color={rgb, 255:red, 208; green, 2; blue, 27 }  ,draw opacity=0.34 ][fill={rgb, 255:red, 208; green, 2; blue, 27 }  ,fill opacity=0.34 ]   (67,137) .. controls (86,92) and (230,71) .. (326,68) .. controls (422,65) and (575,99) .. (593,137) ;
\draw [color={rgb, 255:red, 74; green, 144; blue, 226 }  ,draw opacity=1 ]   (329,275) -- (164.2,56.6) ;
\draw [shift={(163,55)}, rotate = 52.96] [color={rgb, 255:red, 74; green, 144; blue, 226 }  ,draw opacity=1 ][line width=0.75]    (10.93,-3.29) .. controls (6.95,-1.4) and (3.31,-0.3) .. (0,0) .. controls (3.31,0.3) and (6.95,1.4) .. (10.93,3.29)   ;
\draw [color={rgb, 255:red, 74; green, 144; blue, 226 }  ,draw opacity=1 ][fill={rgb, 255:red, 74; green, 144; blue, 226 }  ,fill opacity=1 ]   (329,275) -- (438.13,48.8) ;
\draw [shift={(439,47)}, rotate = 115.76] [color={rgb, 255:red, 74; green, 144; blue, 226 }  ,draw opacity=1 ][line width=0.75]    (10.93,-3.29) .. controls (6.95,-1.4) and (3.31,-0.3) .. (0,0) .. controls (3.31,0.3) and (6.95,1.4) .. (10.93,3.29)   ;
\draw  [color={rgb, 255:red, 74; green, 144; blue, 226 }  ,draw opacity=0 ][fill={rgb, 255:red, 74; green, 144; blue, 226 }  ,fill opacity=0.34 ] (329,275) -- (185,84) -- (425,75) -- cycle ;
\draw [color={rgb, 255:red, 74; green, 144; blue, 226 }  ,draw opacity=0.34 ][fill={rgb, 255:red, 74; green, 144; blue, 226 }  ,fill opacity=0.34 ]   (185,84) .. controls (257,69) and (256,73) .. (306,69) .. controls (356,65) and (387,70) .. (425,75) ;

\draw (250,29.4) node [anchor=north west][inner sep=0.75pt]  [xscale=0.75,yscale=0.75]  {$\textcolor[rgb]{0.29,0.56,0.89}{\cC_{\rm BPS-BB}=\mathrm{Nef}( B)}$};
\draw (523,69.4) node [anchor=north west][inner sep=0.75pt]  [color={rgb, 255:red, 208; green, 2; blue, 27 }  ,opacity=1 ,xscale=0.75,yscale=0.75]  {$\cC_{\rm BPS-B}=M(B)$};

\end{tikzpicture}
    \caption{A schematic illustration of the BPS-B and BPS-BB cones for BPS 1-branes in F-theory compactified on elliptic Calabi--Yau threefolds with a base $B$.}
    \label{fig:Cones in F-theory on ECalabi--Yau threefold}
\end{figure}

\begin{proposition}[F-theory cone dictionary]
In the 6d $(1,0)$ quantum gravity theory constructed as F-theory compactified on an elliptically fibered Calabi--Yau threefold with base $B$, the following are true:
\begin{enumerate}
\item the BPS-B cone is the effective cone of curves $M(B)$;
\item the BPS-BB cone (which is also the tensor cone) is the nef cone $\Nef(B)$;
\item these cones are dual under the intersection pairing on the self-dual string charge lattice.
\end{enumerate}
\end{proposition}

\noindent This identification is illustrated in \cref{fig:Cones in F-theory on ECalabi--Yau threefold} and is one of the cleanest instances in which the physics and the birational geometry match exactly.

\subsection{The worldsheet theory and effectiveness}
\label{sec:proof in ECalabi--Yau threefold}

In this section, we will review the worldsheet theory of BPS strings and also prove \cref{conj:BPS completeness} both in the context of 6d $(1,0)$ quantum gravitational theories constructed from F-theory compactified on an elliptic Calabi--Yau threefold. In F-theory compactified on an elliptic Calabi--Yau threefold with base $B$, a D3-brane wrapping a holomorphic curve $C\subset B$ gives rise to a string in 6d whose worldsheet theory is a 2d $\cN=(0,4)$ theory. Let $M_6$ denote the 6d spacetime and $\Sigma\subset M_6$ the 2d string worldsheet, then the embedding $\Sigma\hookrightarrow M_6$ has a rank-four normal bundle $N\to \Sigma$. The transverse rotation group is $SO(4)_N\cong SU(2)_L\times SU(2)_R$. For later convenience, we denote the second Chern classes associated to these two $SU(2)$ bundles by $c_2(L):=c_2(SU(2)_L)$ and $c_2(R):=c_2(SU(2)_R)$. With this convention, the Euler class and the first Pontryagin class of the normal bundle are $\chi_4(N)=c_2(R)-c_2(L)$ and $p_1(N)=-2(c_2(L)+c_2(R))$. Since the restriction of the six-dimensional tangent bundle to the string worldsheet decomposes as $TM_6\vert_\Sigma=T\Sigma\oplus N$, we have $p_1(TM_6)\vert_\Sigma=p_1(T\Sigma)+p_1(N)$. Lastly, the left-moving $(c_L)$ and right-moving $(c_R)$ central charges of the 2d worldsheet theory can extracted from the 4-form anomaly polynomial obtained by anomaly inflow from the 6d bulk theory. This takes on the following form
\begin{equation}
\label{eq:2d anomaly polynomial}
    X_4=\frac{c_R-c_L}{24}p_1(T\Sigma)+k_Rc_2(R)-k_Lc_2(L)+\dots\,,
\end{equation}
where $k_R,k_L$ denote the $SU(2)_R$ and $SU(2)_L$ current algebra levels, respectively, and $\dots$ denote additional flavor and mixed anomalies of the theory which will not be needed below. These quantities will enter the Cardy formula of black holes in 5d upon a circle reduction of these strings in 6d.

Let us study the anomaly inflow onto the worldsheet theory of such a string in the 6d bulk theory. To start, there is a Green--Schwarz term in the 6d action of the form
\begin{equation}
    S_{6}\supset \int_{M_6} \Omega_{\alpha\beta}B^\alpha\wedge X_4^\beta\,,
\end{equation}
where $B^\alpha$ are the chiral two-form potentials and $X_4^\alpha\supset -\frac14 a^\alpha p_1(TM_6)$ with the 6d gravitational anomaly vector as $a=-K_B$ with $-K_B=c_1(B)$ denoting the anti-canonical class on $B$. In the presence of a string, the Bianchi identity is modified to $\dd H^\alpha=2\pi q^\alpha \delta_4(\Sigma)$. Then, the relevant terms in the anomaly inflow along the string worldsheet becomes
\begin{align}
    X_4&=\frac12 q^2\chi_4(N)-\frac14 (a\cdot q)p_1(TM_6)\vert_{\Sigma}\,,\nonumber\\
    &=\frac12 q^2(c_2(R)-c_2(L))-\frac14 (a\cdot q)p_1(T\Sigma)+\frac12 (a\cdot q)(c_2(R)+c_2(L))\,,
\end{align}
where $q^2:=\Omega_{\alpha\beta}q^\alpha q^\beta$ and $a\cdot q:=\Omega_{\alpha\beta}a^\alpha q^\beta$.
After adding the center-of-mass contribution to the anomaly polynomial, we obtain the following expressions for the central charges
\begin{align}
    c_R&=3q^2+3(a\cdot q)+6=3C.C+3c_1(B).C+6\,,\\
    c_L&=3q^2+9(a\cdot q)+6=3C.C+9c_1(B).C+6\,,
\end{align}
where the latter equalities are expressed in terms of geometric data on the base $B$. One can similarly extract the current algebra levels and obtain
\begin{equation}
    k_L=1+\frac12 (q^2-a\cdot q)=1+\frac12(C.C-c_1(B).C)=g(C)\,,
\end{equation}
where the last equality is the adjunction formula with $g(C)$ denoting the genus of $C$. 

Similar to~\cite{Maldacena:1997de}, we can consider the entropy of the BPS particle obtained from wrapping this BPS string in 6d on a compact circle down to 5d. In this case, the BPS string wrapping the spatial circle can carry quantized KK momentum $P_y=n/R_{\rm KK}$. Then, in the Cardy regime where the KK momentum satisfies $n\gg c_L$, the entropy is
\begin{equation}
    S=2\pi \sqrt{\frac{c_Ln}{6}}\,.
\end{equation}
For a smooth base $B$ of an elliptic Calabi--Yau threefold, we have $\chi_h(B,\cO_B)=1$. Thus, the Hirzebruch--Riemann--Roch theorem~\cite{Hirzebruch1966} gives the holomorphic Euler characteristic of the line bundle $\cO_B(C)$ of a curve in the base surface, $C\subset B$, as
\begin{align}
\label{eq:holomorphic euler 6d}
    \chi_h(B,\cO_B(C))&=h^0(B,\cO_B(C))-h^1(B,\cO_B(C))+h^2(B,\cO_B(C))\,,\\
    &=1+\frac12(C.C+c_1(B).C)\,,
\end{align}
Therefore, the above entropy formula for the 5d BPS particle can be rewritten as
\begin{equation}
    S=2\pi \sqrt{n\big[\chi_h(B,\cO_B(C))+c_1(B).C\big]}\,.
\end{equation}
This shows that $\chi_h(B,\cO_B(C))$ is relevant to the microscopic degeneracy of the wrapped string. Therefore, from the entropy, we have $\chi_h(B,\cO_B(C))+c_1(B).C\geq 0$. This is not yet sufficient to show $C$ must be effective.
However, the lower bound on $\chi_h(B,\cO_B(C))$ can be further improved by considering additional geometric ingredients. When $B$ is a smooth base of an elliptic Calabi--Yau threefold, multiples of the anti-canonical class of $B$ have an effective representation, e.g., $[\Delta]=-12K_B=12c_1(B)$ where $\Delta$ denotes the discriminant of the associated Weierstrass model. In the above expression, we can see that when $C$ is a nef curve, we obtain the relation $c_1(B).C\geq 0$.\footnote{This inequality does not rely on $c_1(B)$ itself being an integral effective divisor class. For example, this inequality is true when $B$ is an Enriques surfaces. However, $c_1(B)$ is not an effective divisor class in $B$ as we will discuss in \cref{sec:Enriques}.} Therefore, we must have $\chi_h(B,\cO_B(C))>0$ for a nef curve as $C.C\geq 0$ by definition, alongside $c_1(B).C\geq 0$. 
When $C$ is ample, by Serre duality, we have $h^2(B,\cO_B(C))=h^0(B,\cO_B(K_B-C))$. As $(K_B-C).C=-c_1(B).C-C.C<0$, $(K_B-C)$ cannot be effective. Therefore, we arrive at $h^2(B,\cO_B(C))=0$.
Combining all the ingredients above, we are led to the following theorem.
\begin{theorem}
\label{theorem:6d effectiveness}
    For any smooth base $B$ of an elliptic Calabi--Yau threefold, all integral ample divisor classes in $B$ are effective.
\end{theorem}
\begin{proof}
    By the positivity of $\chi_h(B,\cO_B(C))$ for nef divisors, the vanishing of $h^2(B,\cO_B(C))$ for ample divisors, and the non-negativity of $h^1(B,\cO_B(C))$, we have $h^0(B,\cO_B(C))>0$ when $C$ is ample.
\end{proof}
As the ingredients of this proof are general, there is an analogous proof for effectiveness of integral ample divisor classes in Calabi--Yau threefolds, relevant to 5d $\cN=1$ theories. This has been stated in e.g.,~\cite{Katz:2020ewz}, which we will review in \cref{sec:proof in Calabi--Yau threefold}.

\subsection{Examples}
\label{sec:6d examples}

While the above proofs and arguments have been presented in full generality for any smooth elliptic Calabi--Yau threefold, it is instructive to see all of this at play in a series of explicit and well-known toric bases. 

The following examples are motivated by a series of questions and curiosities.
\begin{itemize}
    \item What is the simplest example of the BPS-B cone and the BPS-BB cone in a 6d $\cN=(1,0)$ theory? From a geometric perspective, the simplest theories are those obtained from F-theory on elliptic Calabi--Yau threefolds where the base is a Hirzebruch surface. In this case, the cones are two-dimensional while the K\"ahler moduli space (restricted to the constant-volume slice) is one-dimensional. The case of Hirzebruch surfaces is presented in \cref{sec:Hirzebruch}. 
    \item Since the Hirzebruch surfaces, with the exception of $\mathbb F_1$, arise naturally as minimal models of rational surfaces, a relevant question is how do the BPS-B cone and BPS-BB cone behave under birational transformations of the base manifold? In these cases, all of our conjectures remain true while the boundary of the tensor moduli space becomes increasingly complicated. This is discussed in \cref{sec:blowups}.
    \item Another class of minimal models is the del Pezzo surfaces. In these base manifolds, the dimension of the cones also can become large. Nevertheless, the study of these manifolds can be carried out straightforwardly similar to the previous cases and are presented in \cref{sec:dPr}.
    \item There is a family of surfaces at $h^{1,1}=10$ which cannot be obtained from blowups of Hirzebruch surfaces or del Pezzo surfaces, namely the Enriques surfaces. These are discussed in \cref{sec:Enriques}. There is a potential subtlety regarding seemingly non-effective divisors in the effective cone. Nevertheless, all integer sites in the BPS-B cone are occupied by BPS states.  
    \item Within the above examples, BPS completeness is satisfied in the BPS-B cone. However, this is not a general phenomenon in 6d $(1,0)$ quantum gravity theories. Hence, as a last example, we discuss the geometry in e.g.,~\cite[App. D]{Kim:2024eoa} to illustrate non-BPS holes in $\cC_{\rm BPS-B}\backslash\cC_{\rm BPS-BB}$ in F-theory compactified on elliptic Calabi--Yau threefolds.
\end{itemize}

\subsubsection{Hirzebruch surfaces}
\label{sec:Hirzebruch}

The simplest toric varieties that one can consider appearing as the base manifold in elliptic Calabi--Yau threefolds are the Hirzebruch surfaces $\mathbb{F}_n$ with $n=0,1,\dots,8$ and $n=12$. The Mori cone of $\mathbb{F}_n$ is generated by the negative section $H$ and the fibers $F$ with their intersection data being
\begin{equation}
\label{eq:Fn intersection}
    H.H=-n\,,\qquad F.H=H.F=1\,,\qquad F.F=0\,.
\end{equation}
The Mori cone is simplicial and is generated by these two effective curve
classes.
Namely, BPS strings in F-theory on a base $\mathbb{F}_n$ arise from
D3-branes wrapping either the negative section $H$ or the fiber $F$
\begin{equation}
    \cC_{\rm BPS-B}=M(\mathbb{F}_n)=\Cone(H,F)\,.
\end{equation}
Equivalently, since $\mathbb{F}_n$ is a surface, we may also discuss the
effective cone of divisor classes. In the basis $(H,F)$, the effective
cone is again generated by the two classes
\begin{equation}
\label{eq:hirzebruch effective cone}
    \cE(\mathbb{F}_n)=\Cone(H,F)\,.
\end{equation}

The canonical class of $\mathbb{F}_n$ is
\begin{equation}
    K_{\mathbb{F}_n}=-2H-(n+2)F\,.
\end{equation}
Its intersections with the Mori cone generators are
\begin{equation}
    (-K_{\mathbb{F}_n}).F=2\,,\qquad
    (-K_{\mathbb{F}_n}).H=2-n\,.
\end{equation}
Hence, for $n=2$, the anti-canonical class is nef but not ample and for $n>2$, the anti-canonical class is no longer nef.

From these, we can construct the K\"ahler form as
\begin{equation}
    J=h\tilde{H}+f\tilde{F}\,,
\end{equation}
where we have introduced the nef basis $\tilde{F}\sim F$ and $\tilde{H}\sim H+nF$. The K\"ahler cone in this basis is then simply
\begin{equation}
    K(\mathbb{F}_n)=\big\{(h,f)\in \mathbb{R}_+^2\big\}\,,
\end{equation}
which is deduced by demanding the positivity of curves in the Mori cone. The volume of the base $B$ can then be written as
\begin{equation}
    \cF=2hf+nh^2\,.
\end{equation}
Then, we can introduce the canonical normalized scalar $\Delta$ as $h(\Delta)=e^{-\Delta}$ and $f(\Delta)=\frac12(e^\Delta-ne^{-\Delta})$ where the K\"ahler cone constraint along this codimension-one slice is $\Delta\geq \log[n]/2$.
In the basis of divisors used to define the effective cone in \cref{eq:hirzebruch effective cone} where the K\"ahler form is expressed as $J=hH+fF$, the K\"ahler cone is $K(\mathbb{F}_n)=\{(h,f)\in \mathbb{R}_+^2\,\vert\, h>0,f>nh\}$ which is more straightforward to observe that $K(\mathbb{F}_n)\subset \cE(\mathbb{F}_n)$.

Let us study the black string solutions in these theories. To begin, the central charge associated to a string from wrapping the divisor $q_h\tilde{H}+q_F\tilde{F}$ is
\begin{equation}
\label{eq:Fn Z}
    Z=\frac{q_h}{2h}+\bigg(q_f+\frac{nq_h}{2}\bigg)h\,,
\end{equation}
when restricted to the constant volume slice $\cF=1$. Then, the black string potential takes on the form
\begin{equation}
    V_{q}=2Z^2-q^2=2\bigg[\frac{q_h}{2h}+\bigg(q_f+\frac{nq_h}{2}\bigg)h\bigg]^2-(nq_h^2+2q_hq_f)\,.
\end{equation}
Therefore, attractor solutions exist when 
\begin{equation}
    \partial_hV_q=4Z\partial_hZ=0\,.
\end{equation}
The two branches of solutions are $Z=0$ and $\partial_h Z=0$. These each correspond to the non-BPS and BPS attractor solutions.
Restricting to BPS solutions, we find the attractor points as
\begin{equation}
    h_*=\frac{q_h}{\sqrt{nq_h^2+2q_hq_f}}\,,\qquad f_*=\frac{q_f}{\sqrt{nq_h^2+2q_hq_f}}\,.
\end{equation}
Tor such a solution to be valid, we must ensure $(h_*,f_*)\in K(\mathbb{F}_n)$. Therefore, the BPS black string solutions exist for $(q_h,q_f)\in \mathbb{R}_+^2$. These are precisely the ample divisor classes on $\mathbb{F}_n$. Hence, the black string cone for these theories is
\begin{equation}
    \cC_{\rm BPS-BB}=K(\mathbb{F}_n)\,,
\end{equation}
where the closure of $\cC$ includes the supergravity strings~\cite{Katz:2020ewz} and is hence identified with $\cC_{\rm BPS-BB}=M(\mathbb{F}_n)$.

Let us consider a BPS string obtained from a D3-brane wrapping the effective curve $C=xH+yF\in M(\mathbb{F}_n)$ with $x,y\geq 0$. For a black string charge $Q=q_h(H+nF)+q_fF$, by the intersection data \cref{eq:Fn intersection}, we have $Q.H=q_f$ and $Q.F=q_H$. Therefore, the pairing between this black string and the BPS string is
\begin{equation}
    Q.C=xq_f+yq_h\geq 0\,.
\end{equation}
Therefore, indeed from the attractor mechanism along with positivity of tension, we can see that the black cone is dual to the BPS cone. 
Geometrically, this is the dual between the nef cone of divisors and the Mori cone of curves. Similarly, the cone of movable curves is dual to the effective cone of divisors.

\begin{figure}[!tp]
    \centering
    \begin{subfigure}{.475\textwidth}
        \includegraphics[width=\linewidth]{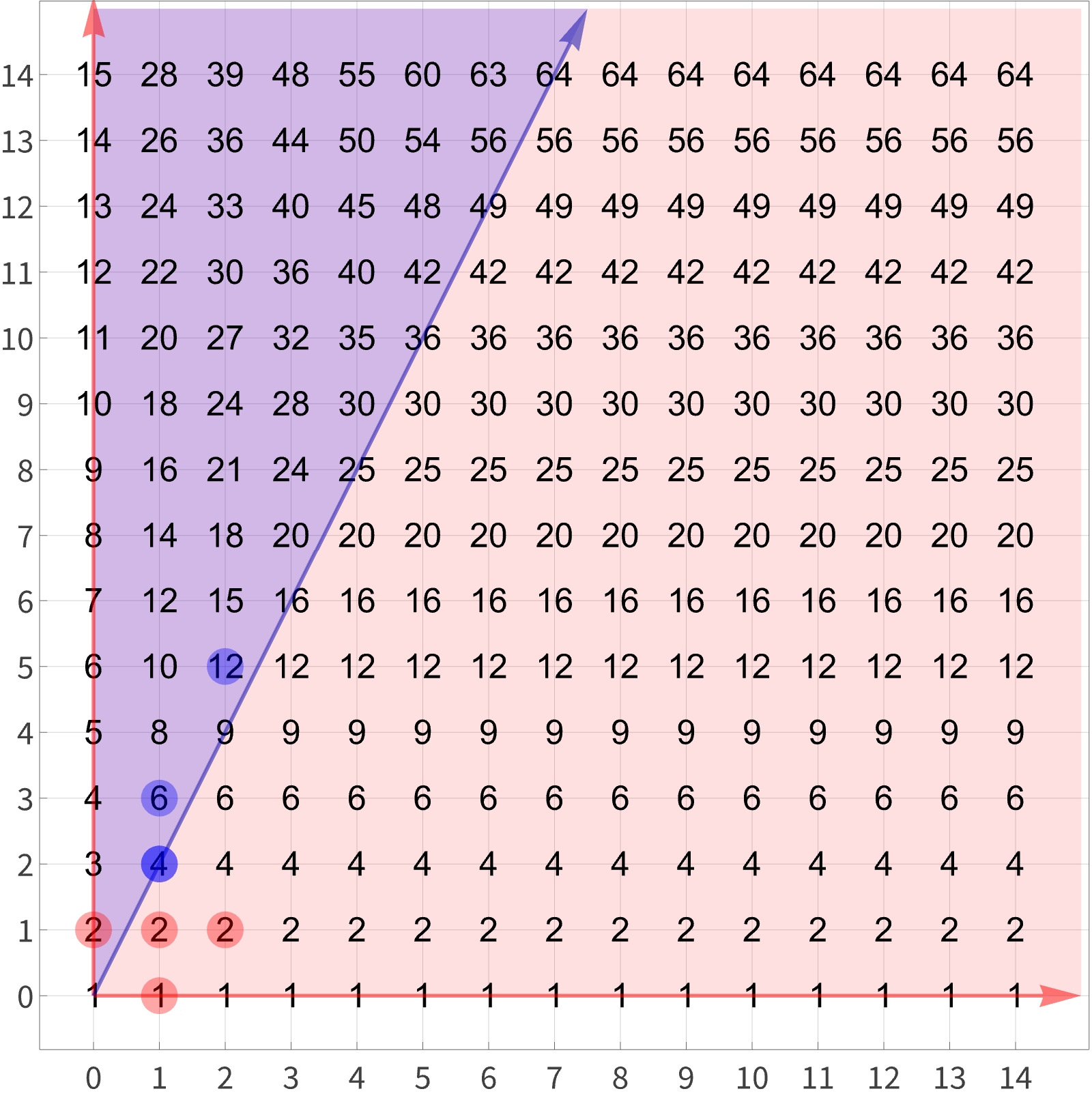}
        \begin{picture}(0,0)\vspace*{-1.2cm}
            \put(-20,140){$[F]$}
            \put(120,0){$[H]$}
            \put(100,250){$[H]+2[F]$}
        \end{picture}
        \vspace*{-0.25cm}
        \caption{$\cC_{\rm BPS-B}$ and $\cC_{\rm BPS-BB}$}
        \label{fig:F2cones}
    \end{subfigure}
    \hfill
    \begin{subfigure}{.475\textwidth}
        \includegraphics[width=\linewidth]{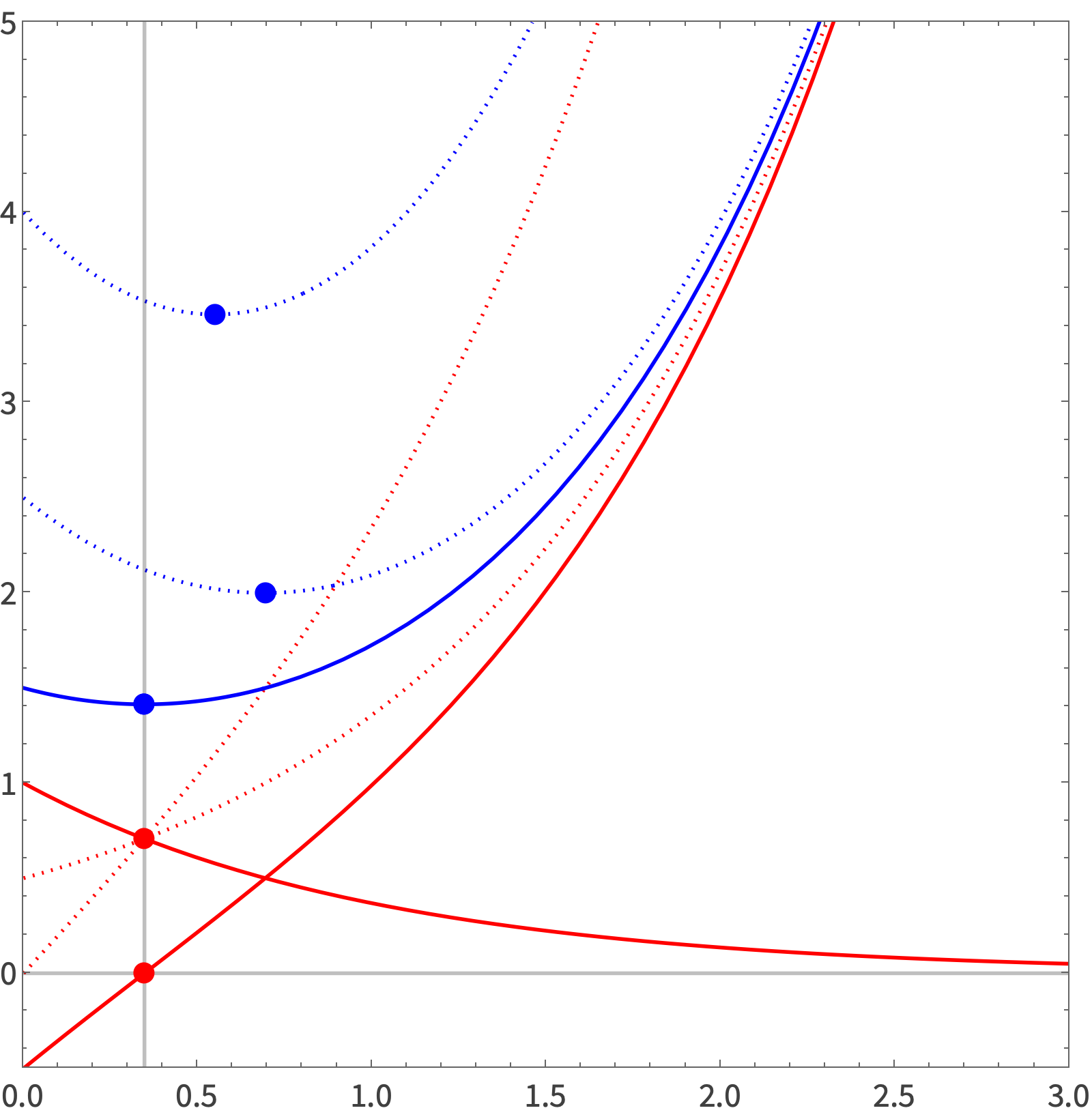}
        \begin{picture}(0,0)\vspace*{-1.2cm}
            \put(-15,140){$T$}
            \put(120,0){$\Delta$}
            \put(16,250){$\frac{1}{2}\log[2]$}
            \put(200,55){\footnotesize \textcolor{red}{$[F]$}}
            \put(100,90){\footnotesize \textcolor{red}{$[H]$}}
            \put(-25,62){\footnotesize \textcolor{red}{$[H]+[F]$}}
            \put(115,250){\footnotesize \textcolor{red}{$2[H]+[F]$}}
            \put(-25,110){\footnotesize \textcolor{blue}{$[H]+2[F]$}}
            \put(32,135){\footnotesize \textcolor{blue}{$[H]+3[F]$}}
            \put(32,200){\footnotesize \textcolor{blue}{$2[H]+5[F]$}}
        \end{picture}
        \vspace*{-0.25cm}
        \caption{Tensions of BPS strings}
        \label{fig:F2tensions}
    \end{subfigure}
    \caption{We illustrate the relevant physics when $B=\mathbb{F}_2$. In \cref{fig:F2cones}, the BPS brane cone is the red region and the BPS black brane cone is the blue region. The two red rays indicate the extremal rays of $\cC_{\rm BPS-B}$ which correspond to $H,F$. Additionally, the blue ray denotes the boundary between $\cC_{\rm BPS-BB}$ and $\cC_{\rm BPS-B}$. The value of $h^0(\mathbb{F}_2,\cO_{\mathbb{F}_2}(aH+bF))$ is computed via \cref{eq:F2h0} and indicated at each integral site in $\cC_{\rm BPS-B}$. The dotted sites are those chosen to be further analyzed in \cref{fig:F2tensions}. In \cref{fig:F2tensions}, the tension of the extremal rays for $\cC_{\rm BPS-B}$ and $\cC_{\rm BPS-BB}$ are shown in red and blue solid lines, respectively. We have also chosen to illustrate the tension of BPS strings in the interior of $\cC_{\rm BPS-B}$ and in the interior of $\cC_{\rm BPS-BB}$ in red and blue dashed lines. The gray vertical solid line indicates the boundary of the K\"ahler cone.}
    \label{fig:F2}
\end{figure}

From the above attractor solutions, we find that the extremal tension associated to strings of charge $q=q_h\tilde{H}+q_f\tilde{F}$ inside the black cone is 
\begin{equation}
    T^{\rm Ext}=\sqrt{nq_h^2+2q_hq_f}\,,
\end{equation}
which exactly coincides with $Z_{\rm min}$ for all integral charge sites inside the black cone. Hence, all integral sites inside the black cone are occupied by a BPS state. Geometrically, as all such sites are non-negative combinations of effective curves $(\tilde{H},\tilde{F})$, they are by definition effective as well. Nevertheless, we can still compute $h^0$ for any generic divisor in the effective cone. Namely, suppose $D=aH+bF$ and viewing $\pi: \mathbb{F}_n=\mathbb{P}(\cO_{\mathbb{P}^1}\oplus \cO_{\mathbb{P}^1}(n))\to \mathbb{P}^1$, then we have the following pushforward to the base $\mathbb{P}^1$
\begin{equation}
    \pi_*\cO_{\mathbb{F}_n}(aH+bF)\cong \bigoplus_{i=0}^a\cO_{\mathbb{P}^1}(b-ni)\,,\quad \Rightarrow \quad H^0(\mathbb{F}_n,\cO_{\mathbb{F}_n}(aH+bF))\cong \bigoplus_{i=0}^aH^0(\mathbb{P}^1,\cO_{\mathbb{P}^1}(b-ni))\,,
\end{equation}
for all $a,b$ such that $[D]\in \NE(\mathbb{F}_n)$. Therefore, for any divisor class $[D]$ in the effective cone, we have
\begin{equation}
\label{eq:F2h0}
    h^0(\mathbb{F}_n,\cO_{\mathbb{F}_n}(aH+bF))=\sum_{i=0}^a\mathrm{max}(b-ni+1,0)\,,
\end{equation}
where we have made use of $h^0(\mathbb{P}^1,\cO_{\mathbb{P}^1}(a))=H(a)(a+1)$ with $H(a)$ being the Heaviside function. The behavior of $h^0$ is illustrated in \cref{fig:F2cones} for all integral sites within $\cC_{\rm BPS-B}$. We can see that all such integral sites in $\cC_{\rm BPS-B}$ for $\mathbb{F}_n$ are hence occupied by a BPS state. Furthermore, the tensions of strings shown in \cref{fig:F2tensions} reveal the distinct four categories of BPS strings we encounter as summarized in \cref{sec:tension properties}. Additionally, we can see that the primitive generators of $\cC_{\rm BPS-B}$ and $\cC_{\rm BPS-BB}$, namely $[H],[F],[H]+2[F]$, all have minimum tension $\lesssim 1$ as conjectured in~\cite{Nevoa:2025xiq}.

As a last remark on these theories, let us consider non-BPS strings with charges $q=q_H\tilde{H}+q_F\tilde{F}$ where $q_H>0$ and $q_F<0$. Clearly $q\notin M(\mathbb{F}_n)$. Let us still suppose $q^2>0$. In this case, the central charge takes on the same form as \cref{eq:Fn Z} and is crucially minimized along the boundary $f\to 0$ and $h\to 1/\sqrt{n}$ where the minimum value is
\begin{equation}
    Z_{\rm min}^{\rm non-BPS}=\frac{nq_h+q_f}{\sqrt{n}}\,.
\end{equation}
Using this, the extremal tension is deduced from the black string potential as
\begin{equation}
    V^{\rm Ext}=nq_h^2+2q_hq_f+\frac{2q_f^2}{h}\,,
\end{equation}
which leads to the extremal tension $T^{\rm Ext}=\sqrt{V^{\rm Ext}}$.
This is an unstable value. Nevertheless, we can see that 
\begin{equation}
    V^{\rm Ext}-(Z_{\rm min}^{\rm non-BPS})^2=\frac{q_f^2}{n}>0\,,\qquad \Rightarrow \qquad T^{\rm Ext}>|Z_{\rm min}|\,.
\end{equation}
Therefore, non-BPS states indeed satisfy the strict inequality
\begin{equation}
    |Z_{\rm min}|<T^{\rm non-BPS}<T^{\rm Ext}\,,
\end{equation}
in these theories.

\subsubsection{Blowups of Hirzebruch surfaces}
\label{sec:blowups}

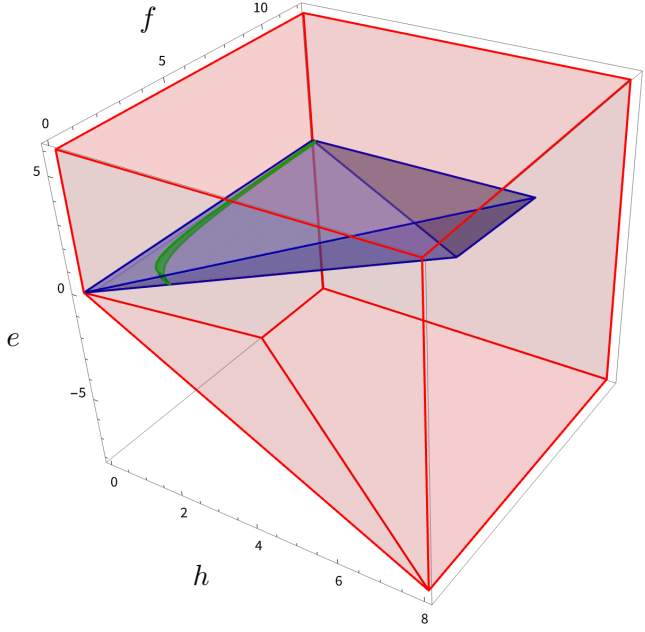
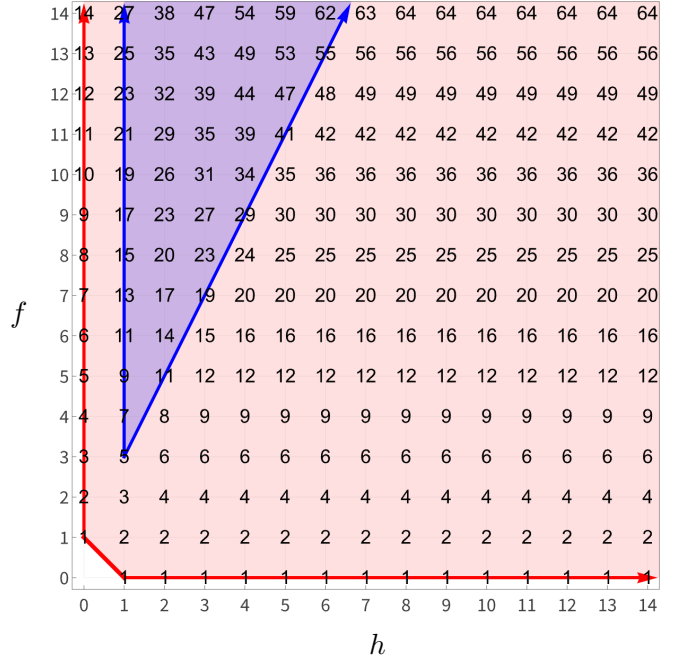
\begin{figure}[!tp]
    \centering
    \begin{subfigure}{.475\textwidth}
        \includegraphics[width=\linewidth]{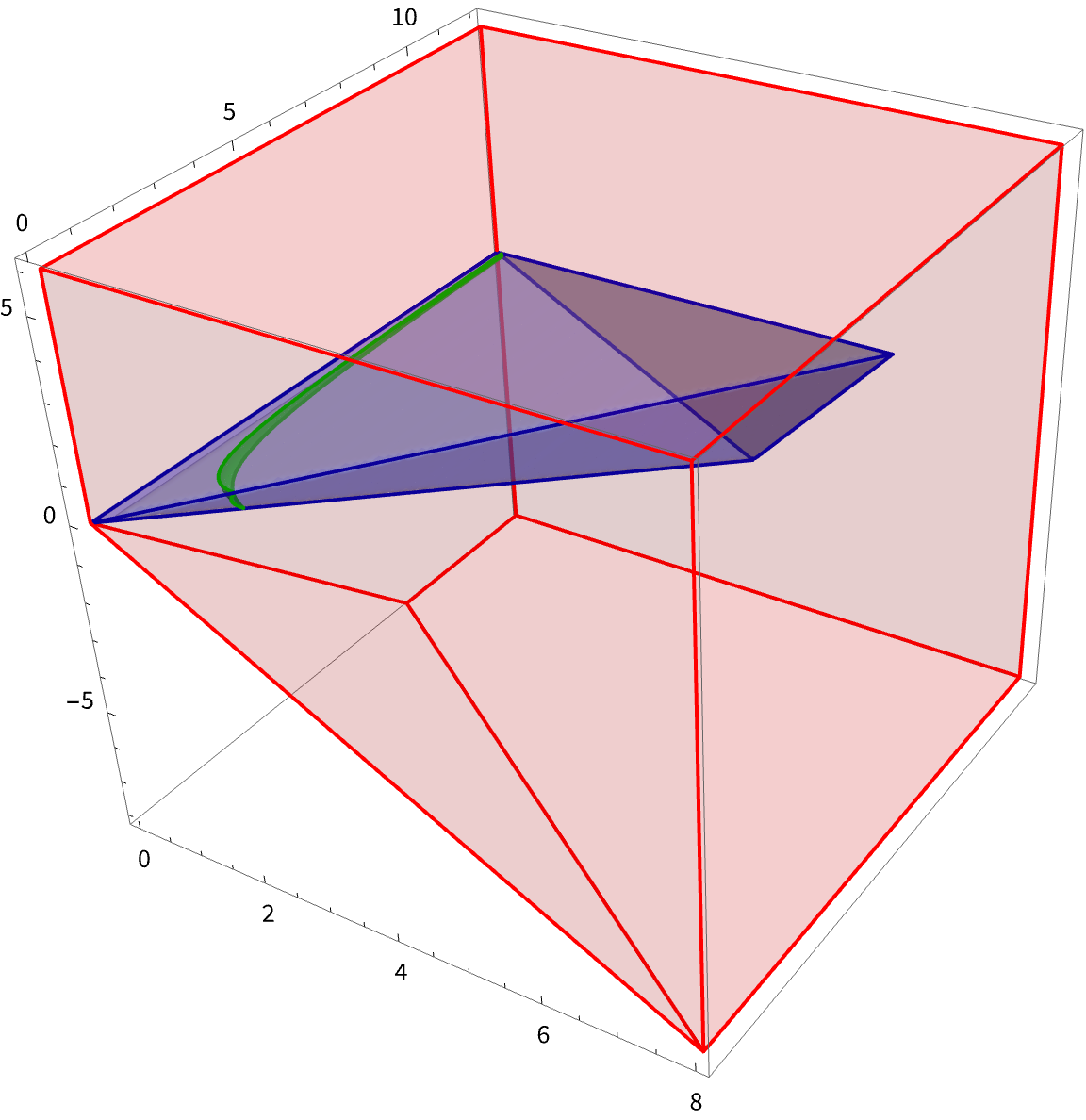}
        \begin{picture}(0,0)\vspace*{-1.2cm}
            \put(-10,120){$e$}
            \put(60,30){$h$}
            \put(40,240){$f$}
        \end{picture}
        \vspace*{-0.25cm}
        \caption{$\cC_{\rm BPS-B}$ and $\cC_{\rm BPS-BB}$}
        \label{fig:Bl1F2cones}
    \end{subfigure}
    \hfill
    \begin{subfigure}{.475\textwidth}
        \includegraphics[width=\linewidth]{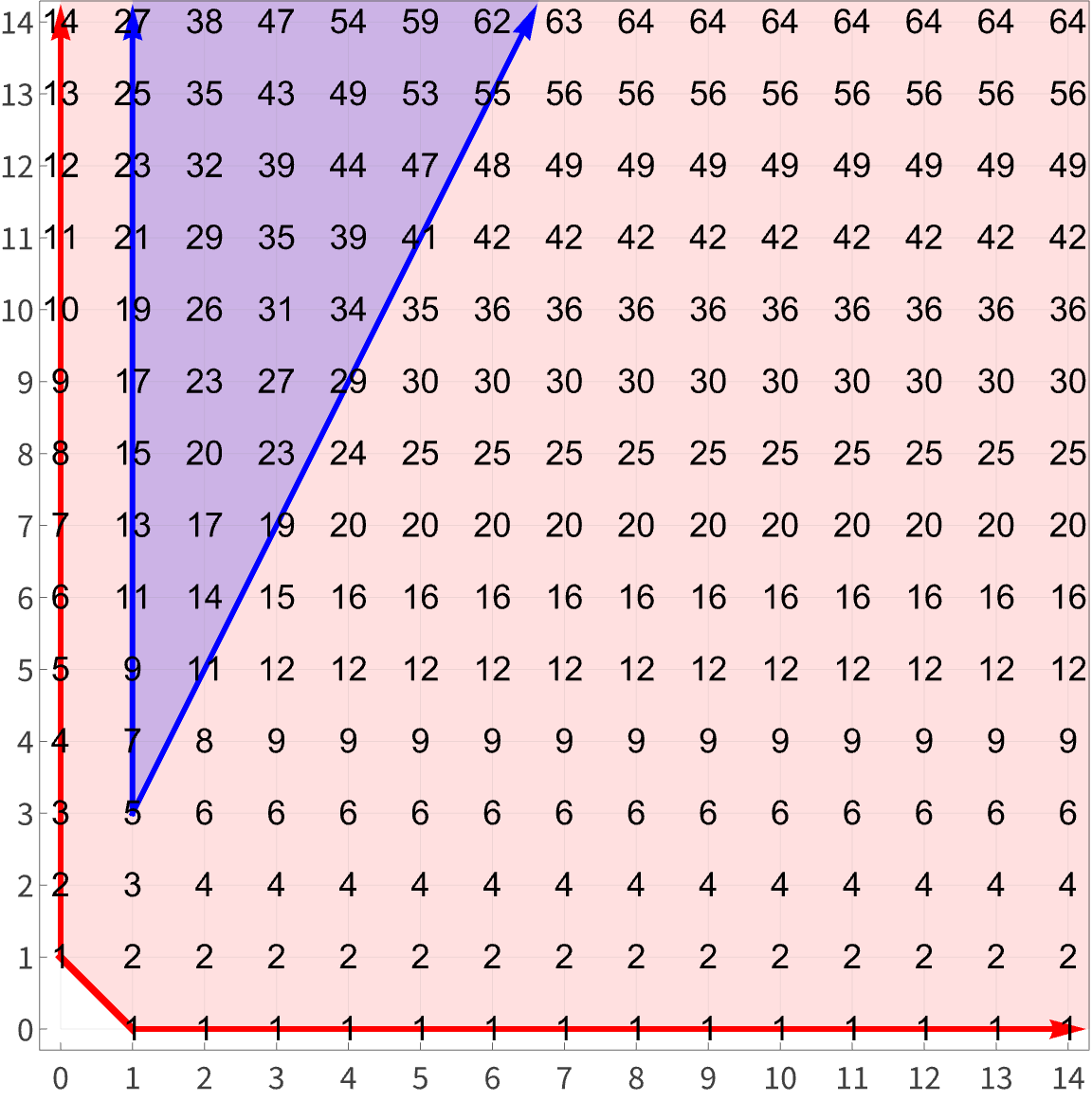}
        \begin{picture}(0,0)\vspace*{-1.2cm}
            \put(-15,125){$f$}
            \put(120,0){$h$}
        \end{picture}
        \vspace*{-0.25cm}
        \caption{$h^0$ along $e=-1$}
        \label{fig:Bl1F2effective}
    \end{subfigure}
    \caption{The behavior of cones and effectiveness of divisors in $\mathrm{Bl}_1(\mathbb{F}_2)$. In \cref{fig:Bl1F2cones}, the red cone indicates $\cC_{\rm BPS-B}$ while the blue cone indicates $\cC_{\rm BPS-BB}$. Additionally, the green surface indicates the constant-volume slice within the K\"ahler cone. In \cref{fig:Bl1F2effective}, we have chosen a generic 2d slice in $\cC_{\rm BPS-B}$ (for simplicity of illustration) to demonstrate BPS completeness at all integral sites in $\cC_{\rm BPS-B}$.}
    \label{fig:Bl1F2}
\end{figure}

The Hirzebruch surfaces serve as the simplest base manifold of which we can iteratively perform blowups on to obtain higher $h^{1,1}$ toric bases for elliptic Calabi--Yau threefolds. Or, in other words, these manifolds provide a perfect ground for studying the relation of $\cC_{\rm BPS-B}$ and $\cC_{\rm BPS-BB}$ within a birationally-equivalent family of surfaces. For simplicitly, let us just consider the one-point blowup of Hirzebruch surfaces such that $h^{1,1}(B)=3$, which we will denote as $\mathrm{Bl}_1(\mathbb{F}_n)$. In this case, blowing up a point at $p\in H\cap F$ and denoting the resulting exceptional curves by $E$, we have the following intersection data
\begin{equation}
    H_1.H_1=-(n+1)\,,\quad H.F=1\,,\quad F.F=0\,,\quad E.E=-1\,,\quad H.E=1\,,\quad E.F=0\,.
\end{equation}
The two curves passing through $p$ are $H_1=H-E$ and $L=F-E$. The primitive rational curves of $\mathrm{Bl}_1(\mathbb{F}_n)$ are $E,L,H_1$ and the canonical class is
\begin{equation}
    K_{\mathrm{Bl}_1(\mathbb{F}_n)}=\pi^*K_{\mathbb{F}_n}+E=-2H_1-(n+2)F\,.
\end{equation}
Alternatively, we could also choose to perform the blowup at a generic point $p\notin H$, then $H_1=H$ in $\mathbb{F}_n$ while the remaining intersection data remains the same.

In the basis $H_1,F,E$, a divisor is $D=hH_1+fF+eE$. Then, the effective cone in this choice of basis is
\begin{equation}
    M(\mathrm{Bl}_1(\mathbb{F}_n))=\{(h,f,e)\in \mathbb{R}^3\,\vert\, h\geq 0\,,f\geq 0\,, h+f+e\geq 0\}\,.
\end{equation}
In other words, the effective cone is generated by $E,F-E,H_1-E$.
Similar to the $\mathbb{F}_n$ cases, we can choose to express the K\"ahler two-form in the more standard positive section basis. However, let us continue to work in the same basis as that for the effective cone. Namely, we can express the K\"ahler two-form as $J=hH_1+fF+eE$. Then, the volumes associated to the generators of the effective cone are 
\begin{equation}
    J.E=-e\,,\qquad J.(F-E)=h+e\,,\qquad J.(H_1-E)=f-nh+e\,.
\end{equation}
Therefore, the K\"ahler cone is
\begin{equation}
    K(\mathrm{Bl}_1(\mathbb{F}_n))=\{(h,f,e)\in \mathbb{R}^3\,\vert\, e<0\,, h+e>0\,, f-nh+e>0\}\,.
\end{equation}
The attractor solution then again exactly reproduces the bounds we have for the K\"ahler cone. Therefore, we have $\cC_{\rm BPS-BB}=K(\mathrm{Bl}_1(\mathbb{F}_n))$. The $\cC_{\rm BPS-B}$ and $\cC_{\rm BPS-BB}$ have been illustrated for $\mathrm{Bl}_1(\mathbb{F}_2)$ in \cref{fig:Bl1F2cones}. Within a generic 2d slice in $\cC_{\rm BPS-B}$, BPS completeness is satisfied as can be seen in \cref{fig:Bl1F2effective}. In particular, we can algorithmically determine $h^0(\mathrm{Bl}_1(\mathbb{F}_n),\cO_{\mathrm{Bl}_1(\mathbb{F}_n)}([D]))$ for any given $[D]\in \cC_{\rm BPS-B}$ and observe that the above observation is indeed general in $n$ and we have BPS completeness within the entire $\cC_{\rm BPS-B}$ for all such $\mathrm{Bl}_1(\mathbb{F}_n)$.

\begin{figure}[!tp]
    \centering
    \begin{subfigure}{.45\textwidth}
        \includegraphics[width=\linewidth]{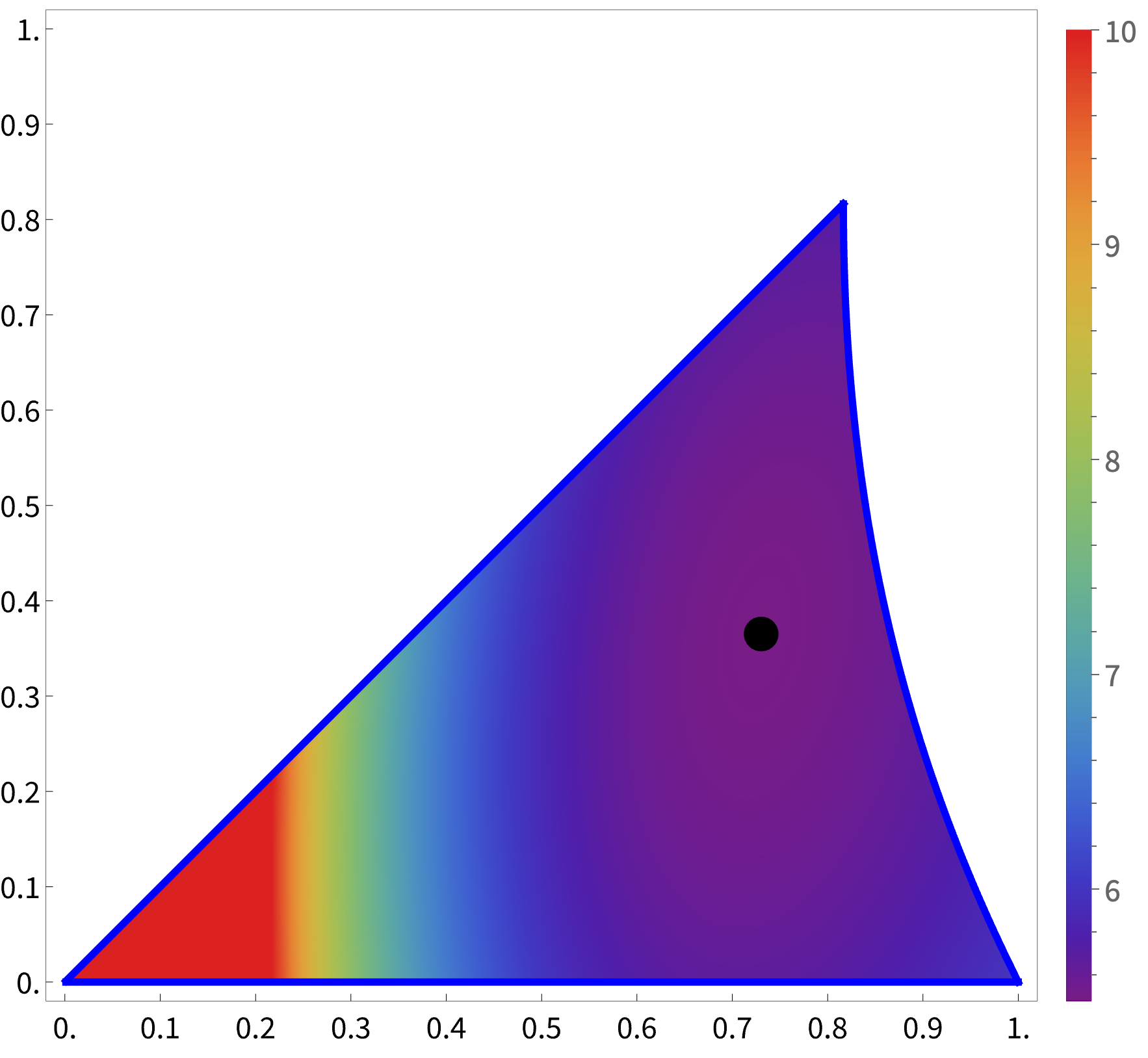}
        \begin{picture}(0,0)\vspace*{-1.2cm}
            \put(-20,110){$\tilde{e}$}
            \put(100,0){$h$}
            \put(200,225){$T$}
            \put(50,225){$[D]=2[H_1]+6[F]-[E]$}
        \end{picture}
        \vspace{-0.5cm}
        \caption{$q\in \cC_{\rm BPS-BB}$}
        \label{fig:Bl1F2ample}
    \end{subfigure}
    \hfill
    \vspace*{0.8cm}
    \begin{subfigure}{.45\textwidth}
        \includegraphics[width=\linewidth]{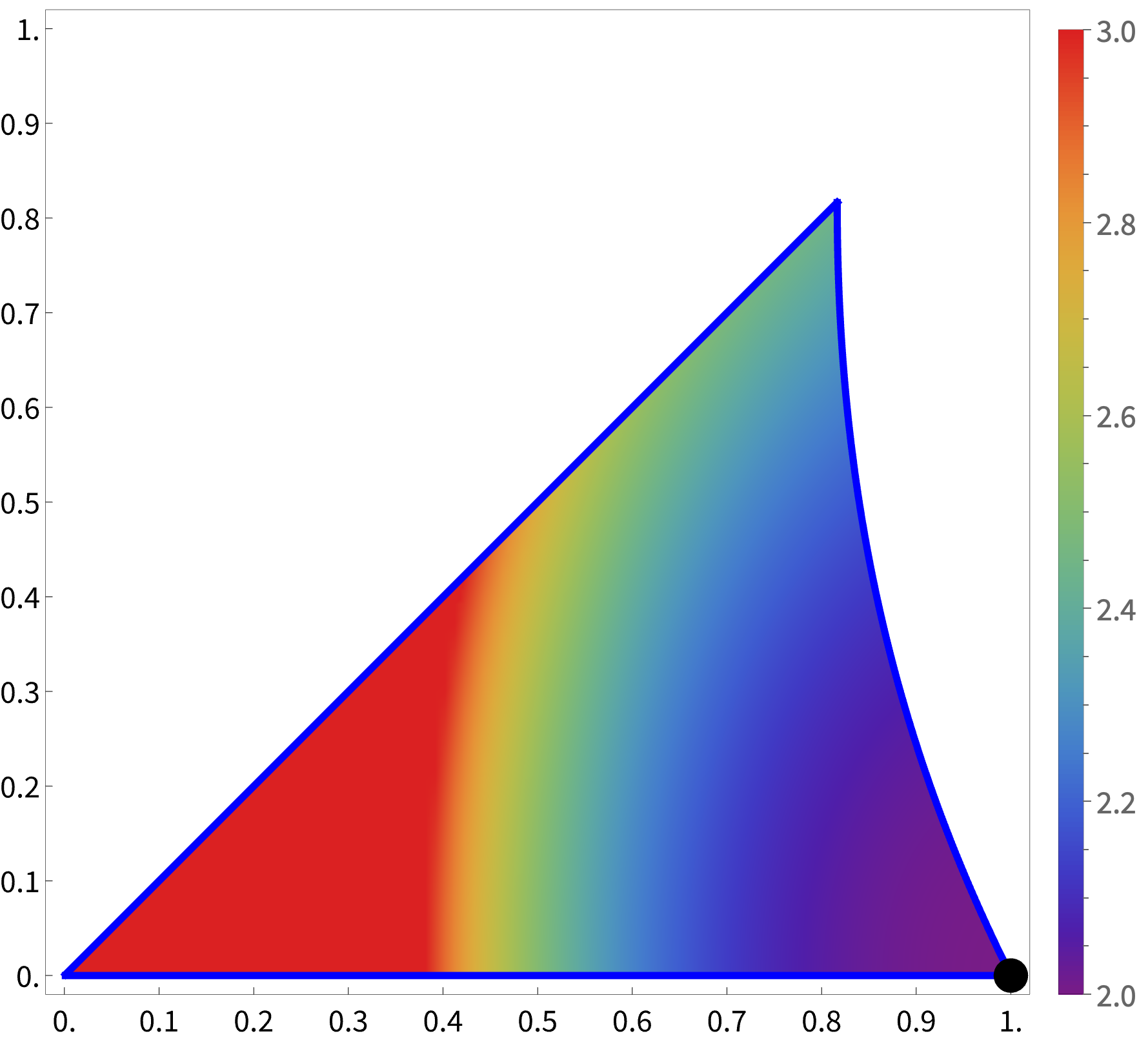}
        \begin{picture}(0,0)\vspace*{-1.2cm}
            \put(-20,110){$\tilde{e}$}
            \put(100,0){$h$}
            \put(200,225){$T$}
            \put(70,225){$[D]=[H_1]+2[F]$}
        \end{picture}
        \vspace{-0.5cm}
        \caption{$q\in \partial \cC_{\rm BPS-BB}$}
        \label{fig:Bl1F2nef}
    \end{subfigure}
    \hfill
    \begin{subfigure}{.45\textwidth}
        \includegraphics[width=\linewidth]{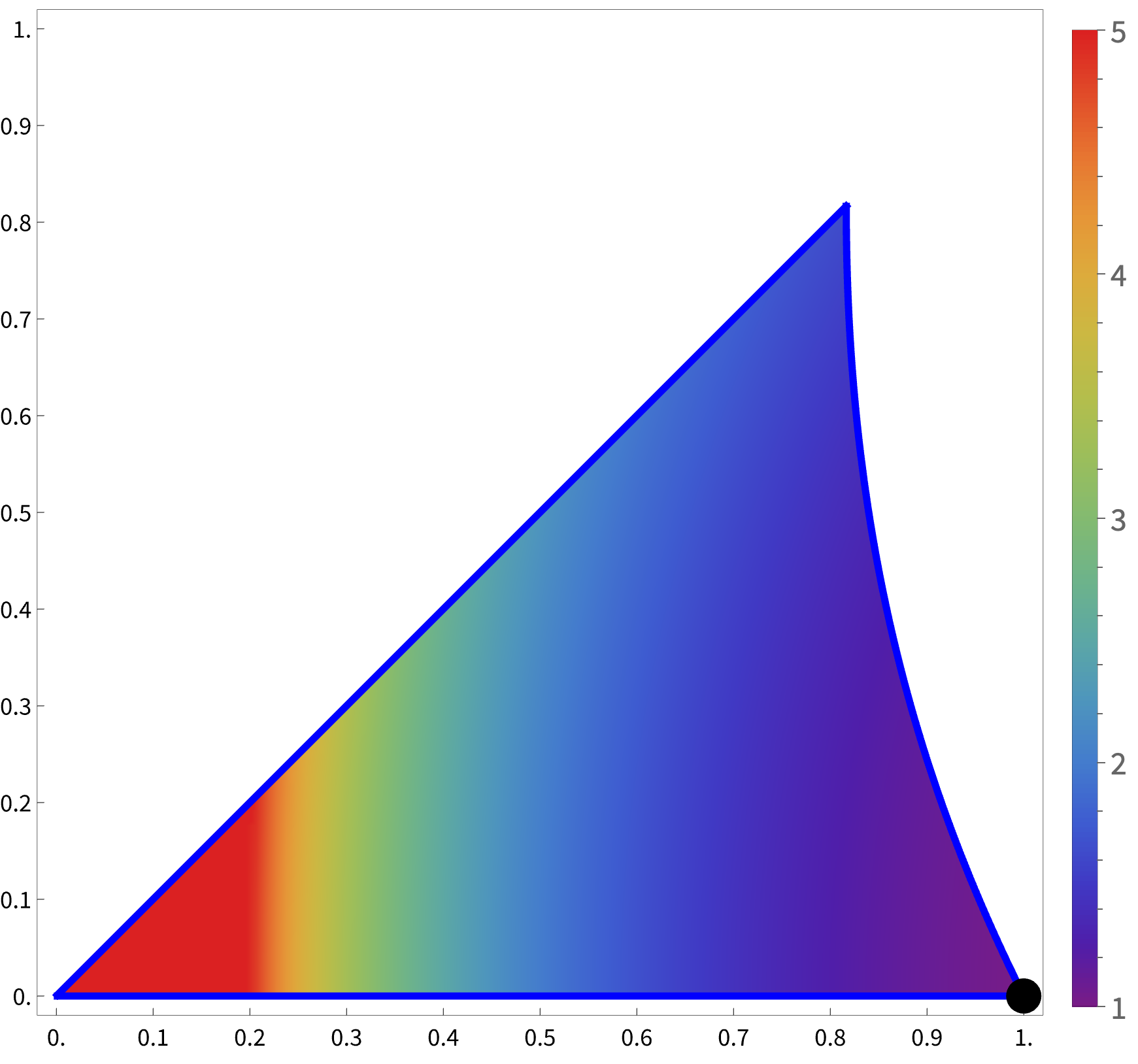}
        \begin{picture}(0,0)\vspace*{-1.2cm}
            \put(-20,110){$\tilde{e}$}
            \put(100,0){$h$}
            \put(205,225){$T$}
            \put(70,225){$[D]=[H_1]+[F]$}
        \end{picture}
        \vspace{-0.5cm}
        \caption{$q\in \cC_{\rm BPS-B}\backslash\cC_{\rm BPS-BB}$ and $q\notin \partial \cC_{\rm BPS-B}$}
        \label{fig:Bl1F2eff}
    \end{subfigure}
    \hfill
    \begin{subfigure}{.45\textwidth}
        \includegraphics[width=\linewidth]{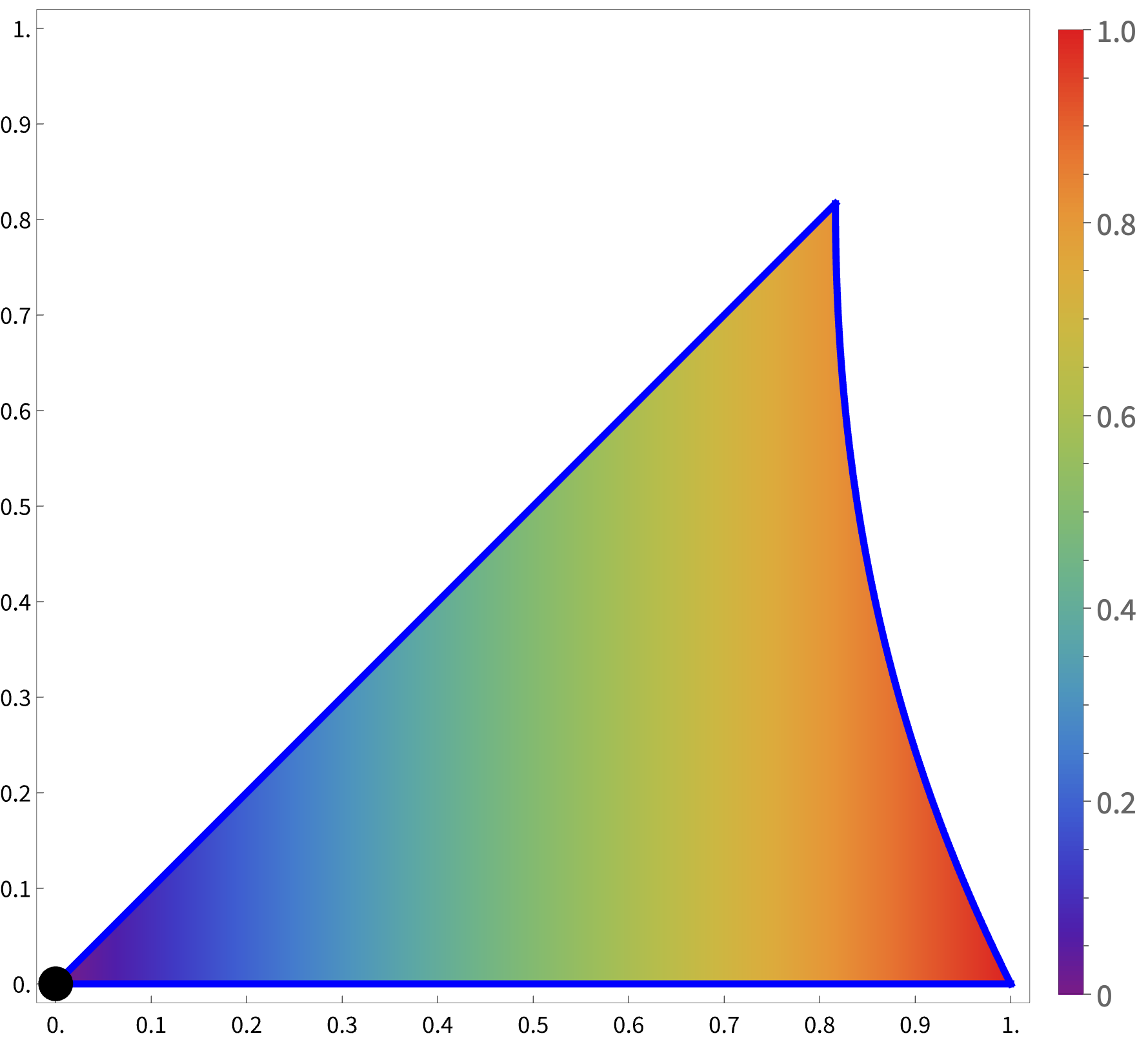}
        \begin{picture}(0,0)\vspace*{-1.2cm}
            \put(-20,110){$\tilde{e}$}
            \put(100,0){$h$}
            \put(205,225){$T$}
            \put(80,225){$[D]=[F]$}
        \end{picture}
        \vspace{-0.5cm}
        \caption{$q\in \partial \cC_{\rm BPS-B}$}
        \label{fig:Bl1F2effboundary}
    \end{subfigure}
    \caption{The tension of the four types of BPS branes in $\cC_{\rm BPS-B}$ of F-theory on elliptic Calabi--Yau threefold with $B=\mathrm{Bl}_1(\mathbb{F}_2)$. Here, for numerical stability, we have chosen to parameterize the K\"ahler form as $J=hH_1+fF-\tilde{e}E$. The infinite-distance limit in this parameterization is the $\tilde{e}\to 0$ and $h\to 0$ limit. The black dot indicates the location of the minimum of the tension of the chosen BPS brane.}
    \label{fig:Bl1F2tensions}
\end{figure}
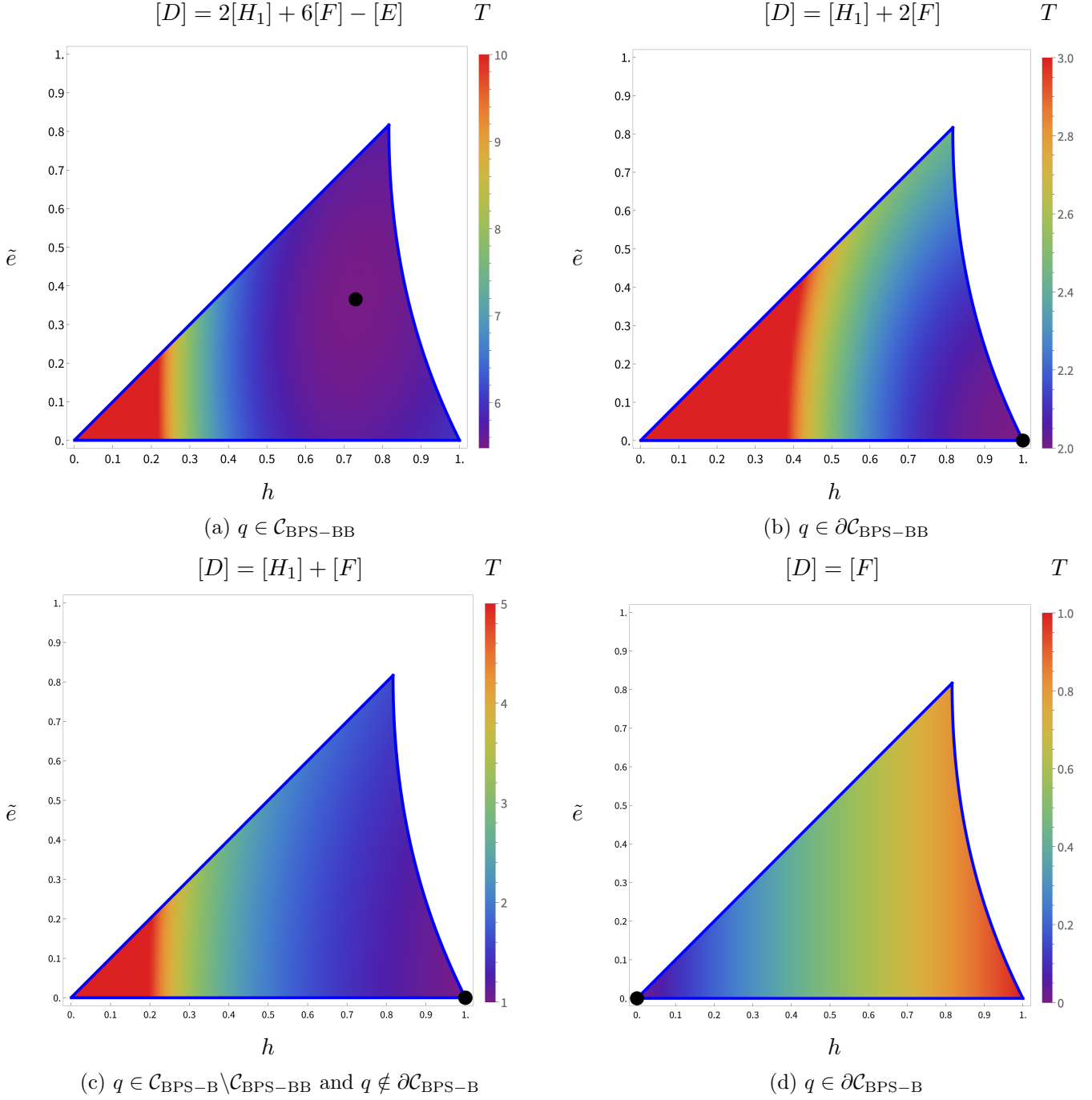

At last, we can also compute the tension of BPS strings within $\cC_{\rm BPS-B}$ for a given $\mathrm{Bl}_1(\mathbb{F}_n)$. Here, we will focus simply on $\mathrm{Bl}_1(\mathbb{F}_2)$ as the moduli space has also been illustrated previously in \cref{fig:Bl1F2cones}, namely the co-dimension one green surface within $\cC_{\rm BPS-BB}$. Examining the four special classes of divisors as laid out in \cref{sec:tension properties}, we see that the observation holds in this example. For the divisor that is in the interior of $\cC_{\rm BPS-BB}$, the tension is minimized in the interior of the moduli space. This is illustrated in \cref{fig:Bl1F2ample}. For the divisor on the boundary of $\cC_{\rm BPS-BB}$, we can see that the minimum tension is located on the boundary of the moduli space. This is illustrated in \cref{fig:Bl1F2nef}. However, to see that this is indeed a stable minimum, let us compute the norm of its gradient as we approach the minimum tension associated to the BPS brane presented in \cref{fig:Bl1F2nef}. The tension of this BPS brane is $T=f$. Then, we have
\begin{equation}
    \lim_{(h,\tilde{e})\to (1,0)}|\nabla T|=\lim_{(h,\tilde{e})\to (1,0)}\sqrt{\bigg(1-\frac{1}{h^2}-\frac{\tilde{e}^2}{2h^2}\bigg)^2+\bigg(\frac{\tilde{e}}{h}\bigg)^2}=0\,,
\end{equation}
with $\nabla T\vert_{(h,\tilde{e})=(1,0)}=(0,0)$. From here, we can also compute the Hessian, which is
\begin{equation}
    \lim_{(h,\tilde{e})\to (1,0)}(\partial_i\partial_j T)=\lim_{(h,\tilde{e})\to (1,0)}\begin{pmatrix}
        \frac{2+\tilde{e}^2}{h^3}&-\frac{\tilde{e}}{h^2}\\
        -\frac{\tilde{e}}{h^2}&\frac{1}{h}
    \end{pmatrix}
    =\begin{pmatrix}
            2&0\\
            0&1
        \end{pmatrix}\,.
\end{equation} 
Therefore, indeed this is a stable minimum. Next, let us consider states that are within $\cC_{\rm BPS-B}$ but strictly outside $\cC_{\rm BPS-BB}$. The tension again minimizes along the boundary of the moduli space as can be seen in \cref{fig:Bl1F2eff}. In this case, we again have to examine the stability of this minimum. The tension associated to the BPS brane in \cref{fig:Bl1F2eff} is $T=f-h$. Then, the norm of the gradient along the constant volume slice as we approach its minimum within the K\"ahler moduli space is
\begin{equation}
    \lim_{(h,\tilde{e})\to (1,0)}|\nabla T|=\lim_{(h,\tilde{e})\to (1,0)}\sqrt{\frac{(1+\tilde{e}^2/2)^2}{h^4}+\bigg(\frac{\tilde{e}}{h}\bigg)^2}=1\,,
\end{equation}
where we have $\nabla T\vert_{(h,\tilde{e})=(1,0)}=(-1,0)$. Therefore, this is not a stable minimum and the global minimum of its tension lies outside the moduli space. Lastly, we can consider states on the boundaries of $\cC_{\rm BPS-B}$. In this case, we have chosen $D=F$ which is both the boundary of $\cC_{\rm BPS-B}$ and that of $\cC_{\rm BPS-BB}$. In this case, the minimum of the tension is pushed out to the infinite-distance limit of the moduli space as can be seen in \cref{fig:Bl1F2effboundary}. As this approaches 0, it must be a global stable minimum. Hence, indeed all four classes of BPS branes agree with the general properties laid out in \cref{sec:tension properties}.

\subsubsection{del Pezzo surfaces}
\label{sec:dPr}

Let us next consider the del Pezzo surfaces $dP_r$, obtained by blowing up
$\mathbb{P}^2$ at $r$ generic points, with $0\leq r\leq 8$. We use the standard
basis
\begin{equation}
    H,\qquad E_i,\quad i=1,\dots,r,
\end{equation}
where $H$ is the pullback of the hyperplane class of $\mathbb{P}^2$ and $E_i$
are the exceptional curves. The intersection pairing is
\begin{equation}
    H.H=1,\qquad H.E_i=0,\qquad E_i.E_j=-\delta_{ij}.
\end{equation}
The anti-canonical class is
\begin{equation}
    -K_{dP_r}=3H-\sum_{i=1}^r E_i.
\end{equation}
and it is ample, which is the defining positivity property of
a del Pezzo surface.

For $r\geq 2$, the Mori cone is generated by the $(-1)$-curves. In the above
basis, these are precisely the effective curve classes $C$ satisfying
\begin{equation}
    C^2=-1,\qquad (-K_{dP_r}).C=1.
\end{equation}
Explicitly, the generators are as follows:
\begin{equation}
\begin{array}{c|c|c}
\text{curve class} & \text{range} & \text{number} \\
\hline
E_i & r\geq 1 & r \\[2pt]
H-E_i-E_j & r\geq 2 & \binom{r}{2} \\[2pt]
2H-\sum_{a=1}^5 E_{i_a} & r\geq 5 & \binom{r}{5} \\[2pt]
3H-2E_i-\sum_{a=1}^6E_{j_a} & r\geq 7 & r\binom{r-1}{6} \\[2pt]
4H-2\sum_{a=1}^3E_{i_a}-\sum_{b=1}^5E_{j_b} & r=8 & \binom{8}{3} \\[2pt]
5H-2\sum_{a=1}^6E_{i_a}-\sum_{b=1}^2E_{j_b} & r=8 & \binom{8}{6} \\[2pt]
6H-3E_i-2\sum_{j\neq i}E_j & r=8 & 8
\end{array}
\end{equation}
where all indices appearing in a given class are distinct, and for the last
three rows the remaining indices are taken inside $\{1,\dots,8\}$. For
$r=0$, the Mori cone is generated by $H$. For $r=1$, one has
\begin{equation}
    M(dP_1)=\Cone(E_1,H-E_1),
\end{equation}
where $H-E_1$ is the fiber class of the Hirzebruch surface
$dP_1\simeq \mathbb{F}_1$ and has self-intersection zero.

Thus, for $r=2,\dots,8$, the numbers of extremal Mori generators are
\begin{equation}
    3\,,\ 6\,,\ 10\,,\ 16\,,\ 27\,,\ 56\,,\ 240\,,
\end{equation}
respectively. The BPS 1-brane cone in F-theory on the elliptic Calabi--Yau threefold with a base $dP_r$ is therefore
\begin{equation}
    \cC_{\rm BPS-B}=M(dP_r)\,,
\end{equation}
with generators given by the above curve classes. Equivalently, a D3-brane
wrapped on any one of these effective curves gives a basic BPS string.

The BPS-BB cone is the dual nef cone,
\begin{equation}
    \cC_{\rm BPS-BB}=\Nef(dP_r)=M(dP_r)^\vee\,.
\end{equation}
If we write a general divisor class as
\begin{equation}
    J=tH-\sum_{i=1}^r s_iE_i\,,
\end{equation}
then nef-ness is the condition that $J.C\geq 0$ for every Mori generator $C$ listed above. For example, the first few inequalities are
\begin{equation}
    s_i\geq 0\,,\qquad
    t-s_i-s_j\geq 0\,,\qquad
    2t-\sum_{a=1}^5s_{i_a}\geq 0\,,
\end{equation}
together with the analogous inequalities coming from the higher-degree generators in the table. The K\"ahler cone is obtained by replacing these non-strict inequalities by strict inequalities.

It is clearly related to the exceptional root systems $E_n$. The lattice orthogonal to the canonical class,
\begin{equation}
    \Lambda_{E_r}:=K_{dP_r}^{\perp}\subset H_2(dP_r)\,,
\end{equation}
is the negative-definite root lattice of type $E_r$, with the standard convention
\begin{equation}
    E_5=D_5,\qquad E_4=A_4,\qquad E_3=A_2\oplus A_1\,.
\end{equation}
For $r\geq 3$, a convenient set of simple roots is
\begin{equation}
    \alpha_i=E_i-E_{i+1},\qquad i=1,\dots,r-1\,,
\end{equation}
together with
\begin{equation}
    \alpha_r=H-E_1-E_2-E_3\,.
\end{equation}
These satisfy
\begin{equation}
    \alpha_i.K_{dP_r}=0,\qquad \alpha_i^2=-2,\qquad
   \alpha_i.\alpha_j=-C^{E_r}_{ij}\,,
\end{equation}
where $C^{E_r}_{ij}$ is the Cartan matrix of the corresponding $E_r$ root system. The Weyl group $W(E_r)$ acts on $\mathrm{Pic}(dP_r)$ by reflections in these roots and permutes the $(-1)$-curve generators of the Mori cone.

This is the same exceptional structure that appears in Seiberg's rank-one five-dimensional SCFTs~\cite{Seiberg:1996bd}. In M-theory on the local Calabi-Yau threefold $K_{dP_r}$, collapsing the compact surface $dP_r$ gives a 5d fixed point whose flavor root lattice is precisely $K_{dP_r}^{\perp}\simeq \Lambda_{E_r}$.

\subsubsection{Enriques surfaces}
\label{sec:Enriques}

As a last example, let us consider the Enriques surfaces $S$ as base manifolds of elliptic Calabi--Yau threefolds in F-theory compactifications. The (complex) Enriques surfaces can be engineered as $K3/\iota$ where $\iota$ is a fixed-point-free involution of $K3$. Notably, the Enriques surfaces cannot be obtained from blowups of other surfaces. For a detailed exposition on Enriques surface, see e.g.,~\cite{CossecDolgachevLiedtke2025}. Distinct from the previous two families of base manifolds, the structure sheaf $\cO_S$ of Enriques surfaces have instead the following properties
\begin{equation}
    h^1(S,\cO_S)=0\,,\quad h^2(S,\cO_S)=0\,,\quad K_S^{\otimes 2}\cong \cO_S\,,
\end{equation}
where $K_S$ denotes the canonical class on $S$. Here, $K_S$ is torsional and is then numerically trivial $K_S\equiv 0$, leading to the special property of $(-K_S).C=c_1(S).C=0$ for all curve classes $C\subset S$.
Hence, it is useful to introduce the N\'eron--Severi group, which is defined as
\begin{equation}
    \mathrm{NS}(S)=\mathrm{Pic}(S)/\mathrm{Pic}^0(S)\,,
\end{equation}
where the connected component of the identity of the Picard group, $\mathrm{Pic}(S)$, is trivial. (This is the cases for all such regular manifolds that we have considered.) However, as we have $\mathrm{Tor\,} \mathrm{Pic}(S)=\langle K_S\rangle\cong\mathbb{Z}_2$, the relevant charge lattice for the 6d theory is then the numerical lattice of the Enriques surface
\begin{equation}
    \mathrm{Num}(S)=\mathrm{NS}(S)/\langle K_S\rangle\cong U\oplus E_8\,,
\end{equation}
where $U$ is the hyperbolic plane and Num$(S)$ is then an even unimodular lattice of signature $(1,9)$. The BPS brane cone is then $\cC_{\rm BPS-B}\subseteq \mathrm{Num}(S)_{\mathbb{R}}$. Note, by definition, this lattice contains the full supersymmetric lattice, namely both BPS states and anti-BPS states.

With this, we can choose a basis $\{[D_\alpha]\}$ of Num$(S)$ such that the K\"ahler form is $J=\sum_{\alpha =0}^9t^\alpha[D_\alpha]$. Furthermore, for any divisor class, its norm is $D^2=\eta_{\alpha\beta}D^\alpha D^\beta$ with $\eta_{\alpha\beta}$ having signature $(1,9)$. Therefore, the prepotential of the resulting 6d theory is
\begin{equation}
    \cF=\frac12\eta_{\alpha\beta}t^\alpha t^\beta=1\,,
\end{equation}
where we restrict to the constant-volume slice. Additionally, for a string of charge $q$, the central charge is $Z=\eta_{\alpha\beta}q^\alpha t^\beta$ and the attractor solution is that of \cref{eq:6d attractor solution} evaluated at the attractor point \cref{eq:6d attractor point}. Consistent with the above two families of base manifolds, we again observe that $\cC_{\rm BPS-BB}=K(S)$ whose closure is the nef cone. With the $\cC_{\rm BPS-BB}$ identified, we also know that there are nodal curves within Enriques surfaces with $C^2=-2$. These are topologically $C\cong \mathbb{P}^1$. While these are effective divisors, they are clearly not nef. Therefore, we have the strict inclusion of $\cC_{\rm BPS-BB}\subset \cC_{\rm BPS-B}$.
Similar to the above two families of manifolds, the duality in \cref{conj:cone dual conjecture} is also manifested via the duality among geometric cones.

Here, let us consider how \cref{theorem:6d effectiveness} is at play when determining the effectiveness of divisors in the Enriques surfaces. To start, let us consider an ample divisor $H\subset S$.  and also a divisor $D$ with 
\begin{equation}
\label{eq:partial nef condition}
    D^2\geq 0\,,\qquad D.H>0\,.
\end{equation} 
Equivalently, from the line-bundle perspective, the above condition amounts to considering $\cO_S(D)\in \mathrm{Pic}(S)$ with $c_1(\cO_S(D))^2\geq 0$ and $c_1(\cO_S(D)).H>0$. From this condition, we can deduce that when such a divisor satisfies these properties, we have $c_1(K_S\otimes \cO_S(D)^{-1}).H=-c_1(\cO_S(D)).H<0$ implying $K_S\otimes \cO_S(D)$ is not effective. 
By Serre duality, we also have
\begin{equation}
    h^2(S,\cO_S(D))=h^0(S,K_S\otimes \cO_{S}(D)^{-1})\,,
\end{equation}
such that from the holomorphic Euler characteristic, we obtain
\begin{equation}
    h^0(S,\cO_S(D))=1+\frac12 c_1(\cO_S(D))^2+h^1(S,\cO_S(D))>0\,.
\end{equation}
Therefore, indeed when \cref{eq:partial nef condition} is satisfied, the divisor class is effective. 

Different from the previous cases, the Enriques surfaces exhibits a torsion subtlety in passing from $\mathrm{Pic}(S)$ to the numerical lattice
$\mathrm{Num}(S)$. In particular, $K_S$ is invisible in $\mathrm{Num}(S)$. Hence, the most obvious non-effective divisor is the Picard lifts of the zero numerical class where while the trivial line bundle (the structure sheaf) is effective $h^0(S,\cO_S)=1$, the canonical torsion line bundle is not $h^0(S,K_S)=0$ as the geometric genus of Enriques surfaces is 0. 
With such a non-effective divisor found in the Picard lifts, it is natural to ask whether this torsion subtlety also appears over
non-zero numerical classes. To this end, we can consider the nodal divisor $C$ and consider the torsion-twisted line bundle $\cO_S(C+K_S)$. In particular, the numerical class of $[C+K_S]$ lies along the same ray as the effective non-nef divisor $C$ in $\mathrm{Pic}(S)_{\mathbb{R}}$. Suppose there is an effective divisor $D\sim C+K_S$ of which the nodal divisor $C$ is irreducible and effective, then as $D.C<0$, this implies that $[C]$ must be a component of $[D]$ leading to the decomposition of $D=C+D'$ where $D'$ is effective. This would imply $D'\sim K_S$ such that $h^0(S,K_S)>0$, contradicting the Enriques condition. Hence, $C+K_S$ cannot be an effective divisor. Thus, we observe that the divisor class $C+K_S$ in $\mathrm{Pic}(S)$ is non-effective. However, similar to the zero numerical class, as $\cO_S(C+K_S)$ has the same numerical class as $\cO_S(C)$, namely $[C]=[C+K_S]\in \mathrm{Num}(S)$, then, in $\cC_{\rm BPS-B}$ which is concerned of the numerical lattice, this divisor class is instead occupied by the effective divisor as the physical BPS charge lattice is obtained upon quotienting out torsion.

We can continue to enumerate through all divisor classes. In particular, the divisor classes in $\cC_{\rm BPS-B}$ are generated by nef classes and nodal classes. As we have shown that these are effective in $\cC_{\rm BPS-B}$ for the Enriques surface, any linear combination of these two types of divisors will both remain within $\cC_{\rm BPS-B}$ and also remain effective. Hence, similar to the above two families of theories, we arrive at a stronger conclusion than \cref{conj:BPS completeness}, namely all integral sites in $\cC_{\rm BPS-B}$ are occupied by BPS states.

\subsubsection{Resolved $(T^2\times T^2)/\mathbb{Z}_3$}
\label{sec:6d holes}

Up to this point, we have seen BPS completeness in $\cC_{\rm BPS-B}$ which might lead us to conclude that for 6d $\cN=(1,0)$ quantum gravity theories, \cref{conj:BPS completeness} can be expanded to include all integral states in $\cC_{\rm BPS-B}$. We will see that this is not true in following example. 

We will consider the example in \cite[App. D]{Kim:2024eoa}. The starting point is to consider F-theory on $T^6/\mathbb{Z}_3$ which was originally studied in~\cite{Morrison:1996pp}. The $\mathbb{Z}_3$ acts diagonally as
\begin{equation}
    g\,:\,(z_1,z_2,z_3)\mapsto (\omega z_1,\omega z_2,\omega z_3)\,,
\end{equation}
with $\omega=e^{2\pi i/3}$. This has 27 fixed point, which upon resolving via blowups, and we can obtain an elliptic Calabi--Yau threefold. Identifying one of the two-torus as the elliptic fiber, then the base of this elliptic Calabi--Yau threefold is
\begin{equation}
    B=\widetilde{(T^2\times T^2)/\mathbb{Z}_3}\,,
\end{equation}
where $B$ denotes the resulting surface as the resolution of the 9 fixed points in $(T_2\times T_2)/\mathbb{Z}_3$. Resolving these fixed points in the base via blowups introduces nine exceptional curves with the following intersections
\begin{equation}
    C_i.C_j=-3\delta_{ij}\,,
\end{equation}
where $i,j=1,\dots,9$ and $C_i\cong \mathbb{P}^1$. Then, the canonical class of the base is~\cite{Kim:2024eoa}
\begin{equation}
    -K_B=\frac13 \sum_{i=1}^9C_i\,.
\end{equation}
such that we have $-K_B.C_i=-1$.
From here, we can see that $h^{1,1}(B)=13$. In addition to these nine exceptional curves, there is also an infinite family of rational curve. To see this, let us consider expressing $T^2\times T^2=E_\omega\times E_\omega$ where $E_\omega$ is an elliptic curve with complex multiplication by Eisenstein integers $\mathbb{Z}[\omega]$.
With this notation, we can choose any such pair of Einstein integers with the slope $[a:b]\in \mathbb{P}^1(\mathbb{Z}[\omega])$ which would give an elliptic curve inside $E_\omega\times E_\omega$, namely
\begin{equation}
    E_{[a:b]}=\{(at,bt)\,\vert\, t\in E_\omega\}\subset E_\omega\times E_\omega\,.
\end{equation}
We can see that the diagonal $\mathbb{Z}_3$ acts on $E_\omega\times E_\omega$ preserves $E_{[a:b]}$ as $(\omega at,\omega bt)=(a(\omega t),b(\omega t))$. This induced action on $E_{[a:b]}\cong E_\omega$ is just $t\mapsto \omega t$ where this automorphism has order 3 and its fixed points satisfy $(\omega-1)t=0$. Therefore, the quotient map $E_\omega\to E_{\omega}/\langle\omega\rangle$ has degree 3, with three ramification points, each with ramification index 3. Then, by the Riemann--Hurwitz formula, we can obtain
\begin{equation}
    2g(E_\omega)-2=3(2g(E_\omega/\langle\omega\rangle)-2)+3(3-1)\,,
\end{equation}
where $g(E)$ denotes the genus of $E$.
As $g(E_\omega)=1$, we obtain $g(E_\omega/\langle\omega\rangle)=0$, implying $E_\omega/\langle\omega\rangle\cong\mathbb{P}^1$. Therefore, all $E_{[a:b]}/\langle\omega\rangle$ are rational curves~\cite[Ch. IV]{Hartshorne1977}.
We can further restrict to primitive pairs $\mathrm{gcd}(a,b)=1$ and nevertheless observe an infinite family of such curves which remain rational in $B$. Hence, there are infinitely many BPS generators of the effective cone. This can be explained via the large duality group $PGL(2,\mathbb{Z}[\omega])$ acting on $E_\omega\times E_\omega$ observed in~\cite{Kim:2024eoa}. 

The BPS brane cone, or equivalently, the effective cone of $B$, is generated as
\begin{equation}
    \cC_{\rm BPS-B}=M(B)=\overline{\mathrm{Cone}([C_1],\dots,[C_9],[E_{[a:b]}/\langle\omega\rangle])}\,.
\end{equation}
The BPS black brane cone is 
\begin{equation}
    \cC_{\rm BPS-BB}=\Nef(B)=\{J\in N^1(B)_{\mathbb{R}}\,\vert\, J.C_i\geq 0\,,J.(E_{[a:b]}/\langle\omega\rangle)\geq 0\,,\text{ for all }[a:b]\in \mathbb{P}^1(\mathbb{Z}[\omega])\}\,.
\end{equation}
Focusing on $-K_B$, we can see that it does not belong to $\cC_{\rm BPS-BB}$ as $-K_B.C_i<0$ and $(-K_B).(-K_B)<0$. Recall that $-K_B\in \mathrm{Pic}(B)$ and its first Chern class $c_1(B)\in H^2(B,\mathbb{Z})$. Hence, the divisor class $-K_B$ is an integral divisor class within the full charge lattice, but is not integral in the sub-lattice generated by the nine exceptional divisors. However, by the argument in~\cite{Kim:2024eoa} where we suppose to the contrary that $-K_B$ is effective, but we then immediately see that the supposedly constant pulled-back anti-canonical section along $T^2\times T^2$ is no longer invariant under actions of $\mathbb{Z}_3$. Hence, by contradiction, we arrive at $h^0(B,K_B)=0$ where there is no global section associated to $K_B$. Therefore, $-K_B$ is not effective despite being in the effective cone, but outside the nef cone.

As a last remark, we note that within $\cC_{\rm BPS-B}$, \cref{conj:BPS completeness} still remains true in this resolved model.

\section{M-theory on Calabi--Yau threefolds}
\label{sec:M-theory on Calabi--Yau threefold}

In this section, we consider M-theory compactified on Calabi-Yau threefolds leading to ${\cal N}=1$ supersymmetric theories in $d=5$.  We prove the duality between all BPS-BB and BPS-B cones (for both 1-branes and 0-branes).  We also review the proofs of completeness within the BPS-BB cone of BPS 1-branes in all smooth Calabi--Yau threefolds and within the BPS-BB cone of BPS 0-branes when the Calabi--Yau threefolds are complete--intersection Calabi--Yau (CICY) threefolds constructed as ample hypersurfaces in projective varieties. Lastly, we present diverse examples each illuminating distinct patterns among $\cC_{\rm BPS-BB}$ and $\cC_{\rm BPS-B}$. Nevertheless, all provide strong evidence for \cref{conj:BPS completeness}, \cref{conj:cone dual conjecture}, and the classification of BPS-Bs in \cref{sec:tension properties}.

Let us consider M-theory compactified on a Calabi--Yau threefold, which we denote as $X$. We expand the three-form gauge potential appearing in 11d M-theory as $C_3=A^I\wedge \omega_I$ where $\omega_I$ denotes a basis of $H^{1,1}(X)$.
Upon a dimensional reduction, the bosonic action of the resulting 5d $\cN=1$ quantum gravity theory takes on the following form
\begin{equation}
\label{eq:5d action}
    S_5=\frac{2\pi}{l_5^3}\int \dd^5 x\sqrt{-g}\bigg(R- G_{IJ}\partial_\mu t^I\partial^\mu t^J\bigg)-\frac{1}{2\pi l_5}\int G_{IJ}F^I\wedge \star F^J-\frac{1}{24\pi^2}\int C_{IJK}A^I\wedge F^J\wedge F^K\,.
\end{equation}
Here, $l_5$ denotes the five-dimensional Planck length, $I,J=1,\dots, h^{1,1}(X)$, and $C_{IJK}=\int_X[D_I]\wedge [D_J]\wedge [D_K]$ denotes the triple-intersection number of the divisor classes in $X$.
Choosing a basis of divisor classes $[D_I]\in H^2(X,\mathbb{Z})$, a K\"ahler two-form can then be expanded as $J=\sum_{I=1}^{h^{1,1}(X)}t^I[D_I]$ where $t^I$ are the K\"ahler parameters parameterizing the vector multiplet moduli space.
The prepotential, or equivalently the volume of $X$, is computed as
\begin{equation}
\label{eq:const Calabi--Yau threefold}
    \cF=\frac16 C_{IJK}t^It^Jt^K\,.
\end{equation}
The metric on the K\"ahler cone is
\begin{equation}
    G_{IJ}=-\frac12\partial_I\partial_J\log \cF\,,
\end{equation}
from which we obtain the pullback metric on the constant-volume slice ($\cF=1$),
\begin{equation}
    g_{ij}=G_{IJ}\partial_it^I\partial_jt^J\,,  
\end{equation}
where $\partial_i:=\partial/\partial\phi^i$ and $\phi^i$ with $i=1,\dots, h^{1,1}-1$ denote the scalars parameterizing the co-dimension one slice $\cF=1$ in the K\"ahler cone. Furthermore, in the physical regimes of the moduli space, the metric must remain positive definite~\cite{Alim:2021vhs}.

The BPS spectrum of the 5d $\cN=1$ quantum gravity theory contains BPS particles, obtained from M2-branes wrapping effective 2-cycles in $X$, and BPS strings, obtained from M5-branes wrapping effective 4-cycles in $X$. The exact structure of the charge lattices can be determined by geometry from a top-down perspective as we show in \cref{sec:Mtheory on Calabi--Yau threefold cones}. The mass of an electric particle of charge $q$ obtained from an M2-brane wrapping a 2-cycle $C_q$ satisfies
\begin{equation}
    m(q)=\volume(C_q)\geq \int_{[C_q]}J= q_It^I=:|Z_e(q)|\,.
\end{equation}
The tension of a magnetic string of charge $p$ obtained from wrapping a 4-cycle $D_p$ satisfies
\begin{equation}
    T(p)=\volume(D_p)\geq \frac12\int_{[D_p]}J\wedge J= \frac12C_{IJK}p^It^Jt^K=:|Z_m(p)|\,.
\end{equation}
In both cases, the inequalities are the known relation between tensions and central charges for $p$-branes and it is saturated only when the state is BPS. From the geometric perspective, the inequality is saturated when the corresponding curve or divisor is effective.

\subsection{The attractor mechanism}

Black holes and black strings in this theory, along with their associated geometric cones, have been studied in~\cite{Maldacena:1997de,Katz:2020ewz,Long:2021lon,Alim:2021vhs,Gendler:2022ztv}. The attractor mechanism for BPS black branes has been reviewed in e.g.,~\cite{Long:2021lon}. Here, for completeness, we will review the essential ingredients for our analysis.

\subsubsection*{BPS black 0-branes}

For black 0-branes, we consider the following spherically symmetric ans\"atz
\begin{equation}
    \dd s_5^2=-e^{-4U(r)}\dd t^2+e^{2U(r)}\big(\dd r^2+r^2\dd \Omega_3^2\big)\,.
\end{equation}
The one-dimensional effective Lagrangian is
\begin{equation}
    \cL=(\partial_rU)^2+\frac13 G_{IJ}\partial_rt^I\partial_rt^J+\frac{e^{-2U}}{r^6}V_{\rm BH}(t,q)+\dots \,.
\end{equation}
When the particle is supersymmetric, the equations of motion reduce to the first-order attractor flow equations
\begin{equation}
    \partial_r U=-\frac{e^{2U}}{3r^3}Z_e\,,\qquad \partial_r \phi^i=-\frac{e^{-2U}}{r^3}g^{ij}\partial_jZ_e\,.
\end{equation}
Imposing regularity at the extremal horizon requires $\phi^i\to \phi^i_*$ as $r\to 0$. Consequently, $\partial_r\phi^i\to 0$ near the horizon. Then, from the equations of motion of the scalar fields, we obtain the BPS attractor equation and its solution
\begin{equation}
\label{eq:att Ze}
    \partial_iZ_e=0\,,\qquad \Rightarrow \qquad q_I=\frac{1}{3\cF}\tau_IZ_e\,,
\end{equation}
where $\tau_I=\frac12C_{IJK}t^Jt^K$.  Evaluated at the attractor point $t_*$, the entropy of the BPS black 0-brane is $S_{\rm BPS}=\frac{16\pi^4}{l_5^3}\big(\frac13 Z_e\big\vert_{t=t_*}\big)^{3/2}$.
More generally, without imposing supersymmetry, the effective black 0-brane potential is
\begin{equation}
    V_{\rm BH}:=G^{IJ}q_Iq_J=\frac{2}{3}Z_e^2+g^{ij}\partial_iZ_e\partial_jZ_e\,,
\end{equation}
where the extremal attractor equation becomes
\begin{equation}
\label{eq:att Vbh 5d}
    D_IV_{\rm BH}=\bigg(\partial_I-\frac{2}{3\cF}\tau_I\bigg)V_{\rm BH}=0\,.
\end{equation}
Thus, in the supersymmetric settings, any electric charge state satisfying \cref{eq:att Ze} necessarily satisfies \cref{eq:att Vbh 5d}. The converse, however, need not hold. 

If such an attractor solution exists, then the tension of the BPS black 0-brane with charge $q$ can be determined via the attractor mechanism as
\begin{equation}
    T^{\rm Ext}(q)=\sqrt{\frac{3}{2}V_{\rm BH}}\bigg\vert_{t=t_*}\,,
\end{equation}
where $t_*$ is the attractor point.
In particular, non-supersymmetric extremal solutions that do not satisfy the BPS condition can still occur, and have been discussed in~\cite{Long:2021lon}. In these cases, we have the strict inequality 
\begin{equation}
\label{eq:BPS predict}
    \big\vert Z_{\rm min}(q)\big\vert< T(q)< T^{\rm Ext}(q)\,,
\end{equation}
where $Z_{\rm min}$ denotes the minimum of the central charge across all moduli space.

\subsubsection*{BPS black 1-branes}

For black 1-branes, the relevant static, translationally invariant ans\"atz is
\begin{equation}
    \dd s_5^2=e^{-U(r)}\big(-\dd t^2+\dd y^2\big)+e^{2U(r)}\big(\dd r^2+r^2\dd\Omega_2\big)\,,
\end{equation}
where the magnetic field strengths are parameterized as $F^I=p^I\sin\theta \,\dd\theta\wedge \dd\phi$. 
The first-order attractor flow equations are
\begin{equation}
    \partial_r U=-\frac{e^{-U}}{r^2}Z_m\,,\qquad \partial_r \phi^i=\frac{e^{-U}}{r^2}g^{ij}\partial_jZ_m\,.
\end{equation}
A supersymmetric solution with a regular horizon therefore satisfies 
\begin{equation}
\label{eq:att Zm 5d}
    \partial_iZ_m=0\,.
\end{equation}
On the constant-volume slice $\cF=1$, this gives $p^I=\frac{Z_m}{3}t^I$.
The magnetic black-string potential is
\begin{equation}
    V_{\rm BS}=4G_{IJ}p^Ip^J=\frac23 Z_m^2+ g^{ij}\partial_iZ_m\partial_jZ_m\,,
\end{equation}
such that the attractor equation is
\begin{equation}
\label{eq:att str 5d}
    D_IV_{\rm BS}=\bigg(\partial_I-\frac{4}{3\cF}\tau_i\bigg)V_{\rm BS}=0\,.
\end{equation}
Again, in the supersymmetric case, \cref{eq:att Zm 5d} necessarily implies \cref{eq:att str 5d}. 

The extremal tension of the black string is then
\begin{equation}
    T^{\rm Ext}=\sqrt{\frac32 V_{\rm BS}}\bigg\vert_{t=t_*}\,,
\end{equation}
where $t_*$ is the attractor point. Similar to the black 0-branes, non-supersymmetric extremal 1-brane solutions which do not satisfy the BPS condition can occur and the strict inequality \cref{eq:BPS predict} again holds for these non-BPS black 1-branes.

\subsection{Charge lattices and cones}
\label{sec:Mtheory on Calabi--Yau threefold cones}

\begin{figure}
    \centering

\tikzset{every picture/.style={line width=0.75pt}} 

\begin{tikzpicture}[x=0.75pt,y=0.75pt,yscale=-1,xscale=1]

\draw [color={rgb, 255:red, 208; green, 2; blue, 27 }  ,draw opacity=0.34 ][fill={rgb, 255:red, 208; green, 2; blue, 27 }  ,fill opacity=0.34 ]   (37.19,88.02) .. controls (63.92,76.21) and (112.92,60.47) .. (164.6,58.5) .. controls (216.28,56.53) and (257.26,71.29) .. (307.16,86.05) ;
\draw [color={rgb, 255:red, 74; green, 144; blue, 226 }  ,draw opacity=0.34 ][fill={rgb, 255:red, 74; green, 144; blue, 226 }  ,fill opacity=0.34 ]   (90.65,70.31) .. controls (116.49,63.42) and (128.07,59.49) .. (164.6,58.5) .. controls (201.13,57.52) and (222.51,62.44) .. (251.03,70.31) ;
\draw [color={rgb, 255:red, 208; green, 2; blue, 27 }  ,draw opacity=1 ]   (177.96,273) -- (7.21,48.29) ;
\draw [shift={(6,46.69)}, rotate = 52.77] [color={rgb, 255:red, 208; green, 2; blue, 27 }  ,draw opacity=1 ][line width=0.75]    (10.93,-3.29) .. controls (6.95,-1.4) and (3.31,-0.3) .. (0,0) .. controls (3.31,0.3) and (6.95,1.4) .. (10.93,3.29)   ;
\draw [color={rgb, 255:red, 208; green, 2; blue, 27 }  ,draw opacity=1 ]   (177.96,273) -- (331.86,50.31) ;
\draw [shift={(333,48.66)}, rotate = 124.65] [color={rgb, 255:red, 208; green, 2; blue, 27 }  ,draw opacity=1 ][line width=0.75]    (10.93,-3.29) .. controls (6.95,-1.4) and (3.31,-0.3) .. (0,0) .. controls (3.31,0.3) and (6.95,1.4) .. (10.93,3.29)   ;
\draw  [color={rgb, 255:red, 208; green, 2; blue, 27 }  ,draw opacity=0 ][fill={rgb, 255:red, 208; green, 2; blue, 27 }  ,fill opacity=0.34 ] (307.16,86.05) -- (177.96,273) -- (37.19,88.02) -- cycle ;
\draw [color={rgb, 255:red, 74; green, 144; blue, 226 }  ,draw opacity=1 ]   (177.96,273) -- (72.73,31.8) ;
\draw [shift={(71.93,29.97)}, rotate = 66.43] [color={rgb, 255:red, 74; green, 144; blue, 226 }  ,draw opacity=1 ][line width=0.75]    (10.93,-3.29) .. controls (6.95,-1.4) and (3.31,-0.3) .. (0,0) .. controls (3.31,0.3) and (6.95,1.4) .. (10.93,3.29)   ;
\draw [color={rgb, 255:red, 74; green, 144; blue, 226 }  ,draw opacity=1 ][fill={rgb, 255:red, 74; green, 144; blue, 226 }  ,fill opacity=1 ]   (177.96,273) -- (265.49,33.81) ;
\draw [shift={(266.17,31.94)}, rotate = 110.1] [color={rgb, 255:red, 74; green, 144; blue, 226 }  ,draw opacity=1 ][line width=0.75]    (10.93,-3.29) .. controls (6.95,-1.4) and (3.31,-0.3) .. (0,0) .. controls (3.31,0.3) and (6.95,1.4) .. (10.93,3.29)   ;
\draw  [color={rgb, 255:red, 74; green, 144; blue, 226 }  ,draw opacity=0 ][fill={rgb, 255:red, 74; green, 144; blue, 226 }  ,fill opacity=0.34 ] (177.96,273) -- (90.65,70.31) -- (251.03,70.31) -- cycle ;
\draw [color={rgb, 255:red, 208; green, 2; blue, 27 }  ,draw opacity=0.34 ][fill={rgb, 255:red, 208; green, 2; blue, 27 }  ,fill opacity=0.34 ]   (369.66,90.37) .. controls (395.08,78.33) and (441.69,62.27) .. (490.84,60.27) .. controls (539.99,58.26) and (578.97,73.31) .. (626.43,88.36) ;
\draw [color={rgb, 255:red, 74; green, 144; blue, 226 }  ,draw opacity=0.34 ][fill={rgb, 255:red, 74; green, 144; blue, 226 }  ,fill opacity=0.34 ]   (420.5,72.31) .. controls (445.08,65.28) and (456.1,61.27) .. (490.84,60.27) .. controls (525.58,59.26) and (545.92,64.28) .. (573.04,72.31) ;
\draw [color={rgb, 255:red, 208; green, 2; blue, 27 }  ,draw opacity=1 ]   (503.55,279) -- (341.16,49.86) ;
\draw [shift={(340,48.23)}, rotate = 54.67] [color={rgb, 255:red, 208; green, 2; blue, 27 }  ,draw opacity=1 ][line width=0.75]    (10.93,-3.29) .. controls (6.95,-1.4) and (3.31,-0.3) .. (0,0) .. controls (3.31,0.3) and (6.95,1.4) .. (10.93,3.29)   ;
\draw [color={rgb, 255:red, 208; green, 2; blue, 27 }  ,draw opacity=1 ]   (503.55,279) -- (649.92,51.92) ;
\draw [shift={(651,50.23)}, rotate = 122.8] [color={rgb, 255:red, 208; green, 2; blue, 27 }  ,draw opacity=1 ][line width=0.75]    (10.93,-3.29) .. controls (6.95,-1.4) and (3.31,-0.3) .. (0,0) .. controls (3.31,0.3) and (6.95,1.4) .. (10.93,3.29)   ;
\draw  [color={rgb, 255:red, 208; green, 2; blue, 27 }  ,draw opacity=0 ][fill={rgb, 255:red, 208; green, 2; blue, 27 }  ,fill opacity=0.34 ] (626.43,88.36) -- (503.55,279) -- (369.66,90.37) -- cycle ;
\draw [color={rgb, 255:red, 74; green, 144; blue, 226 }  ,draw opacity=1 ]   (503.55,279) -- (403.46,33.02) ;
\draw [shift={(402.71,31.17)}, rotate = 67.86] [color={rgb, 255:red, 74; green, 144; blue, 226 }  ,draw opacity=1 ][line width=0.75]    (10.93,-3.29) .. controls (6.95,-1.4) and (3.31,-0.3) .. (0,0) .. controls (3.31,0.3) and (6.95,1.4) .. (10.93,3.29)   ;
\draw [color={rgb, 255:red, 74; green, 144; blue, 226 }  ,draw opacity=1 ][fill={rgb, 255:red, 74; green, 144; blue, 226 }  ,fill opacity=1 ]   (503.55,279) -- (586.8,35.07) ;
\draw [shift={(587.44,33.18)}, rotate = 108.84] [color={rgb, 255:red, 74; green, 144; blue, 226 }  ,draw opacity=1 ][line width=0.75]    (10.93,-3.29) .. controls (6.95,-1.4) and (3.31,-0.3) .. (0,0) .. controls (3.31,0.3) and (6.95,1.4) .. (10.93,3.29)   ;
\draw  [color={rgb, 255:red, 74; green, 144; blue, 226 }  ,draw opacity=0 ][fill={rgb, 255:red, 74; green, 144; blue, 226 }  ,fill opacity=0.34 ] (503.55,279) -- (420.5,72.31) -- (573.04,72.31) -- cycle ;

\draw (120.07,30.24) node [anchor=north west][inner sep=0.75pt]  [xscale=0.75,yscale=0.75]  {$\textcolor[rgb]{0.29,0.56,0.89}{\cC_{\rm BPS-BB}^{(1)}=\mathrm{Nef}(X)}$};
\draw (10,10.9) node [anchor=north west][inner sep=0.75pt]  [color={rgb, 255:red, 208; green, 2; blue, 27 }  ,opacity=1 ,xscale=0.75,yscale=0.75]  {$\cC_{\rm BPS-B}^{(1)}=\cE(X)$};
\draw (454.7,31.6) node [anchor=north west][inner sep=0.75pt]  [xscale=0.75,yscale=0.75]  {$\textcolor[rgb]{0.29,0.56,0.89}{\cC_{\rm BPS-BB}^{(0)}=\mathrm{Mov}(X)}$};
\draw (591.87,10.67) node [anchor=north west][inner sep=0.75pt]  [color={rgb, 255:red, 208; green, 2; blue, 27 }  ,opacity=1 ,xscale=0.75,yscale=0.75]  {$\cC_{\rm BPS-BB}^{(0)}=M(X)$};

\draw (122.07,285.67) node [anchor=north west][inner sep=0.75pt]  [xscale=0.75,yscale=0.75]  {Cones of BPS 1-branes};
\draw (445.7,285.67) node [anchor=north west][inner sep=0.75pt]  [xscale=0.75,yscale=0.75]  {Cones of BPS 0-branes};

\end{tikzpicture}
\vspace*{0.cm}
    \caption{A schematic illustration of the BPS-B and BPS-BB cones associated to BPS 1-branes and BPS 0-branes in M-theory compactified on a Calabi--Yau threefold $X$.}
    \label{fig:Cones in M-theory on Calabi--Yau threefolds}
\end{figure}
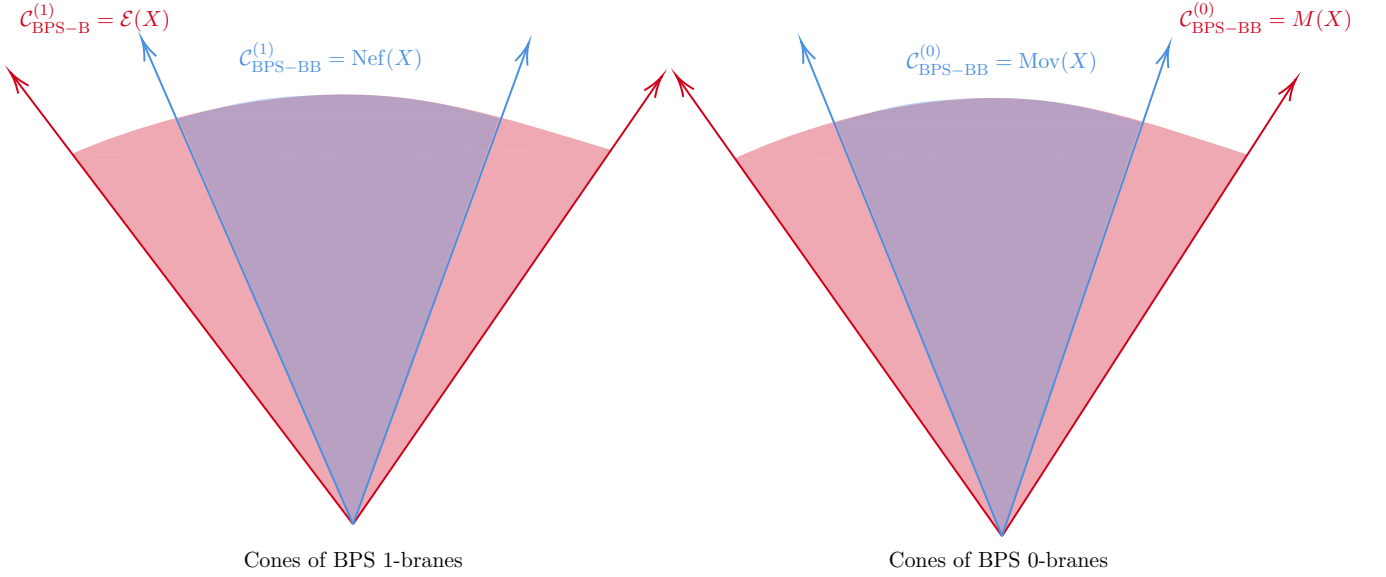

As reviewed above, the BPS objects in these theories consist of electric BPS 0-branes and magnetic BPS 1-branes. From a top-down perspective, the physics of particles and strings is encoded in the geometric properties of the associated 2-cycles and 4-cycles that are wrapped by the M2-branes and the M5-branes, respectively. See, for example,~\cite{Witten:1996qb,Katz:2020ewz,Long:2021lon,Alim:2021vhs,Gendler:2022ztv}.

From a geometric perspective, for an M5-brane wrapping a divisor $D\subset X$ to give a BPS 1-brane, the divisor class must be \textit{effective}. When the divisor class admits a holomorphic effective representative, the volume of the divisor is calibrated by the K\"ahler form and the BPS relation between the tension and the central charge of the magnetic string is saturated. Equivalently, a divisor class $[D]$ is effective in $X$ if and only if
\begin{equation}
    h^0(X,\cO_X(D))>0\,.
\end{equation}
Hence, the relevant cone for the BPS 1-brane charges is the \textit{effective cone} of divisors, $\cE(X)$. In general, determining all generators of the effective cone of a Calabi--Yau threefold is difficult.
For example, when $X$ is a toric Calabi--Yau threefold constructed as a hypersurface inside a toric fourfold, the toric effective cone $\cE_{\rm toric}(X)\subseteq \cE(X)$ is generated by the classes of prime toric divisors in $X$ which are prime divisors in the toric fourfold restricted onto $X$. However, there can exist \textit{autochthonous divisors} that lie outside of the toric effective cone, yet are effective in $X$~\cite{Demirtas:2018akl}.

Similar to the 6d scenario above, the cone of BPS black 1-brane charges is the \textit{ample cone}, or equivalently, the \textit{K\"ahler cone}, $K(X)\subset H^{1,1}(X)\cap H^2(X,\mathbb{R})$.
For an integral divisor class divisor $[D]\in K(X)$, $[D]$ is ample if and only if it satisfies
\begin{equation}
    D^3>0\,,\qquad D^2.D_I>0\,,\qquad D.C>0\,.
\end{equation}
for all effective divisors $D_I$ and effective curves $C$ in $X$.

At last, the closure of the K\"ahler cone is the \textit{nef cone}. Namely, for a divisor $[D]\in H_4(X,\mathbb{R})$ to be numerically eventually free (nef), it must satisfy
\begin{equation}
    D^3\geq 0\,,\qquad D.D_I\geq 0\,,\qquad D.C\geq 0\,.
\end{equation}
for all effective divisors $D_I$ and curves $C$ in $X$. 
When the extremal rays of the nef cone are non-big nef divisors, namely $D^3=0$, the boundary ray generated by this non-big nef divisor class in $\partial\cC_{\rm BPS-BB}$ coincides with a boundary of $\cC_{\rm BPS-B}$. Then, in addition to signaling the appearance of massless/tensionless states leading to the emergent string conjecture~\cite{Lee:2019wij}, wrapping such divisors with M5-branes correspond to supergravity strings in the 5d theory~\cite{Katz:2020ewz}.

The relevant cone for BPS 0-brane charges is the \textit{Mori cone} $M(X)$, which is generated by effective curve classes. In other words, the Mori cone of $X$ is
\begin{equation}
    M(X)=\bigg\{\sum_i a_i[C_i]\bigg\vert a_i\in\mathbb{R}_{\geq 0}\,,[C_i]\in H_2(X,\mathbb{Z})\text{ is an effective curve class}\bigg\}\,.
\end{equation}
In particular, by the definition of nef, the nef cone is dual to the Mori cone, $M(X)=\Nef(X)^\vee$.
For BPS black 0-branes, the relevant cone in this geometric setting is the \textit{cone of movable curves}. Namely, a curve class is \textit{movable} if it is the pushforward from a projective birational model $X'\to X$ of the class of a complete intersection of ample divisors in $X'$~\cite{lazarsfeld2017positivity}. Indeed, the attractor solution for a BPS black 0-brane, given in \cref{eq:att Ze}, relates the BPS black 0-brane charge to $\tau_I=\frac12 C_{IJK}t^Jt^K$. At a regular attractor point $t_*\in K(X)$, this corresponds to the curve class obtained as the intersection of ample divisor classes, $J_*^2$. Therefore, any electric charge admitting a regular BPS black 0-brane attractor lies in the cone generated by intersections of ample divisor classes. The closure of this cone is the cone of movable curves.

The dictionary between $\cC_{\rm BPS-B}$ and $\cC_{\rm BPS-BB}$ and their associated geometric definitions in M-theory compactified on a Calabi--Yau threefold is summarized in \cref{fig:Cones in M-theory on Calabi--Yau threefolds}.

With the geometric definitions of the relevant curve classes and divisor classes associated to $\cC_{\rm BPS-B}$ and $\cC_{\rm BPS-BB}$, we obtain the duality among cones, $\cC_{\rm BPS-B}=\cC_{\rm BPS-BB}^\vee$. However, from the 5d $\cN=1$ supergravity perspective, this duality can also be derived from the attractor mechanism, as discussed in \cref{sec:duality between cones}. For a BPS black 1-brane of magnetic charge $p^I$, the attractor solution is $t_*^I=p^I/Z_{m,*}$ where $Z_{m,*}$ is also evaluated at the attractor point. Evaluating the electric central charge of a BPS 0-brane at the attractor point give $Z_{e}(q)\vert_{t=t_*}=q_It^I_*=\frac{3q_Ip^I}{Z_{m,*}}$. Rearranging this, we have
\begin{equation}
    \frac13Z_{m,*}Z_{e}(q)\vert_{t=t_*}=q_Ip^I\geq 0\,.
\end{equation}
The same reasoning, with the roles of 0-branes and 1-branes exchanged, gives the corresponding duality between BPS black 0-branes and BPS 1-branes.

\subsection{BPS completeness in black cones}
\label{sec:proof in Calabi--Yau threefold}

For 5d BPS black 0-branes, we don't have a general proof for M-theory compactified on a Calabi--Yau threefold $X$. However, if $X$ is constructed as a complete intersection in an ambient toric variety $V$, and if a BPS black hole charge comes from a movable curve class in $V$, we can construct BPS representatives using intersections of toric line bundles. Skauli~\cite{skauli} showed that this is the case if all the line bundles used to define the CICY threefolds are ample. Hence, the following theorem is true for this case.
\begin{theorem}
    For all smooth CICY threefolds $X$ constructed as an ample hypersurface inside a smooth projective toric variety, then all movable curves in $X$ are irreducible curves.
\end{theorem}
\noindent In this case, all the relevant cones are simplicial and we give a brief sketch of the proof. The proof starts from the Lefschetz hyperplane theorem~\cite[Th. 3.1.17]{Lazarsfeld2004}. For a smooth ample CICY threefold $X \subset V$, the inclusion $X \hookrightarrow V$ identifies the relevant curve lattices, so one can first understand curve classes in the ambient smooth projective toric variety $V$. On $V$, we know that the effective curve semigroup is generated by special primitive classes called contractible classes. These are not arbitrary curve classes: each one comes from an equivariant toric contraction, hence from the toric Minimal Model Program (MMP). For physicists, this is closely analogous to the phase structure of a two-dimensional $\mathcal N=(2,2)$ NLSM/GLSM: varying the Kähler/FI parameters moves toward phases where certain curve classes shrink, producing contractions or birational transitions. The “standard generators” of the Mori cone are precisely the curve classes selected by the ambient toric MMP, or equivalently by the ambient toric phase structure.

The main issue is to show that these ambient generators are still represented by curves lying inside $X$. Fix a contractible class $[\beta]\in H_2(V, \mathbb{Z})$. The corresponding toric contraction has fibers containing lines of class $\beta$, and these lines are parametrized by a relative Grassmannian, let us denote by $Gr$. The defining equations of $X$ restricted to each such line, produce a section
$
s_X \in H^0(Gr,M_{X,\beta})
$
of a natural vector bundle $M_{X,\beta}$ over $Gr$. A zero of this section is a point $[\ell]\in Gr$, hence a particular line $\ell$, such that all defining equations of $X$ vanish identically on $\ell$; equivalently, $\ell \subset X$. If $M_{X,\beta}$ has rank $r$ and $s_X$ were nowhere zero, then $s_X$ would define an injection
$
\mathcal O_{Gr} \hookrightarrow M_{X,\beta}.
$
Thus the quotient $M_{X,\beta}/\mathcal O_{Gr}$ would have rank $r-1$, and by multiplicativity of Chern classes, this would force
$
c_r(M_{X,\beta})=0.
$
Skauli calculated that
$
c_r(M_{X,\beta})\neq 0,
$
Therefore, $s_X$ cannot be nowhere zero. That zero corresponds to one line $\ell$ on which the complete-intersection equations vanish everywhere, so $X$ contains a rational curve in the ambient contractible class $\beta$. Since Lefschetz identifies the curve lattices of $X$ and $V$, and the ambient effective semigroup is generated by these contractible classes, the Mori semigroup of $X$ is generated by rational curves coming from the standard toric MMP generators. This finishes the proof of \cref{conj:BPS completeness} for M-theory on CICY threefolds defined by ample hypersurfaces.

For BPS black 1-branes in the strict interior of $\cC_{\rm BPS-BB}^{(1)}$, a relevant version of \cref{conj:BPS completeness} is realized by the following geometric theorem.
\begin{theorem}
    For every smooth Calabi--Yau threefold $X$, every ample divisor class on $X$ is effective.
\end{theorem}
\noindent This corresponds to the BPS occupancy within the strict interior of the BPS black 1-brane cone in M-theory compactified on Calabi--Yau threefolds. We now review the proof of effectiveness for ample divisor classes, following~\cite{Katz:2020ewz,Gendler:2026uux}. Consider the 2d $(0,4)$ string worldsheet theory from an M5-brane wrapping a divisor $D$ in a Calabi--Yau threefold $X$. In this case, the transverse rotation group of the string is $SO(3)_N$, whose double cover $SU(2)_R$ is identified with the right-moving $R$-symmetry of the (0,4) worldsheet theory. We denote the second Chern class of the $SU(2)_R$ bundle by $c_2(R)$. In addition to the Chern--Simons term in \cref{eq:5d action}, there is also a mixed gauge-gravitational Chern--Simons term in the action. For convenience of the anomaly-inflow computation, we pass from the supergravity normalization in \cref{eq:5d action} to the standard topological normalization in which the gauge fields have integral periods and the relevant Chern--Simons couplings appear as
\begin{equation}
    S_5\supset -\frac{1}{6}\int C_{IJK}A^I\wedge F^J\wedge F^K-\frac{1}{48}\int c_{2,I}A^I\wedge p_1(TM_5)\,,
\end{equation}
where $c_{2,I}=\int_{D_I}c_2(X)$ and $TM_5$ denotes the tangent bundle to $M_5$. In the presence of a string with magnetic charge $p^I$, the Bianchi identity becomes $\dd F^I=2\pi p^I\delta_3(\Sigma)$ where $\Sigma$ denotes the 2d string worldsheet. Then, the anomaly inflow induced by the bulk Chern--Simons terms for an M5-brane wrapping a divisor $D=p^ID_I$ in $X$ gives the following four-form anomaly polynomial on the string worldsheet
\begin{equation}
    I_4=\frac12 C_{IJK}p^KF^IF^J+\bigg(\frac16 D^3+\frac{1}{12}c_2(X).D\bigg)c_2(R)-\bigg(\frac{1}{48}c_2(X).D\bigg)p_1(T\Sigma)+\dots\,,
\end{equation}
where $\dots$ denotes possible flavor and mixed anomaly terms. Comparing to the general form of the anomaly polynomial given in \cref{eq:2d anomaly polynomial}, we can read off the right-moving and left-moving central charges of the worldsheet theory. In particular, we obtain~\cite{Maldacena:1997de}
\begin{equation}
    c_R=D^3+\frac12c_2(X).D\,,\qquad c_L=D^3+c_2(X).D\,.
\end{equation}
The corresponding current algebra levels are~\cite{Katz:2020ewz}
\begin{equation}
    k_{IJ}=C_{IJK}p^K\,,\qquad k_R=\frac16 D^3+\frac{1}{12}c_2(X).D\,.
\end{equation}
Having identified these worldsheet quantities, we compactifying the theory on $S^1$ and consider the string wrapping the spatial circle with quantized momenta $n$. In the Cardy regime $n\gg c_L$, the entropy of the resulting 4d BPS black hole is
\begin{equation}
\label{eq:5d MSW entropy}
    S=2\pi \sqrt{\frac{c_Ln}{6}}=2\pi\sqrt{n\bigg(\frac16 D^3+\frac16 c_2(X).D\bigg)}\,.
\end{equation}
We can express the above entropy in terms of the holomorphic Euler characteristic of the line bundle $\cO_X(D)$. By the Hirzebruch--Riemann--Roch theorem applied to Calabi--Yau threefolds, we have
\begin{equation}
    \chi_h(X,\cO_X(D))=\frac16 D^3+\frac{1}{12}c_2(X).D\,.
\end{equation}
Hence, \cref{eq:5d MSW entropy} can be rewritten as
\begin{equation}
    S=2\pi \sqrt{n\bigg[\chi_h(X,\cO_X(D))+\frac{1}{12}c_2(X).D\bigg]}\,.
\end{equation}
Therefore, the entropy of the 4d BPS black hole is naturally sensitive to the holomorphic Euler characteristic of the line bundle $\cO_X(D)$. The positivity of the entropy of BPS black holes then translates to $\chi_h(X,\cO_X(D))+\frac{1}{12}c_2(X).D>0$. This inequality is not sufficient in showing that the divisor class $[D]$ is effective. However, in what follows, we will argue that with additional geometric inputs, the divisor classes associated to BPS black 1-branes with charges in the strict interior of $\cC_{\rm BPS-BB}^{(1)}$ in 5d are indeed effective.

From the geometric perspective, on $X$, each string charge class $D\in H^2(X,\mathbb{Z})$ determines a line bundle $\cO_X(D)$, and the existence of a BPS string within the charge class is equivalent to the existence of a nonzero section of the line bundle. When $D$ is a BPS black string charge in the strict interior of $\cC_{\rm BPS-BB}$, we have $D^3>0$ and $c_2(X).D\geq0$~\cite{miyaoka1987chern}. Hence, $\chi_h(X,\cO_X(D))>0$. By the Kodaira--Viehweg vanishing theorem~\cite{kawamata1982generalization,viehweg1982vanishing}, $h^i(X,\cO_X(D))=0$ for $i>0$, hence we have $h^0(X,\cO_X(D))>0$ and there exists a BPS string within the charge class $[D]$.

It is worth noting that the corresponding statement for nef divisor classes is stronger and remains conjectural in general for smooth Calabi--Yau threefolds. Such a statement would extend the above theorem for ample divisor to include its boundary, and is also closely related to the non-vanishing and abundance conjectures~\cite{lazic2020generalised,lazarsfeld2017positivity,Katz:2020ewz}.
\begin{conj}
    For every smooth Calabi--Yau threefold $X$, every non-zero nef divisor class $D$ is effective.
\end{conj}

\subsection{Examples}
\label{sec:5d examples}

The most natural relation between the BPS-B and BPS-BB cones is $\cC_{\rm BPS-BB}\subseteq \cC_{\rm BPS-B}$, which we have argued for in \cref{sec:duality between cones}. There are two scenarios that can occur and we highlight their physics through the example geometries studied below.
\begin{enumerate}
    \item The first scenario is when $\cC_{\rm BPS-BB}=\cC_{\rm BPS-B}$. In this case, \cref{conj:BPS completeness} directly implies that BPS completeness holds in the BPS-B cone. This is illustrated via the example in \cref{sec:bicubic}.
    \item The second scenario is when the inclusion is strict, $\cC_{\rm BPS-BB}\subset \cC_{\rm BPS-B}$. In this case, we encounter two further possibilities:
    \begin{itemize}
        \item the theory exhibits BPS completeness in the BPS-B cone. This is illustrated in \cref{sec:CICY};
        \item the theory exhibits BPS completeness in the BPS-BB cone, while there is a non-BPS semigroup contained within the BPS-B cone. The non-BPS holes in the BPS-B cone can either appear strictly on the boundary, as illustrated in \cref{sec:Calabi--Yau threefold with boundary holes}, or appear in the interior at higher $h^{1,1}$'s (from toric Calabi--Yau threefolds in the Kreuzer--Skarke database), as illustrated in \cref{sec:Calabi--Yau threefold with interior holes}. It is worth noting that the existence of these non-BPS holes as well as the mechanism by which they are found, is similar to those in \cref{sec:6d holes}. In particular, they are all distinct from the apparent holes found in \cref{sec:Enriques}, where the seemingly non-effective classes arise from twisting the relevant line bundle by the torsion canonical class and are projected out upon identifying the correct charge lattice.
    \end{itemize}
\end{enumerate}
Regardless of whether BPS completeness holds in the full BPS-B cone, the BPS-B and BPS-BB cones always exhibit the expected dualities in all examples. Furthermore, the tension of BPS 1-branes are remarkably consistent with the general properties discussed in \cref{sec:tension properties}. We also observe that when there is a non-big nef divisor present in $H_4(X)$, the associated BPS 1-brane obtained from an M5-brane wrapping this divisor is along the wall shared by both $\cC_{\rm BPS-BB}^{(1)}$ and $\cC_{\rm BPS-B}^{(1)}$ and hence its tension vanishes in the infinite-distance limit of moduli space with a decay rate of $e^{-\alpha \Delta}$ expected from the distance conjecture~\cite{Ooguri:2006in} where $\Delta$ denotes the geodesic distance traversed in moduli space and $\alpha\sim \cO(1)$. Lastly, the minimum tension of the primitive generators of $\cC_{\rm BPS-B}$ and $\cC_{\rm BPS-BB}$ satisfy the elementary constituent conjecture~\cite{Nevoa:2025xiq}.

To verify \cref{conj:BPS completeness} in all these example geometries we study, we compute the genus-0 Gopakumar--Vafa (GV) invariants of curve classes~\cite{Gopakumar:1998jq} and $h^0(X,\cO_X(D))$ for divisor classes to determine their effectiveness.\footnote{As emphasized in~\cite{Gendler:2022ztv}, a non-zero genus-0 GV invariant for a curve class is a sufficient but not necessary condition for the resulting 0-brane from M2-branes wrapping this curve to be supersymmetric.} Among the studied geometries in this section, we observe another refinement of \cref{conj:BPS completeness} in the BPS black 0-brane cone. Namely, not only are all integral sites within the BPS black 0-brane cone occupied by BPS states, the corresponding genus-zero GV invariants are positive. In the geometric language, this observation leads to the following conjecture.
\begin{conj}
\label{conj:GV BPS-BB}
    Every integral curve classes in the cone of movable curves for a smooth Calabi--Yau threefold has positive genus-zero GV invariant.
\end{conj}

\subsubsection{Bicubic in $\mathbb{P}^2\times \mathbb{P}^2$}
\label{sec:bicubic}

Let us first consider a Calabi--Yau threefold $X$ constructed as an anti-canonical hypersurface in the ambient space $V=\mathbb{P}^2\times \mathbb{P}^2$ as in e.g.,~\cite{Long:2021lon}. This bicubic takes the configuration
\begin{equation}
    X_{\rm bicubic}=
    \left(
    \begin{array}{c|c}
    \mathbb{P}^2 & 3 \\
    \mathbb{P}^2 & 3
    \end{array}
    \right)\,.
\end{equation}
This geometry and its associated cones have been studied in~\cite{ottem2015birational}. A convenient basis of divisors is given by the pullbacks of the hyperplane classes onto the Calabi--Yau threefold, $D_i=\pi_i^* H$ where $\pi_i:X\to \mathbb{P}^2$ denotes the projection to the $i$-th factor and $i=1,2$. Then, the intersection data is
\begin{equation}
    D_1.D_1.D_1=D_2.D_2.D_2=0\,,\qquad D_1.D_1.D_2=D_1.D_2.D_2=3\,.
\end{equation}
We can express the K\"ahler 2-form in this basis as $J=t^1[D_1]+t^2[D_2]$.
Then, the prepotential in this basis is
\begin{equation}
    \cF=\frac12 t^1t^2(t^1+t^2)\,.
\end{equation}
and the metric is simply
\begin{equation}
    G_{IJ}=\frac12\begin{pmatrix}
        \frac{1}{(t^1)^2}+\frac{1}{(t^1+t^2)^2}&\frac{1}{(t^1+t^2)^2}\\
        \frac{1}{(t^1+t^2)^2}&\frac{1}{(t^2)^2}+\frac{1}{(t^1+t^2)^2}
    \end{pmatrix}\,.
\end{equation}

From geometry, we can deduce the relevant cones of curves and divisors~\cite{ottem2015birational}. Starting with curves, a convenient basis of effective curves is $C_1=D_2.D_2$ and $C_2=D_1.D_1$. Thus, the Mori cone is
\begin{equation}
    \cC_{\rm BPS-B}^{(0)}=M(X_{\rm bicubic})=\big\{a_1[C_1]+a_2[C_2]\big\vert a_1,a_2\in\mathbb{R}_{\geq 0}\big\}\,.
\end{equation}
In this example, as the dual divisors $D_1$ and $D_2$ are both nef divisors, the cone of movable curves is identical to the Mori cone, $\mathrm{Mov}[X_{\rm bicubic}]=\cM[X_{\rm bicubic}]$.
For divisors, knowing the (Hilbert) basis of effective divisors, we have
\begin{equation}
    \cC_{\rm BPS-B}^{(1)}=\cE(X_{\rm bicubic})=\big\{a_1[D_1]+a_2[D_2]\big\vert a_1,a_2\in\mathbb{R}_{\geq 0}\big\}\,.
\end{equation}
Similar to the case of curves, the nef cone is equivalent to the effective cone, $\mathrm{Nef}[X_{\rm bicubic}]=\cE[X_{\rm bicubic}]$. Therefore, the K\"ahler cone is
\begin{equation}
    K(X_{\rm bicubic})=\big\{(t^1,t^2)\in \mathbb{R}_+^2\big\}\,,
\end{equation}
with the closure given by $\overline{K(X_{\rm bicubic})}=\Nef(X_{\rm bicubic})$.
To study the boundary of the closure of the K\"ahler cone, let us recall that for a curve realized as a complete intersection of divisors in $X$, namely $C=D\cap D'\subset X$, the adjunction formula is
\begin{equation}
    K_C=(K_X+D+D')\vert_C=(D+D')\vert_C\,,\qquad \Rightarrow \qquad 2g(C)-2=(D+D').D.D'\,,
\end{equation}
where $g(C)$ denotes the genus of the curve.
For the bicubic, the basis of curves are $C_1,C_2$ which generate the Mori cone. Using the adjunction formula, we have
\begin{equation}
    g(C_1)=g(C_2)=1\,.
\end{equation}
Therefore, the curves that shrink along the boundaries of the K\"ahler cone are not rational $\cO_{\mathbb{P}^1}(-1)\oplus \cO_{\mathbb{P}^1}(-1)$ curves (flop curves), instead they are genus-one fibers. In particular, the volume of $C_1$ is $\volume(C_1)=3t^2$. The infinite-distance limit along the constant-volume slice at which this curve shrinks to zero size is $t^1\sim (t^2)^{-1/2}\to\infty$. The divisor that shrinks to zero size along this direction is $D_1$, whose calibrated volume behaves as $\volume(D_1)=3t^1t^2+3(t^2)^2/2\sim (t^1)^{-1}$. In this limit, the shrinking curve and divisor are related by $\volume(C_1)\sim \volume(D_1)^{1/2}$. Therefore, the leading light tower is the Kaluza--Klein (KK) tower indicating a decompactification limit along this boundary of the K\"ahler moduli space. By symmetry, the analogous analysis for $C_2$ gives an infinite-distance limit $t^2\sim (t^1)^{-1/2}\to \infty$ where $C_2$ shrinks to zero size. In this case, we also have $\volume(C_2)\sim \volume(D_2)^{1/2}$ indicating a decompactification limit in the moduli space. Therefore, the infinite-distance physics associated to the two walls of the K\"ahler moduli space both corresponds to a decompactification limit.

With this understanding, we can turn to analyzing the black brane solutions anticipated from the 5d $\cN=1$ supergravity action. The black 0-brane potential for a particle of charge $q=(q_1,q_2)$ is
\begin{equation}
    V_{\rm BH}=\frac{-2q_1q_2(t^1)^2(t^2)^2+q_2^2(t^2)^2\big(2(t^1)^2+2t^1t^2+(t^2)^2\big)+q_1^2(t^1)^2\big((t^1)^2+2t^1t^2+2(t^2)^2\big)}{(t^1)^2+t^1t^2+(t^2)^2}\,.
\end{equation}
One can solve for an extremum of the above effective potential after restricting to the constant-volume slice, but a simpler approach is to solve for the attractor equation using the central charge and a Lagrange multiplier $\lambda$ to enforces the constant-volume condition
\begin{equation}
    \partial_I\big(q_It^I-\lambda(\cF-1)\big)=0\,,\qquad \Rightarrow \qquad q_I=\frac13 Z \partial_I\cF\,.
\end{equation}
This gives the unique attractor point associated to the charge $q$
\begin{equation}
    r_*:=\frac{t_*^1}{t_*^2}=\frac{q_2-q_1+\sqrt{q_1^2-q_1q_2+q_2^2}}{q_1}\,.
\end{equation}
Then, together with the constant-volume constraint, we arrive at the unique attractor point associated to a given charge $q$
\begin{equation}
    t_*^2=\bigg(\frac{2}{r_*(r_*+1)}\bigg)^{1/3}\,,\qquad t_*^1=r_*t_*^2\,,
\end{equation}
which must lie in the K\"ahler moduli space.
As $(t_*^1,t_*^2)$ remains positive when $q_1,q_2>0$, the 0-brane with positive charge $q=(q_1,q_2)$ admits a BPS black 0-brane solution to the attractor mechanism (when $q$ is sufficiently large). Hence, the BPS black 0-brane cone is
\begin{equation}
    \cC_{\rm BPS-BB}^{(0)}=\mathrm{Mov}(X_{\rm cubic})=\big\{a_1[C_1]+a_2[C_2]\big\vert (a_1,a_2)\in\mathbb{R}_{\geq 0}^2\big\}\,,
\end{equation}
which is consistent with the expectation that $\cC_{\rm BPS-BB}^{(0)}\subset \cC_{\rm BPS-B}^{(0)}$. Furthermore, we can see that this is indeed dual to the effective cone $\cE(X_{\rm bicubic})$, namely $(\cC_{\rm BPS-BB}^{(0)})^\vee=\cC_{\rm BPS-B}^{(1)}$.

\begin{table}[!tp]
\begin{adjustwidth}{-1in}{-1in}
    \centering
    \begin{tabular}{c|c|c|c|c|c|c}
        \diagbox[width=1.5cm,height=0.7cm]{$q_2$}{$q_1$} & 0 & 1 & 2 & 3 & 4 & 5  \\
        \hline
        0 & \Pnk \textasteriskcentered & \Pnk 189 & \Pnk 189 & \Pnk 162 & \Pnk 189 & \Pnk 189 \\
        \hline
        1 & \Pnk 189 & \Pnk 8262 & \Pnk 142884 & \Pnk 1492290 & \Pnk 11375073 & \Pnk 69962130 \\
        \hline
        2 & \Pnk 189 & \Pnk 142884 & \Pnk 13108392 & \Pnk 516953097 & \Pnk 12289326723 & \Pnk 206210244204\\
        \hline
        3 & \Pnk 162 & \Pnk 1492290 & \Pnk 516953097 & \Pnk 55962304650 & \Pnk 3154647509010 & \Pnk 114200061474474 \\
        \hline
        4 & \Pnk 189 & \Pnk 11375073 & \Pnk 12289326723 & \Pnk 3154647509010 & \Pnk 366981860765484 & \Pnk 25255131122299086 \\
        \hline
        5 & \Pnk 189 & \Pnk 69962130 & \Pnk 206210244204 & \Pnk 114200061474474 & \Pnk 25255131122299086 & \Pnk 3057363233014221000\\
    \end{tabular}
\tikz[baseline,remember picture,overlay]{
  \draw[red,thick,->] (-16.50,1) -- (-0.1,1cm);
  \draw[red, thick,->] (-16.50,1) -- (-16.50,-1.75cm);
}
\tikz[baseline,remember picture,overlay]{
  \draw[blue,thick,dashed,->] (-16.62,1) -- (-0.1,1cm);
  \draw[blue,thick,dashed,->] (-16.62,1) -- (-16.62,-1.9cm);
  \fill[black,opacity=0.3] (-16.62,0.5) circle[radius=5pt];
  \fill[black,opacity=0.3] (-15.1,1) circle[radius=5pt];
  \fill[black,opacity=0.3] (-15.1,0.5) circle[radius=5pt];
}
\end{adjustwidth}
    \caption{The genus-0 GV invariants for curves classes $(q_1,q_2)\in H_2(X_{\rm bicubic},\mathbb{Z})$. The shaded blue region indicates the cone of movable curves which in this geometry is identical to the Mori cone. Hence, the (solid and dashed) lines indicate the extremal rays of both the Mori cone and the cone of movable curves.}
    \label{tab:GV bicubic}
\end{table}

\begin{figure}[!tp]
    \centering
    \begin{subfigure}{.475\textwidth}
        \includegraphics[width=\linewidth]{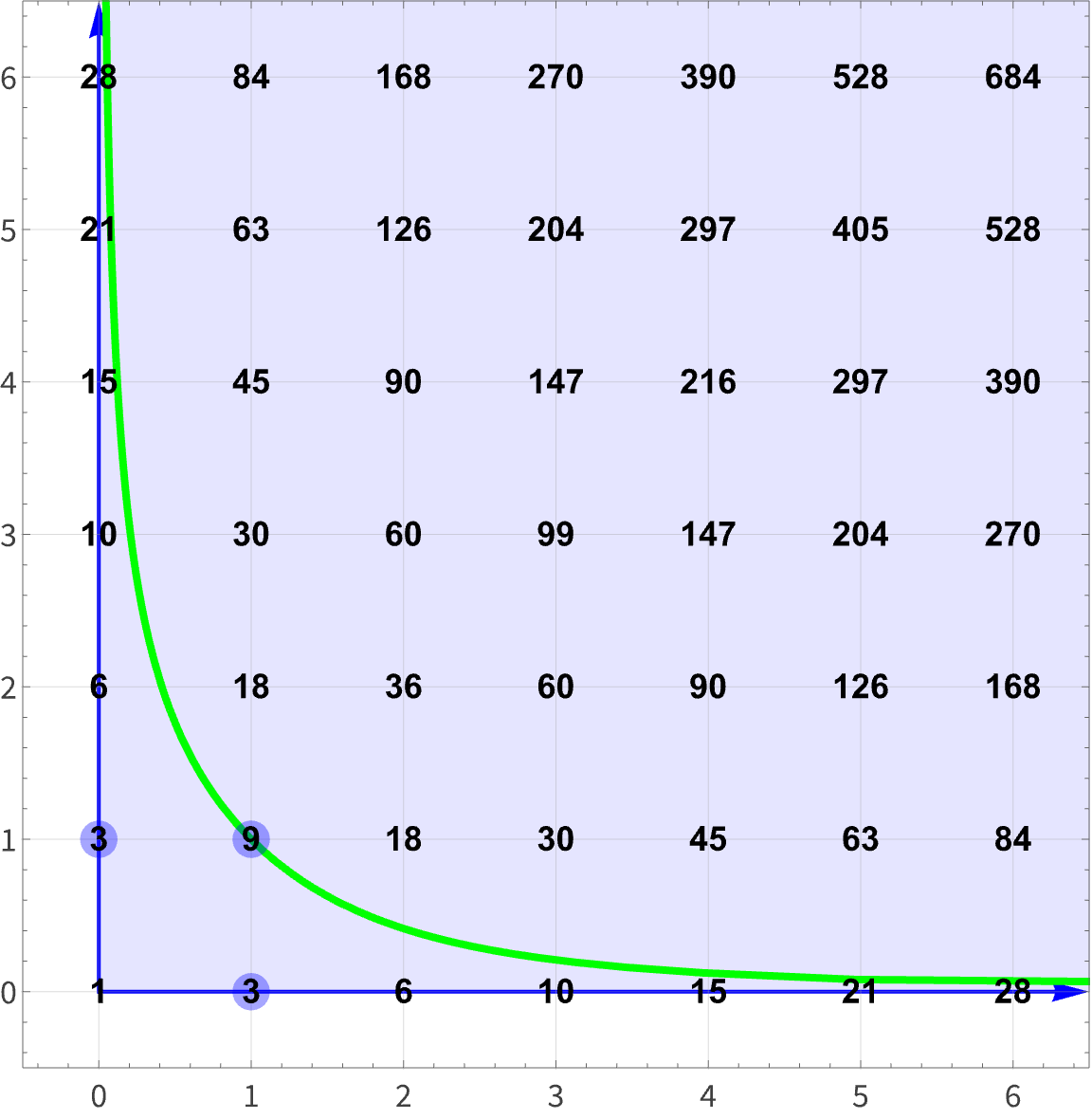}
        \begin{picture}(0,0)\vspace*{-1.2cm}
            \put(-20,140){$t^2$}
            \put(120,0){$t^1$}
        \end{picture}
        \vspace*{-0.25cm}
        \caption{$\cC_{\rm BPS-BB}^{(1)}=\cC_{\rm BPS-B}^{(1)}$}
        \label{fig:Xbicubiccones}
    \end{subfigure}
    \hfill
    \begin{subfigure}{.475\textwidth}
        \includegraphics[width=\linewidth]{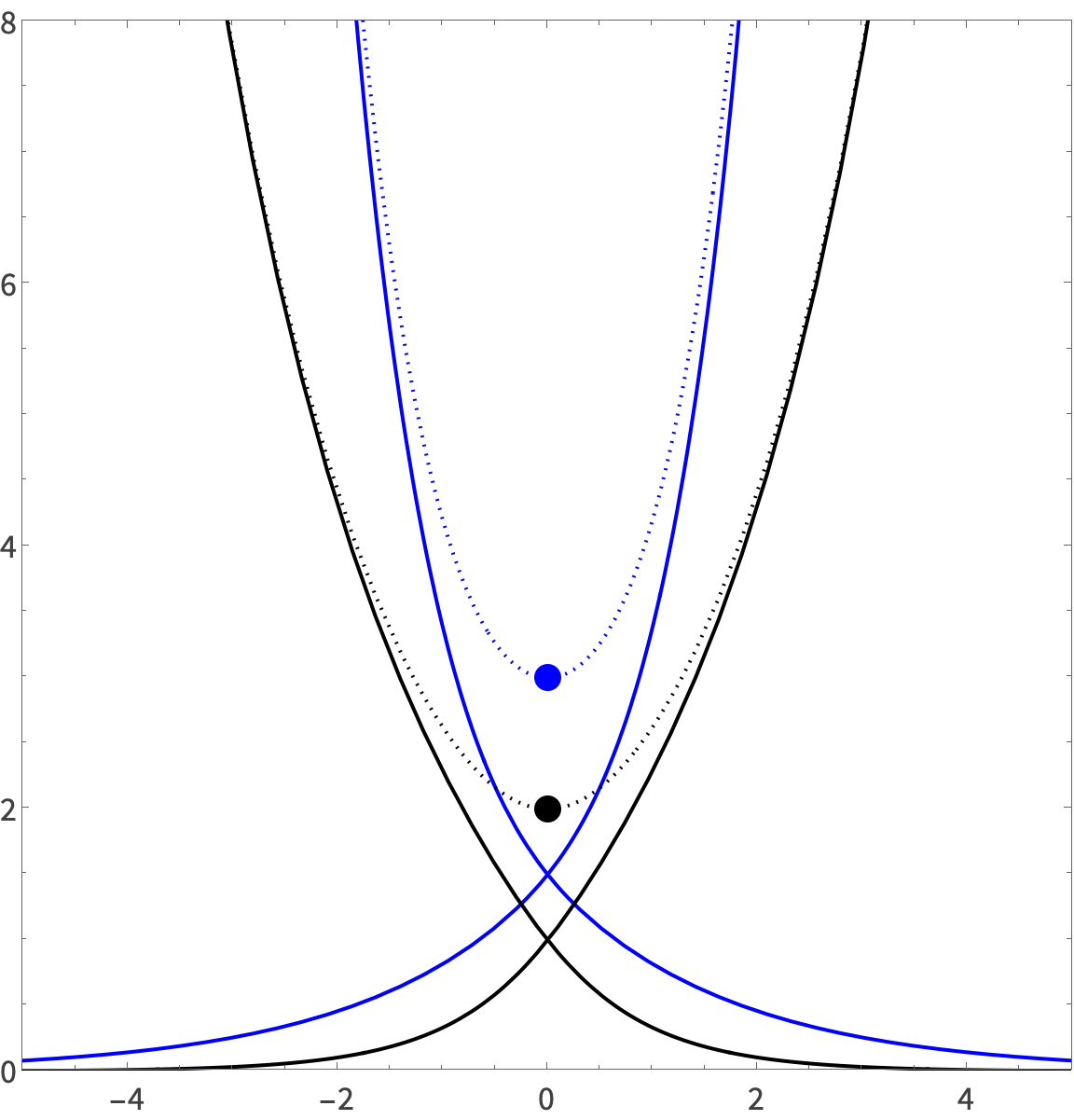}
        \begin{picture}(0,0)\vspace*{-1.2cm}
            \put(-10,140){$T$}
            \put(120,0){$\Delta$}
            \put(107,130){\footnotesize \textcolor{blue}{$(1,1)$}}
            \put(20,40){\footnotesize \textcolor{blue}{$(1,0)$}}
            \put(193,40){\footnotesize \textcolor{blue}{$(0,1)$}}
            \put(107,90){\footnotesize {$(1,1)$}}
            \put(107,90){\footnotesize {$(1,1)$}}
            \put(25,160){\footnotesize {$(1,0)$}}
            \put(185,160){\footnotesize {$(0,1)$}}
        \end{picture}
        \vspace*{-0.25cm}
        \caption{Tension of BPS 1-branes}
        \label{fig:Xbicubictension}
    \end{subfigure}
    \caption{The properties of BPS 1-branes in M-theory compactified on $X_{\rm bicubic}$. In \cref{fig:Xbicubiccones}, we have shown the simplicial cones $\cC_{\rm BPS-BB}^{(1)}=\cC_{\rm BPS-B}^{(1)}$ in red. The constant-volume slice in the K\"ahler moduli space is illustrated in green. We also highlight the divisor classes whose volumes are studied in the right figure. Additionally, we overlay the values of $h^0(X_{\rm bicubic},\cO_{X_{\rm bicubic}}([D]))$ for divisor classes in $\cC_{\rm BPS-B}^{(1)}$. In \cref{fig:Xbicubictension}, we illustrate the masses of BPS 0-branes (in black) and the tensions of BPS 1-branes (in blue) as a function of the canonically normalized scalar along the constant-volume slice. The curve classes are specified in the basis $(q_1,q_2)\in \cC_{\rm BPS-B}^{(0)}$ and the divisor classes are specified in the basis $(p^1,p^2)\in \cC_{\rm BPS-B}^{(1)}$. The minima of their masses or tensions are indicated by the solid dot. Here, the dotted lines indicate curve classes or divisor classes strictly in the interior of $\cC_{\rm BPS-BB}$ while the solid lines indicate those along the boundaries of $\cC_{\rm BPS-BB}$.}
    \label{fig:Xbicubic}
\end{figure}
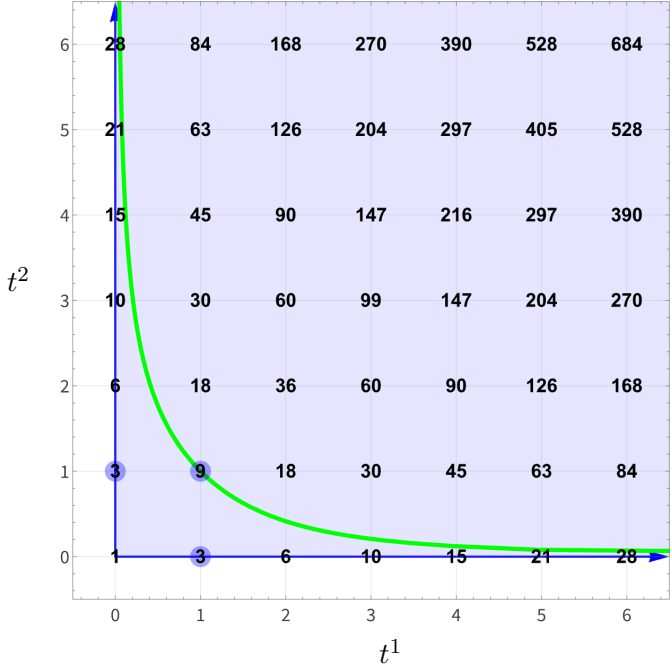
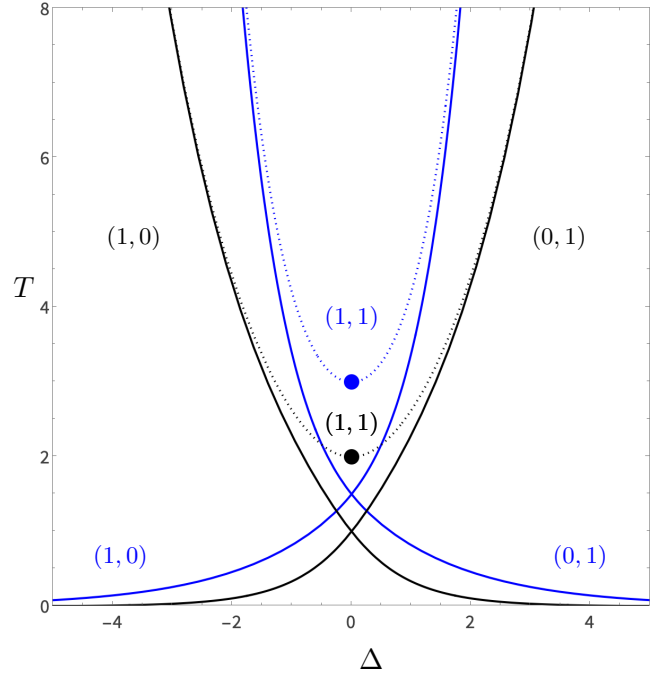

The effective potential for a black 1-brane of charge $p=(p^1,p^2)$ is
\begin{equation}
    V_{\rm BS}=2\bigg[\frac{(p^1)^2}{(t^1)^2}+\frac{(p^2)^2}{(t^2)^2}+\frac{(p^1+p^2)^2}{(t^1+t^2)^2}\bigg]\,.
\end{equation}
Similarly, we can examine whether the magnetic central charge admits a minimum inside the K\"ahler cone. To this end, the solutions to the black 1-brane attractor mechanism, restricted to the constant-volume slice, are
\begin{equation}
    t^I_*=\frac{p^I}{\big(\frac12 p^1p^2(p^1+p^2)\big)^{1/3}}\,.
\end{equation}
As the K\"ahler cone constraints are $t^1,t^2>0$, the magnetic string with charge $p=(p^1,p^2)$ admits a BPS black 1-brane solution only when both $p^1,p^2>0$. Hence, the BPS black 1-brane cone is the closure of the ample cone
\begin{equation}
    \cC_{\rm BPS-BB}^{(1)}=\Nef(X)\,.
\end{equation}
With this identification, the duality $\cC_{\rm BPS-BB}^{(1)}=(\cC_{\rm BPS-B}^{(0)})^\vee$ is satisfied.

There also exist non-BPS solutions to the attractor mechanism in this theory as was noted in~\cite{Long:2021lon}. Nevertheless, in all such cases, while the extremal non-BPS solution would flow to an attractor point where $D_IV=0$, it no longer satisfies $\partial_i Z=0$ in the interior of the K\"ahler cone. This would necessarily imply that for all extremal black $p$-branes in this theory, the tensions of these $p$-branes satisfy the following strict inequalities
\begin{equation}
    |Z_{\rm min}|<T(t^I)<T^{\rm Ext}\,,
\end{equation}
for all $t^I\in K(X_{\rm bicubic})$.

In addition to our diagnosis \cref{eq:BPS predict}, there are independent methods to determine if the particle or string is supersymmetric. In this model, these explicitly. 

To determine whether a curve class is effective (and, hence, whether the resulting particle is supersymmetric), we compute the genus-0 GV invariants associated to the curve class~\cite{Gopakumar:1998jq}. For this particular geometry, the GV invariants have been previously computed in~\cite{Hosono:1994ax}. 
Let us briefly review the computations here, as we will also use this algorithm to compute the genus-0 GV invariants for curve classes in later geometries.\footnote{One can also find a review for calculating genus-0 GV invariants for CICYs in~\cite[App. A]{Alim:2021vhs}.} To start, the relevant Picard--Fuchs equations for the bicubic are
\begin{equation}
    L_i=\theta_i^3-z_i\prod_{I=1}^3(3(\theta_1+\theta_2)+I)\,,
\end{equation}
where $i=1,2$, $z_1,z_2$ denote the complex-structure parameters in the B-model, and $\theta_i:=z_i\partial_{z_i}$.
The fundamental period is then the following hypergeometric series
\begin{equation}
    \varpi_0(z_1,z_2)
    =\sum_{m,n\geq 0}\frac{(3m+3n)!}{(m!)^3(n!)^3}z_1^mz_2^n
    =\sum_{m,n\geq 0}\frac{\Gamma(1+3m+3n+3\rho_1+3\rho_2)}{\Gamma(1+m+\rho_1)^3\Gamma(1+n+\rho_2)^3}z_1^{m+\rho_1}z_2^{n+\rho_2}\,,
\end{equation}
where we have introduced the Frobenius deformation parameters $\rho_1,\rho_2$.
From here, we can introduce
\begin{equation}
    \varpi_i=\frac{1}{2\pi i}\partial_{\rho_i}\varpi_0\big\vert_{\rho=0}\,.
\end{equation}
This gives the mirror map $t^i(z)=\varpi_i/\varpi_0$, where $t^i$ denote the K\"ahler parameters in the A-model.

With the definition of the K\"ahler parameters, the prepotential of the theory is
\begin{equation}
    \cF=\frac16C_{ijk}^{(0)}t^it^jt^k+\frac{1}{24}c_{2,i}t^i+\frac{\zeta(3)}{2\pi i}\chi(X_{\rm bicubic})+\frac{1}{2\pi i}\sum_{(d_1,d_2)\in \cM}n_{d_1,d_2}^{(0)}\mathrm{Li}_3(q_1^{d_1}q_2^{d_2})\,,
\end{equation}
where $\chi(X_{\rm bicubic})=-162$, $n_{d_1,d_2}^{(0)}$ denotes the genus-0 GV invariants associated to the curve class $(d_1,d_2)$, and $q_i=\exp(2\pi i\,t^i)$. From here, we can extract the quantum-corrected triple intersection number via
\begin{equation}
\label{eq:quantum Cijk 1}
    C_{ijk}=\partial_{t^i}\partial_{t^j}\bigg[\frac12C_{abk}^{(0)}\partial_{\rho_a}\partial_{\rho_b}\log\varpi_0\big\vert_{\rho=0}\bigg]\,.
\end{equation}
Another equivalent derivation of the triple intersection number involves the genus-0 GV invariants as
\begin{equation}
\label{eq:quantum Cijk 2}
    C_{ijk}=C_{ijk}^{(0)}+\sum_{(d_1,d_2)\in \cM}n_{d_1,d_2}^{(0)}d_id_jd_k\frac{q_1^{d_1}q_2^{d_2}}{1-q_1^{d_1}q_2^{d_2}}\,.
\end{equation}
Expanding both \cref{eq:quantum Cijk 1} and \cref{eq:quantum Cijk 2} as $q$-series and matching the appropriate coefficients, one can extract the values of the genus-0 GV invariants.

Using this algorithmic description, we can perform this computation for a collection of CICY threefolds straightforwardly. A general implementation of this algorithm can be found in e.g., \verb+cygv+~\cite{cygv}. We have included a simplified version of this general code for favorable CICYs; the corresponding python implementation can be found in the ancillary files accompanying this submission.
Using this algorithm, the genus-0 GV invariants for a choice of curve class in the basis $(q_1,q_2)$ are shown in \cref{tab:GV bicubic}. We observe the symmetry of genus-0 GV invariants under the exchange $q_1\leftrightarrow q_2$ as noted in~\cite{Hosono:1994ax}. More importantly, we observe that integral lattice sites within the cone of movable curves indeed continue to have non-zero GV invariants. Since in this geometry the BPS 0-branes satisfy $\cC_{\rm BPS-BB}^{(0)}=\cC_{\rm BPS-B}^{(0)}$, we consequently observe BPS completeness in $\cC_{\rm BPS-B}^{(0)}$.

For a divisor class $[D]=a_1[D_1]+a_2[D_2]\in \cE(X_{\rm bicubic})$, we have the following short exact sequence
\begin{equation}
    0\to \cO_V((a_1-3)[\hat{D}_1]+(a_2-3)[\hat{D}_2])\to \cO_V(a_1[\hat{D}_1]+a_2[\hat{D}_2])\to \cO_{X}(a_1[D_1]+a_2[D_2])\to 0\,,
\end{equation}
where $D\subset X$ is the restriction of $\hat{D}\subset V$ to $X$ and $V$ denotes the ambient projective space. From this data, we can deduce the values of $h^0(X,\cO_X(D))$.
For the bicubic geometry, we have the following analytic expression
\begin{equation}
    h^0(X,\cO_X(a_1D_1+a_2D_2))=\binom{a_1+2}{2}\binom{a_2+2}{2}-H(a_1-3)H(a_2-3)\binom{a_1-1}{2}\binom{a_2-1}{2}\,.
\end{equation}
where $H$ is the Heaviside function.
Therefore, every integral divisor class in the effective cone is effective and admits a minimal-volume representative in its class.

Thus, with the explicit values of the genus-0 GV invariants for curve classes in the Mori cone along with $h^0(X_{\rm bicubic},\cO_{X_{\rm bicubic}}(D))$ for divisor classes in the effective cone shown in \cref{tab:GV bicubic} and \cref{fig:Xbicubiccones}, respectively, \cref{conj:BPS completeness} for $X_{\rm bicubic}$ is evidently satisfied.

Let us turn to the masses and tensions of BPS branes in the BPS-B cones. We begin with the masses of BPS 0-branes. The relevant three 0-branes are shown as black dots in \cref{tab:GV bicubic} and their masses are plotted as black curves in \cref{fig:Xbicubictension}. We can see that a BPS 0-brane whose charge lies strictly in the interior of $\cC_{\rm BPS-BB}^{(0)}$ acquires its minimum in the interior of the moduli space. By contrast, for charge sites on the boundaries of $\cC_{\rm BPS-BB}^{(0)}$ which coincides with the boundaries of $\cC_{\rm BPS-B}^{(0)}$, the corresponding masses vanish along the infinite-distance walls of moduli space. As the relevant cones are all simplicial in this geometry, the exact same behavior prevails among the three choices of BPS 1-branes shown as blue dots in \cref{fig:Xbicubiccones} and their tensions are illustrated as blue curves in \cref{fig:Xbicubictension}. Notably, as analyzed at the beginning of this example, the masses of the BPS 0-branes vanish faster than the tensions of the BPS 1-branes that become tensionless in the same limit. Hence, both infinite-distance limits correspond to decompactification limits in moduli space. Lastly, the BPS 0-branes and BPS 1-branes whose masses and tensions vanish along the infinite-distance limits shown in \cref{fig:Xbicubictension} are the primitive generators of their respective BPS-B cones (and equivalently their respective BPS-BB cones in this example). Hence, the elementary constituent conjecture~\cite{Nevoa:2025xiq} is satisfied among these boundary primitive BPS branes that generate the cones.

\subsubsection{CICY in $\mathbb{P}^4\times \mathbb{P}^1$}
\label{sec:CICY}

In the bicubic example, all relevant cones were simply the positive orthant in the chosen charge basis. Additionally, the extended K\"ahler cone associated with the birational family of $X_{\rm bicubic}$ has only one chamber. These features are not generic for Calabi--Yau threefolds as we will demonstrate via the following example. 

Let us consider the geometry originally studied in~\cite{Greene:1995hu,Greene:1996dh} and more recently studied in the context of the WGC in~\cite{Alim:2021vhs,reece2025co}. The Calabi--Yau threefold has $(h^{1,1},h^{2,1})=(2,86)$. We will refer to it as $X_{2,86}$. The Calabi--Yau threefold is defined by the following configuration matrix 
\begin{equation}
\label{eq:X286 config}
    X_{2,86}=\left(\begin{array}{c|cc}
        \mathbb{P}^4 & 4 & 1\\
        \mathbb{P}^1 & 1 & 1
    \end{array}\right)\,.
\end{equation}
The two prime toric divisor classes inherited from the ambient projective toric variety are
\begin{equation}
    D_1=\cO_{\mathbb{P}^4\times \mathbb{P}^1}(1,0)\vert_{X_{2,86}}\,,\qquad D_2=\cO_{\mathbb{P}^4\times \mathbb{P}^1}(0,1)\vert_{X_{2,86}}\,.
\end{equation}
These two divisor classes are effective on $X_{2,86}$ by construction and generate the effective cone, $\cE(X_{2,86})$. This CICY is indeed realized as a complete intersection of ample hypersurfaces in the ambient projective toric variety. The relevant cones associated to this geometry are all simplicial cones in $\mathbb{R}^2$. In the basis of prime toric divisors, where the K\"ahler form can be expressed as $J=t^1[D_1]+t^2[D_2]$, the prepotential and the K\"ahler cone are
\begin{equation}
\label{eq:CICY prepot I}
    \cF^{(I)}=\frac16\left(5(t^1)^3+12(t^1)^2t^2\right)\,,\qquad K_I(X_{2,86})=\{(t^1,t^2)\in\mathbb{R}_+^2\}\,.
\end{equation} 
Let us analyze the physics along the walls of the closure of the K\"ahler cone. From the configuration matrix given in \cref{eq:X286 config}, the defining equations of this geometry can be written as $y_0f_4(x)+y_1g_4(x)=0$ and $y_0l_1(x)+y_1m_1(x)=0$ where $[y_0:y_1]$ are homogeneous coordinates on the $\mathbb{P}^1$, $f_4,g_4$ are quartics, while $l_1,m_1$ are linear on $\mathbb{P}^4$. Projection onto the first factor maps $X_{2,86}$ to the quintic $f_4(x)m_1(x)-g_4(x)l_1(x)=0\subset \mathbb{P}^4$ whose singular points occur when $f_4(x)=g_4(x)=l_1(x)=m_1(x)=0$. For generic choices of sections, the number of such points is $16$. This corresponds to the 16 exceptional curves $C_e$ in $X_{2,86}$, which are contracted under the map to the singular quintic. These intersect with the prime toric divisors as $D_1.C_e=0$ and $D_2.C_e=1$, since $C_e$ lies entirely in the $\mathbb{P}^1$ fiber over a point in $\mathbb{P}^4$. Consequently, $\volume(C_e)=t^2$. 
The restriction of the normal bundle of $X_{2,86}$ in the ambient space to $C_e$ is $N_{X_{2,86}/(\mathbb{P}^4\times \mathbb{P}^1)}\vert_{C_e}\cong \cO_{\mathbb{P}^1}(1)\oplus \cO_{\mathbb{P}^1}(1)$.
As the curves $C_e$ are isolated, we also have $H^0(C_e,N_{C_e/X_{2,86}})=0$ which together with the Birkhoff--Grothendieck theorem~\cite[Cor. 5.2.8]{Huybrechts2005} forces 
\begin{equation}
    N_{C_e/X_{2,86}}\cong \cO_{\mathbb{P}^1}(-1)\oplus \cO_{\mathbb{P}^1}(-1)\,.
\end{equation}
Hence, the curves $C_e$ are flop curves in this geometry and the flop wall is located at $t^2=0$ in the K\"ahler moduli space. 
The other curve class relevant to the $t^1=0$ wall of the K\"ahler cone is $C=D_1.D_2$, whose volume is $\volume(C)=4t^1$. From the adjunction formula, we have $g(C)=3$; thus it is not a flop curve. This is also evident as the $t^1\to 0$ limit is an infinite-distance boundary along the constant-volume K\"ahler moduli space. This infinite-distance limit corresponds to the emergent weakly-coupled heterotic string limit, as has been discussed in~\cite{vandeHeisteeg:2023dlw}.

Thus, the extended K\"ahler cone has at least two phases, corresponding to two birationally equivalent Calabi--Yau threefolds connected across the flop wall discussed above. The difference in the prepotentials across these two phases is $\cF^{(II)}-\cF^{(1)}=-\frac{16}{6}(t^2)^3$, capturing the effect of the flopping 16 rational curves.
Using the same coordinates, the prepotential and K\"ahler cone in the second phase are~\cite{Alim:2021vhs} 
\begin{equation}
\label{eq:CICY prepot II}
    \cF^{(II)}=\frac{1}{6}\left(5(t^1)^3+12(t^1)^2t^2-16(t^2)^3\right)\,,
    \qquad K_{II}(X_{2,86})=\{t^1>-4t^2>0\}\,.
\end{equation}
The second wall associated to the second phase of the K\"ahler cone is a finite-distance boundary of the K\"ahler moduli space. The divisor that shrinks to zero size along this wall is $D=D_1-D_2$ where $\volume(D)=(t^1+4t^2)^2/2$. 
The collapsing divisor is identified as $D\cong\mathbb{P}^2$, with normal bundle $\cO_{\mathbb{P}^2}(-3)$ and $D^3=9$.
Consider a line $L\subset D$. Since $D^2=-3L$, we have $D_1.L=1$ and $D_2.L=4$ leading to $\volume(L)=t^1+4t^2$, which also shrinks along the finite-distance wall $t^1+4t^2=0$. By adjunction, we have $N_{L/D}\cong \cO_{\mathbb{P}^1}(1)$ which gives $N_{L/X_{2,86}}\cong \cO_{\mathbb{P}^1}(1)\oplus \cO_{\mathbb{P}^1}(-3)$, the characteristic normal bundle of a line in a collapsing $\mathbb{P}^2$. Hence, this is a finite-distance CFT wall, where the divisor $D$ shrinks to a point.

The extended K\"ahler cone of this geometry is therefore
\begin{equation}
\label{eq:CICY extended Kahler cone}
    K_{\infty}(X_{2,86})=K_I(X_{2,86})\cup K_{II}(X_{2,86})\,.
\end{equation}
In particular, the boundary $t^1=0$ corresponds to an emergent string limit, while the boundary $t^1=-4t^2$ corresponds to a finite-distance CFT wall where $\mathbb{P}^2$ collapses to a point.

The genus-0 GV invariants for this geometry were computed in~\cite{Alim:2021vhs}. From these invariants, one can deduce the Mori cone. In particular, the Mori cones for each phase of the geometry are
\begin{equation}
    M_{I}(X_{2,86})=\big\{(q_1,q_2)\in\mathbb{R}_+^2\big\}\,,\qquad M_{II}(X_{2,86})=\big\{(q_1,q_2)\in\mathbb{R}^2\,\vert\, q_1\geq 0\,,4q_1\geq q_2\big\}\,.
\end{equation}
The restricted Mori cone is then the intersection of these two Mori cones. Namely, we have 
\begin{equation}
    M_{\infty}(X_{2,86})=M_I(X_{2,86})\cap M_{II}(X_{2,86})=\big\{(q_1,q_2)\in\mathbb{R}^2\,\vert\, q_2\geq 0\,,4q_1\geq q_2\big\}\,.
\end{equation}

\begin{table}[!tp]
\begin{adjustwidth}{-1in}{-1in}
    \centering
    \begin{tabular}{c|c|c|c|c|c|c}
        \diagbox[width=1.5cm,height=0.7cm]{$q_2$}{$q_1$} & 0 & 1 & 2 & 3 & 4 & 5  \\
        \hline
        0 & \Pnk \textasteriskcentered & \Pk 640 & \Pk 10032 & \Pk 288384 & \Pk 10979984 & \Pk 495269504\\
        \hline
        1 & 16 & \Pnk 2144 & \Pk 231888 & \Pk 23953120 & \Pk 2388434784 & \Pk 232460466048\\
        \hline
        2 & 0 & \Y 120 & \Pnk 356368 & \Pk 144785584 & \Pk 36512550816 & \Pk 7251261673320\\
        \hline
        3 & 0 & \Y $-32$ & \Y 14608 & \Pnk 144051072 & \Pk 115675981232 & \Pk 50833652046112 \\
        \hline
        4 & 0 & \Y 3 & \Y $-4920$ & \Y 5273880 & \Pnk 85456640608 & \Pnk 106397389165188 \\
        \hline
        5 & 0 & 0 & \Y 1680 & \Y $-1505472$ & \Y 3018009984 & \Pnk 62800738246496 \\
        \hline
        6 & 0 & 0 & \Y $-480$ & \Y 512136 & \Y $-748922304$ & \Y 2196615443648\\
        \hline
        7 & 0 & 0 & \Y 80 & \Y $-209856$ & \Y 218062416 & \Y $-493158341440$ \\
        \hline
        8 & 0 & 0 & \Y $-6$ & \Y 75300 & \Y $-90910176$ & \Y 120803175480\\
        \hline
        9 & 0 & 0 & 0 & \Y $-21600$ & \Y 37721680 & \Y $-48214555136$\\
        \hline
        10 & 0 & 0 & 0 & \Y 4312 & \Y $-15086208$ & \Y 19843557320\\
    \end{tabular}
\tikz[baseline,remember picture,overlay]{
  \draw[blue, thick,->] (-11.50,2.2) -- (-0.1,2.2cm);
  \draw[blue, thick,->] (-11.50,2.2) -- (-3.0,-0.5cm);
}
\tikz[baseline,remember picture,overlay]{
  \draw[red,thick,dashed,->] (-11.50,2.2) -- (-0.1,2.2cm);
  \draw[red,thick,->] (-11.75,2.2) -- (-9.50,-2.cm);
  \fill[red,opacity=0.3] (-10.67,0.25) circle[radius=5pt];
  \fill[red,opacity=0.3] (-10.67,0.75) circle[radius=5pt];
  \fill[blue,opacity=0.3] (-10.67,1.73) circle[radius=5pt];
  \fill[blue,opacity=0.2] (-10.67,2.25) circle[radius=5pt];
  \fill[blue,opacity=0.3] (-1.9,0.25) circle[radius=5pt];
  \fill[black,opacity=0.3] (-11.65,1.73) circle[radius=5pt];
  \fill[blue,opacity=0.2] (-7.4,1.73) circle[radius=5pt];
}
\end{adjustwidth}
    \caption{The genus-0 GV invariants for curve classes $(q_1,q_2)\in H_2(X_{2,86},\mathbb{Z})$. Here, the red shaded region indicates $\cC_{\rm BPS-B}^{(0)}$ (the restricted Mori cone), while the blue shaded regions indicate $\cC_{\rm BPS-BB}^{(0)}$ (the cone of movable curves). Here, the fainter blue region indicates $\cC_{\rm BPS-BB}^{(0,I)}$ while the darker blue region indicates $\cC_{\rm BPS-BB}^{(0,II)}$. The blue and red lines indicate the extremal rays of $\cC_{\rm BPS-BB}^{(0)}$ and $\cC_{\rm BPS-B}^{(0)}$, respectively. The dotted integral sites in this table are further analyzed in \cref{fig:X286 mass}.}
    \label{tab:GV X286}
\end{table}

\begin{figure}
    \centering
    \includegraphics[width=0.5\linewidth]{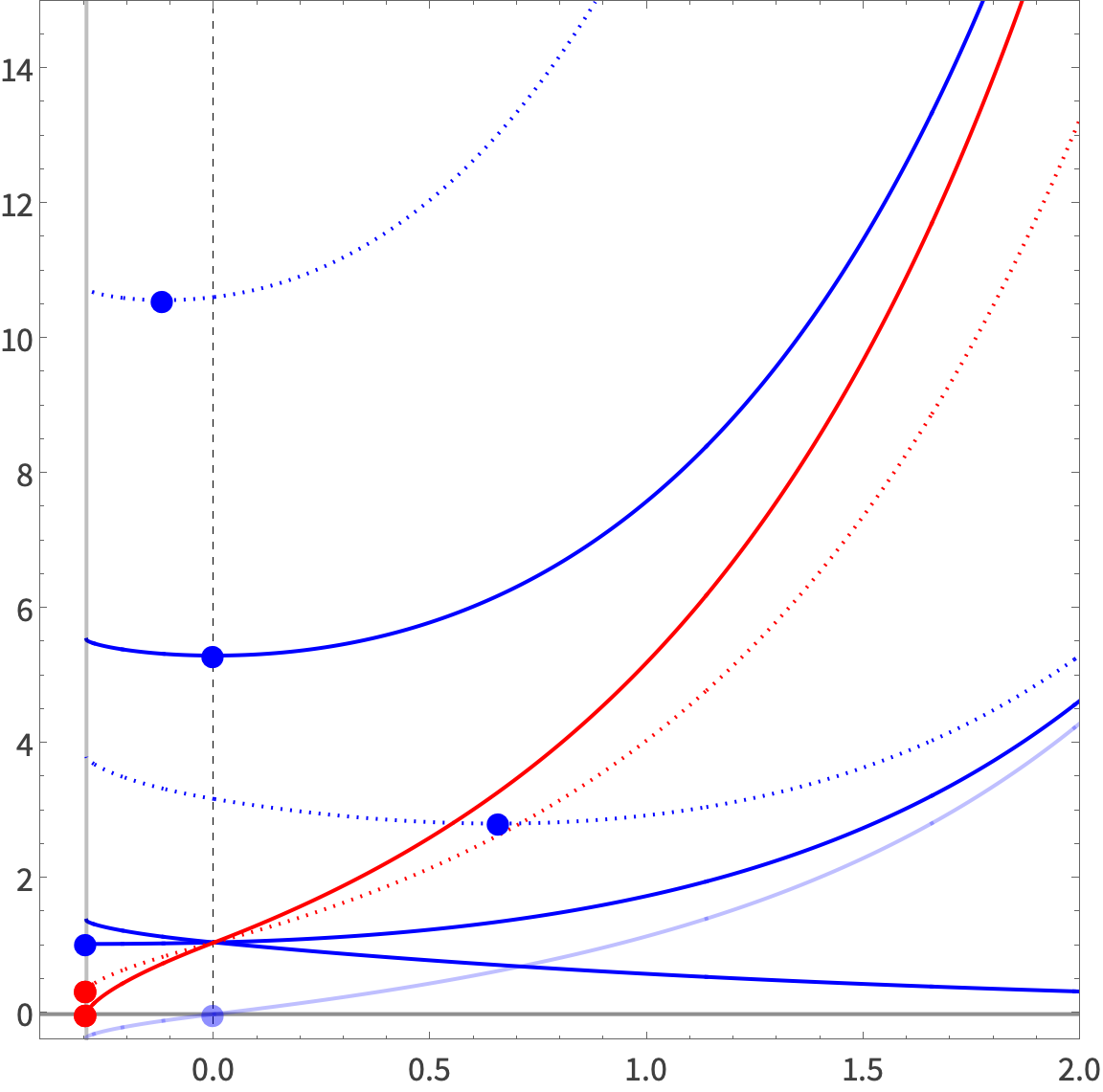}
    \begin{picture}(0,0)\vspace*{-1.2cm}
        \put(-265,120){$m$}
        \put(-120,-10){$\Delta$}
        \put(-230,245){\footnotesize $\Delta_{\rm CFT}$}
        \put(-205,245){\footnotesize $\Delta_{\rm flop}$}
        \put(-125,245){\footnotesize \textcolor{blue}{$(10,9)$}}
        \put(-180,110){\footnotesize \textcolor{blue}{$(5,4)$}}
        \put(-30,245){\footnotesize \textcolor{red}{$(1,4)$}}
        \put(-5,215){\footnotesize \textcolor{red}{$(1,3)$}}
        \put(-180,70){\footnotesize \textcolor{blue}{$(3,1)$}}
        \put(-5,85){\footnotesize \textcolor{blue}{$(1,1)$}}
        \put(-35,50){\footnotesize \textcolor{blue!50}{$(0,1)$}}
        \put(-5,20){\footnotesize \textcolor{blue}{$(1,0)$}}
    \end{picture}
    \vspace*{0.25cm}
    \caption{The masses of BPS 0-branes in $\cC_{\rm BPS-B}^{(0)}$ for M-theory compactified on $X_{2,86}$, shown as functions of the canonically normalized scalar parameterizing the constant-volume slice in $K(X_{2,86})$. The finite-distance CFT wall and flop wall are shown as vertical gray lines. The blue and red lines indicate states that are either in $\cC_{\rm BPS-BB}$ or $\cC_{\rm BPS-B}\backslash \cC_{\rm BPS-BB}^{(0)}$, respectively. (These states are also shown in \cref{tab:GV X286}.) The dotted lines indicate states that lie strictly in the interior of either $\cC_{\rm BPS-BB}$ or $\cC_{\rm BPS-B}\backslash \cC_{\rm BPS-BB}$. The solid dots indicate the location of the minimum of the mass. The charges shown are in the basis $(q_1,q_2)\in \cC_{\rm BPS-B}^{(0)}$. Although the flop curve lies within $\cC_{\rm BPS-B}^{(0,I)}$, it is not an element of the restricted cone $\cC_{\rm BPS-B}^{(0)}$.}
    \label{fig:X286 mass}
\end{figure}
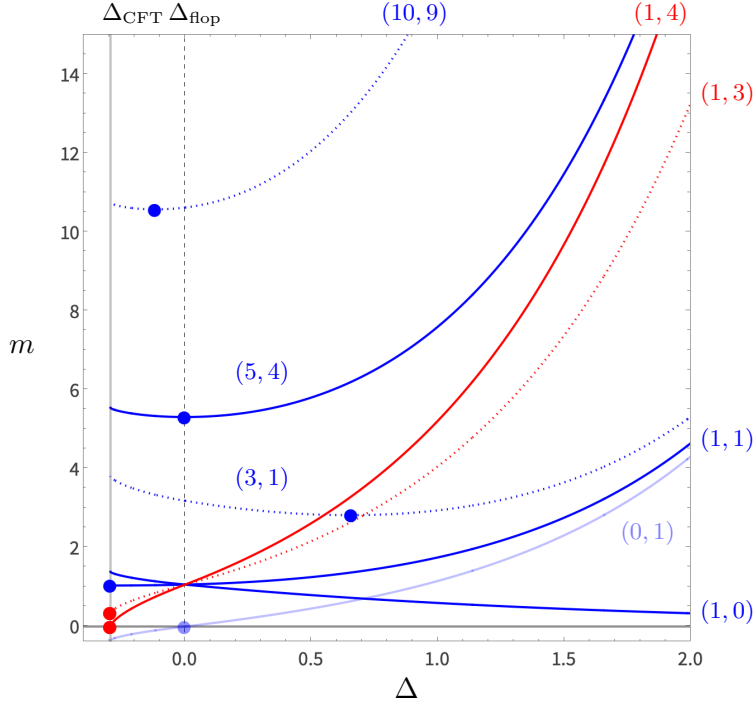

Let us study the attractor mechanism in each phase. To this end, the black hole attractor potential in the first phase $K_I$ is
\begin{equation}
    V_{\rm BH}^{(I)}=\frac{1}{48}\bigg[\left(48q_1^2-40q_1q_2+25q_2^2\right)(t^1)^2+80q_2^2t^1t^2+96q_2^2(t^2)^2\bigg]\,.
\end{equation}
For BPS attractor solutions, it is convenient to introduce the parameter $\phi=t^2/t^1$ where the K\"ahler constraint in the first phase is $\phi\geq 0$. Then, the attractor solution becomes
\begin{equation}
    \frac{q_2}{q_1}=\frac{4}{5+8\phi}\,.
\end{equation}
Therefore, the BPS black 0-brane cone associated with the first phase of the K\"ahler cone is
\begin{equation}
    \cC_{\rm BPS-BB}^{(0,I)}=\bigg\{(q_1,q_2)\in\mathbb{R}^2\,\bigg\vert\, \frac45 q_1\geq q_2\geq 0\bigg\}\,.
\end{equation}
We can also solve for the attractor equations in the second phase using the prepotential in \cref{eq:CICY prepot II}.
The BPS black 0-brane attractor equation then gives
\begin{equation}
    \frac{q_2}{q_1}=\frac{4-16\phi^2}{5+8\phi}\,,
\end{equation}
with the K\"ahler cone constraints in the second phase as $-\frac14\leq \phi\leq 0$. This leads to the second phase of the BPS black 0-brane cone as
\begin{equation}
    \cC_{\rm BPS-BB}^{(0,II)}=\bigg\{(q_1,q_2)\in\mathbb{R}^2\,\bigg\vert\, q_1\geq 0\,,q_1\geq q_2\geq \frac45 q_1\bigg\}\,.
\end{equation}
Therefore, the total BPS black 0-brane cone associated to M-theory compactified on the birational family $[X_{2,86}]$ is
\begin{equation}
    \cC_{\rm BPS-BB}^{(0)}=\{(q_1,q_2)\in\mathbb{R}^2\,\vert\, q_1\geq q_2\geq 0\}\,.
\end{equation}
This matches the result obtained in~\cite{Alim:2021vhs}. Using~\cite[Table 4]{Alim:2021vhs}, which we have slightly expanded in \cref{tab:GV X286} using our numerical implementation of the genus-0 GV invariants calculation for favorable CICYs, we observe that the genus-0 GV invariants are non-vanishing for integral sites within $\cC_{\rm BPS-BB}^{(0)}$. Furthermore, the genus-0 GV invariants are also non-vanishing at integral sites within $\cC_{\rm BPS-B}^{(0)}$.

\begin{figure}[!tp]
    \centering
    \begin{subfigure}{.475\textwidth}
        \includegraphics[width=\linewidth]{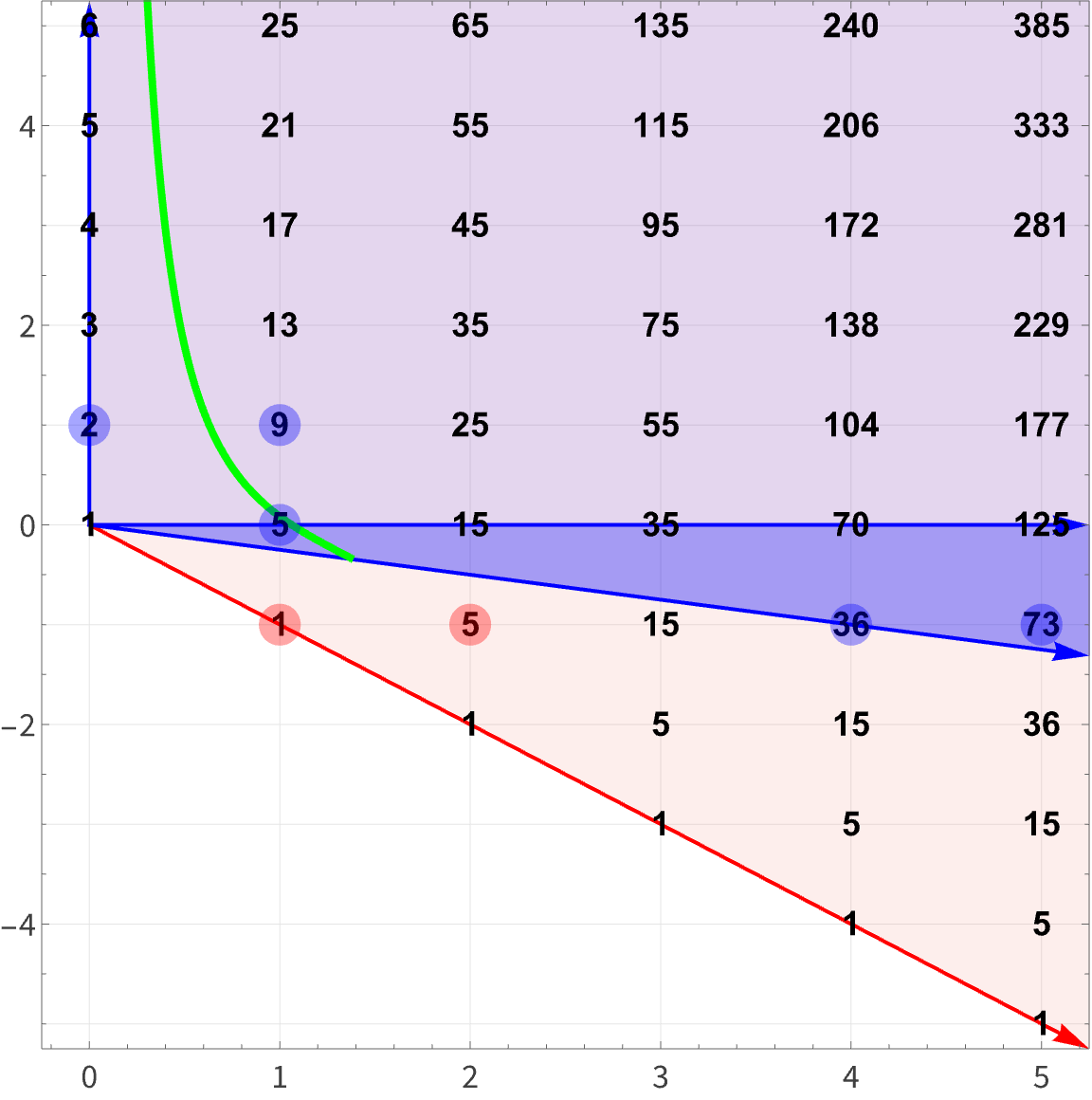}
        \begin{picture}(0,0)\vspace*{-1.2cm}
            \put(-20,140){$t^2$}
            \put(120,0){$t^1$}
        \end{picture}
        \vspace*{-0.25cm}
        \caption{$\cC_{\rm BPS-B}^{(1)}$ and $\cC_{\rm BPS-BB}^{(1)}$}
        \label{fig:X286cones}
    \end{subfigure}
    \hfill
    \begin{subfigure}{.475\textwidth}
        \includegraphics[width=\linewidth]{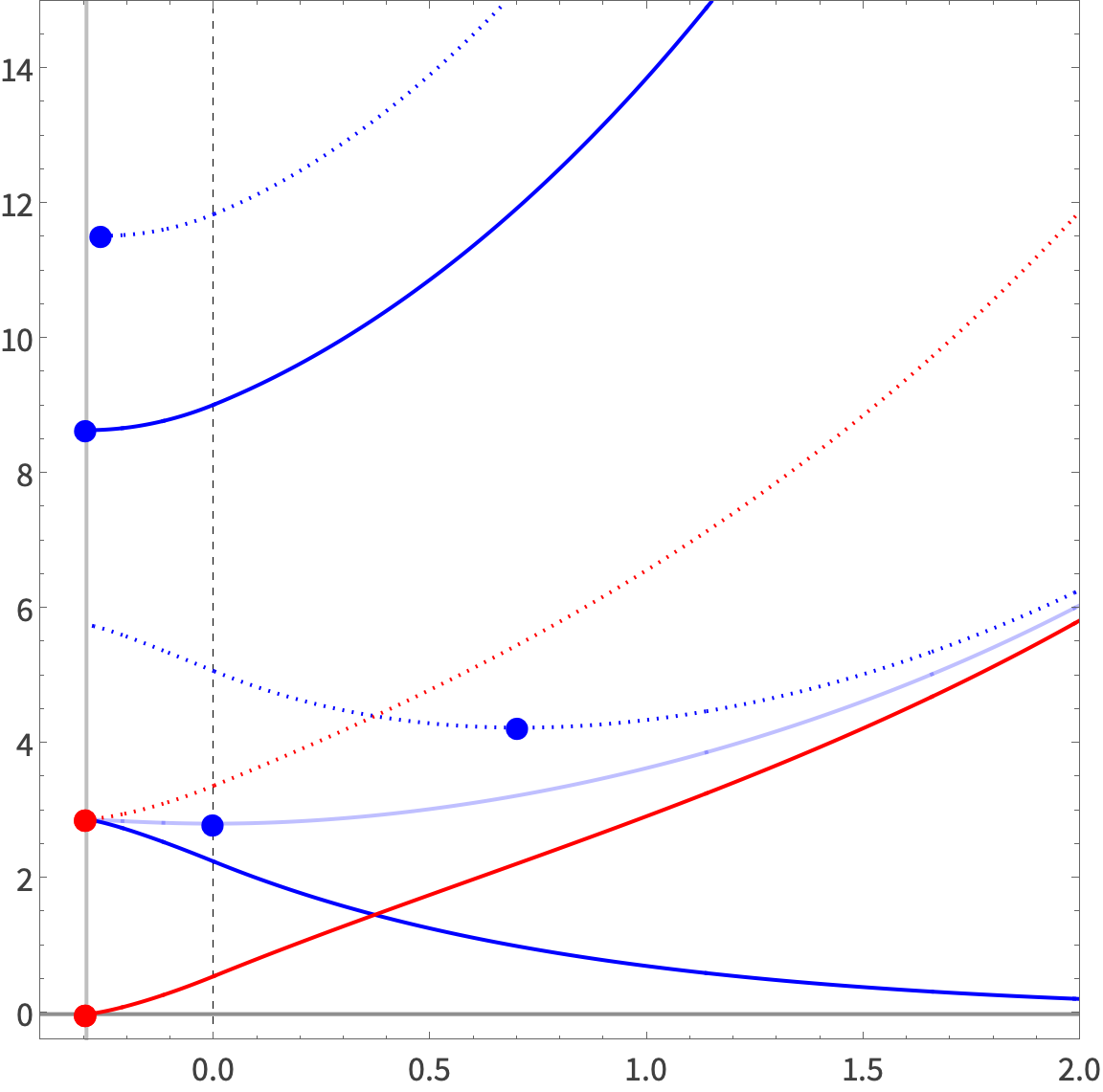}
        \begin{picture}(0,0)\vspace*{-1.2cm}
            \put(-10,140){$T$}
            \put(120,0){$\Delta$}
            \put(10,245){\footnotesize $\Delta_{\rm CFT}$}
            \put(35,245){\footnotesize $\Delta_{\rm flop}$}
            \put(90,245){\footnotesize \textcolor{blue}{$(5,-1)$}}
            \put(135,245){\footnotesize \textcolor{blue}{$(4,-1)$}}
            \put(190,200){\footnotesize \textcolor{red}{$(2,-1)$}}
            \put(50,105){\footnotesize \textcolor{blue}{$(1,1)$}}
            \put(60,75){\footnotesize \textcolor{blue}{$(1,0)$}}
            \put(190,83){\footnotesize \textcolor{red}{$(1,-1)$}}
            \put(193,40){\footnotesize \textcolor{blue}{$(0,1)$}}
        \end{picture}
        \vspace*{-0.25cm}
        \caption{Tension of BPS 1-branes}
        \label{fig:X286tension}
    \end{subfigure}
    \caption{The properties of BPS 1-branes in M-theory compactified on the birationl family $[X_{2,86}]$. In \cref{fig:X286cones}, we show $\cC_{\rm BPS-B}^{(1)}$ in red and $\cC_{\rm BPS-BB}^{(1)}$ in different shades of blue, associated to different phases of the 5d $\cN=1$ quantum gravitational theory. In particular, the lighter blue region indicates $\cC_{\rm BPS-BB}^{(1,I)}$ and the darker blue region indicates $\cC_{\rm BPS-BB}^{(1,II)}$. Additionally, we overlay the values of $h^0(X_{2,86},\cO_{X_{2,86}}(D))$ for divisor classes in $\cC_{\rm BPS-B}^{(1)}$. The constant-volume slice is illustrated in green across the extended K\"ahler cone. We also highlight the divisor classes whose tensions are studied in the right figure. In \cref{fig:X286tension}, we illustrate the tension of BPS 1-branes as a function of the geodesic distance along the constant-volume slice in each phase of the moduli space. The constant-volume slice in the second phase of the extended K\"ahler cone is a finite interval $[\Delta_{\rm CFT},\Delta_{\rm flop}]$. The flop wall is chosen to be at $\Delta=0$ and is indicated as the gray dashed vertical line, while the CFT wall is illustrated as a solid gray line at $\Delta\approx -0.295$. The divisor classes are specified in the basis $(p^1,p^2)\in \cC_{\rm BPS-B}^{(1)}$ and the global minima of their tensions are indicated by the solid dots. The dotted lines indicate divisor classes lying strictly in the interior of either $\cC_{\rm BPS-BB}$ or $\cC_{\rm BPS-B}^{(1)}\backslash \cC_{\rm BPS-BB}^{(1)}$ while the solid lines indicate divisor classes along the boundaries of $\cC_{\rm BPS-BB}^{(1,I)}$, $\cC_{\rm BPS-BB}^{(1,II)}$, and $\cC_{\rm BPS-B}^{(1)}$.}
    \label{fig:X286}
\end{figure}

Let us consider the BPS 1-branes in the theory. In addition to the prime toric divisors inherited from the ambient toric variety, there is an additional finite-distance CFT boundary generated by the exceptional $\mathbb{P}^2$ divisor $E=H-L$. This divisor shrinks to zero size at the CFT boundary as can be seen by the solid red curve in \cref{fig:X286tension}, hence giving rise to a CFT sector in the quantum gravity theory. Therefore, the BPS 1-brane cone, or equivalently the effective cone, is
\begin{equation}
    \cC_{\rm BPS-B}^{(1)}=\cE(X_{2,86})=\big\{(p^1,p^2)\in \mathbb{R}^2\,\vert\, p^1\geq 0\,,p^1\geq -p^2\big\}\,.
\end{equation}
The attractor potential for the BPS 1-branes then takes on the following form in the first phase of the extended K\"ahler cone
\begin{equation}
    V_{\rm BS}^{(I)}=\frac{3}{2(t^1)^2(5t^1+12t^2)^2}\bigg(\left(25(p^1)^2+40p^1p^2+48(p^2)^2\right)(t^1)^2+80(p^1)^2t^1t^2+96(p^1)^2(t^2)^2\bigg)\,.
\end{equation}
Similarly, we can also express $V_{\rm BS}^{(II)}$ using $\cF^{(II)}$ in \cref{eq:CICY prepot II}. In both cases, the BPS black 1-brane solution becomes simply 
\begin{equation}
    p^I\propto t^I\,.
\end{equation}
Therefore, the total BPS black 1-brane cone is the extended K\"ahler cone given in \cref{eq:CICY extended Kahler cone}. We can compute $h^0(X_{2,86},\cO_{X_{2,86}}(D))$ for divisor classes $[D]\in H^2(X_{2,86},\mathbb{Z})$ using e.g.,~\verb+pyCICY+~\cite{Larfors:2019sie}. The results are shown in \cref{fig:X286cones} where we can see that all integral sites within $\cC_{\rm BPS-B}^{(1)}$ are occupied by BPS states. 

Having identified and illustrated all relevant cones in \cref{tab:GV X286} and \cref{fig:X286cones}, we observe the following duality between cones
\begin{equation}
\label{eq:X286 duality}
    \cC_{\rm BPS-B}^{(0)}=(\cC_{\rm BPS-BB}^{(1)})^\vee\,,\qquad \cC_{\rm BPS-BB}^{(0)}=(\cC_{\rm BPS-B}^{(1)})^\vee\,.
\end{equation}

Let us analyze the masses of BPS 0-branes in $\cC_{\rm BPS-B}^{(0)}$. Here, we emphasize that $\cC_{\rm BPS-B}^{(0)}$ is the restricted Mori cone, which is the appropriate cone dual to the extended K\"ahler cone, as shown in \cref{eq:X286 duality}. The BPS 0-branes we analyze are shown as dots in \cref{tab:GV X286}. We observe that for the charges within $\cC_{\rm BPS-BB}^{(0)}$ whose masses are represented as blue dotted lines in \cref{fig:X286 mass}, the minima of the masses are located inside the extended K\"ahler moduli space. For the BPS 0-brane with charges located along $\partial \cC_{\rm BPS-BB}^{(0,I)}\cap \partial\cC_{\rm BPS-BB}^{(0,II)}$, the masses are minimized along the flop wall. This is exemplified by the charge $(5,4)$ in \cref{fig:X286 mass}. As the physical moduli space is $K_{\infty}(X_{2,86})$, this finite-distance flop wall is lies in the strict interior of the extended moduli space. Thus, all particles with charges in the strict interior of $\cC_{\rm BPS-BB}^{(0)}$ have masses that are minimized in the strict interior of the extended moduli space. For the BPS 0-branes located on the boundary of $\cC_{\rm BPS-BB}^{(0)}$, their masses are minimized on the boundaries of $K_{\infty}(X_{2,86})$ as shown by the solid blue curves associated with the charges $(1,0)$ and $(1,1)$. The mass of the particle with charge $(1,1)$ is minimized along the CFT wall with a positive minimum mass, while the mass of the particle with charge $(1,0)$ vanishes along the infinite-distance wall of the moduli space. For particles with charges in the interior of $\cC_{\rm BPS-B}^{(0)}\backslash \cC_{\rm BPS-BB}^{(0)}$, the masses are minimized along the CFT wall and have positive minimum masses inherited from their decomposition into charges in $\cC_{\rm BPS-BB}^{(0)}$ evaluated along the CFT wall. Lastly, for charges lying on the boundary of $\cC_{\rm BPS-B}^{(0)}$ but outside $\cC_{\rm BPS-BB}^{(0)}$, their masses vanish along the CFT wall.

Let us next analyze the tensions of various BPS 1-branes in the spectrum of this theory. The tensions of the relevant distinct classes of BPS 1-branes are shown in \cref{fig:X286tension} and each reflect the behaviors discussed in \cref{sec:tension properties}. First, despite the presence of two distinct phases of the extended moduli space, BPS 1-brane states in $\cC_{\rm BPS-BB}^{1,I}$ or $\cC_{\rm BPS-BB}^{(II)}$ have global minima strictly in the interior of the corresponding phases of the K\"ahler moduli space, $K_I$ or $K_{II}$, respectively. This is exemplified by the divisor classes $(5,-1)$ and $(1,1)$. 
Next, in this theory, there are three types of boundaries associated with $\cC_{\rm BPS-BB}^{(1)}$: 1) the infinite-distance boundary along $(0,1)$ where we encounter a weakly-coupled heterotic string, 2) the finite-distance flop wall connecting $\cC_{\rm BPS-BB}^{(1,I)}$ and $\cC_{\rm BPS-BB}^{(1,II)}$ along $(1,0)$; and 3) the finite-distance CFT boundary along $(4,-1)$. In each case, the corresponding BPS 1-brane tension vanishes along the associated boundary/wall of the extended moduli space. 
Different from the bicubic example studied in \cref{sec:bicubic}, in $X_{2,86}$, we encounter a strict inclusion of cones $\cC_{\rm BPS-BB}^{(1)}\subset \cC_{\rm BPS-B}^{(1)}$. Thus, another class of BPS 1-branes consists of those with charges lying in $\cC_{\rm BPS-B}^{(1)}\backslash \cC_{\rm BPS-BB}^{(1)}$. In particular, those with charges in the strict interior of this region acquire their minimum tension along the boundaries of the K\"ahler moduli space. For this particular geometry, all such interior BPS 1-branes have minimum tension along the CFT wall.
Lastly, this region also contains BPS 1-branes on the boundary of $\cC_{\rm BPS-B}^{(1)}$. These are BPS 1-brane arising from M5-branes wrapping the shrinking $\mathbb{P}^2$. In this case, all such BPS 1-branes become tensionless as we approach the CFT wall. In particular, at the CFT wall, the tension of the BPS 1-brane with magnetic charge $(1,-1)$ has a stable vanishing minimum. Hence, the divisor classes in the strict interior of $\cC_{\rm BPS-B}^{(1)}\backslash \cC_{\rm BPS-BB}^{(1)}$ also have stable tension minima along this CFT wall.

As a last remark, we observe that the primitive BPS generators of each cone displayed in \cref{fig:X286 mass} and \cref{fig:X286tension} attain minimum tensions $T_{\rm min}\lesssim 1$, consistent with the elementary constituent conjecture~\cite{Nevoa:2025xiq}.

\subsubsection{Toric Calabi--Yau threefold with boundary holes}
\label{sec:Calabi--Yau threefold with boundary holes}

Let us consider another Calabi--Yau threefold with $h^{1,1}=2$, obtained from a triangulation of a toric four-fold in the Kreuzer--Skarke database~\cite{Kreuzer:2000xy}. This is studied in~\cite[Example 1]{Gendler:2026uux} where non-BPS holes appear in the toric effective cone of the Calabi--Yau threefold $X_{2,106}$. We continue to work in the default charge basis for the divisor classes determined by \verb+CYTools+ in which the GLSM charge matrix for the toric fourfold is
\begin{equation}
\label{eq:GLSM h112}
    \begin{pmatrix}
        1&1&1&3&0&2\\
        0&0&0&1&1&-2
    \end{pmatrix}\,.
\end{equation}
In this basis, the K\"ahler form can be expressed as $J=t^1[D_1]+t^2[D_2]$. The prepotential of the Calabi--Yau threefold is
\begin{equation}
    \cF=\frac16\big((t^1)^3+3(t^1)^2t^2-9(t^1)(t^2)^2+9(t^2)^3\big)\,,
\end{equation}
where the K\"ahler cone is
\begin{equation}
    K(X_{2,106})=\big\{(t^1,t^2)\in \mathbb{R}^2\,\big\vert\, t^2>0\,, t^1-3t^2>0\big\}\,.
\end{equation}
Using \verb+CYTools+~\cite{Demirtas:2022hqf}, the K\"ahler cone is obtained as the dual of the Mori cone, which can be determined from genus-0 GV invariants~\cite{Gendler:2022ztv}. 
In this case, the genus-0 GV invariants displayed in \cref{tab:GV X2106}, determine the Mori cone. Consequently, the BPS-B cone for BPS 0-branes in this 5d $\cN=1$ theory becomes
\begin{equation}
    \cC_{\rm BPS-B}^{(0)}=M(X_{2,106})=\big\{(q_1,q_2)\in\mathbb{R}^2\,\big\vert\, q_1\geq 0\,, 3q_1+q_2\geq 0\big\}\,.
\end{equation}
In this example, there is a finite-distance CFT wall located at $t^1-3t^2=0$, while a symmetric flop wall lies along $t^2=0$. As the symmetric flop wall connects two isomorphic Calabi--Yau threefolds, we can identify the K\"ahler cones across the two phases with each other. Hence, the relevant extended K\"ahler cone consists of a single chamber as was identified in~\cite{Gendler:2026uux}.

From the prepotential, we compute the effective potential for black 0-branes of electric charge $q=(q_1,q_2)\in M(X_{2,106})$
\begin{multline}
    V_{\rm BH}=\frac{1}{12t^1(t^1-3t^2)}\bigg[(9q_1^2-2q_1q_2+q_2^2)(t^1)^4+4(-9q_1^2+6q_1q_2+q_2^2)(t^1)^3t^2\\
    +6(9q_1^2-6q_1q_2+q_2^2)(t^1)^2(t^2)^2-36(3q_1^2+2q_1q_2+q_2^2)t^1(t^2)^3+9(3q_1+q_2)^2(t^2)^4\bigg]\,.
\end{multline}
Introducing the positive coordinate in the K\"ahler cone $z=(t^1-3t^2)/t^2>0$, the BPS attractor equation $\partial_iZ_e=0$ (restricted to the constant-volume slice) reduces to the following quadratic equation
\begin{equation}
    \bigg(\frac{q_1}{q_2}-1\bigg)z^2-8z-12=0\,.
\end{equation}
Requiring $z\geq 0$ (imposed by the closure of the K\"ahler cone) selects the charges in the Mori cone satisfying $q_1\geq q_2\geq 0$. Therefore, the BPS black 0-brane cone is
\begin{equation}
    \cC^{(0)}_{\rm BPS-BB}=\big\{(q_1,q_2)\in \mathbb{R}^2\,\big\vert\,q_1\geq q_2\geq 0\big\}\,.
\end{equation}
Examining the red region of \cref{tab:GV X2106}, we see indeed all such curve classes within $\cC_{\rm BPS-BB}^{(0)}$ are populated by BPS states, as indicated by the non-vanishing genus-0 GV invariants. 

\begin{table}[!tp]
\begin{adjustwidth}{-1in}{-1in}
    \centering
    \begin{tabular}{c|c|c|c|c|c|c}
        \diagbox[width=1.5cm,height=0.7cm]{$q_2$}{$q_1$} & 0 & 1 & 2 & 3 & 4 & 5 \\
        \hline
        5 & \Y $0$ & \Y $3$ & \Y $5297640$ & \Y $130884557151368$ & \Y $80340977343000471528$ & \Pnk $24515410461322092913996250$\\
        \hline
        4 & \Y $0$ & \Y $-80$ & \Y $558340176$ & \Y $409802802458720$ & \Pnk $106512530106827231360$ & \Pnk $19465866403699147683379920$\\
        \hline
        3 & \Y $0$ & \Y $780$ & \Y $3032521320$ & \Pnk $588796074267170$ & \Pnk $80340977343000471528$ & \Pnk $9622944909425877495302600$\\
        \hline
        2 & \Y $0$ & \Y $54192$ & \Pnk $5024265100$ & \Pnk $409802802458720$ & \Pnk $33668587006610357040$ & \Pnk $2837663700463208952866320$\\
        \hline
        1 & \Y $40$ & \Pnk $121410$ & \Pnk $3032521320$ & \Pnk $130884557151368$ & \Pnk $7190639867765075200$ & \Pnk $453780874656629812634630$\\
        \hline
        0 & \Pnk \textasteriskcentered & \Pnk $54192$ & \Pnk $558340176$ & \Pnk $1521408917328$ & \Pnk $617023212218998848$ & \Pnk $30974618236891196250320$ \\
        \hline
        $-1$ & 0 & \Y $780$ & \Y $5297640$ & \Y $118586805820$ & \Y $4140467079805800$ & \Y $175116010682266925040$\\
        \hline
        $-2$ & $0$ & \Y $-80$ & \Y $-425600$ & \Y $-8543274080$ & \Y $-310549789827952$ & \Y $-15325178180082450880$\\
        \hline
        $-3$ & $0$ & \Y $3$ & \Y $29460$ & \Y $371614680$ & \Y $7824681548160$ & \Y $207112452710480780$\\
        \hline
        $-4$ & $0$ & $0$ & \Y $-3120$ & \Y $-41185408$ & \Y $-870431709600$ & \Y $-26144330541329920$\\
        \hline
        $-5$ & $0$ & $0$ & \Y $200$ & \Y $3953900$ & \Y $77907925400$ & \Y $2126880638515300$\\
        \hline
        $-6$ & $0$ & $0$ & \Y $-6$ & \Y $-365040$ & \Y $-7496506320$ & \Y $-180759346759680$\\
        \hline
        $-7$ & $0$ & $0$ & $0$ & \Y $27580$ & \Y $767501168$ & \Y $18963009987180$
    \end{tabular}
\tikz[baseline,remember picture,overlay]{
  \draw[red,thick, ->] (-17.70,0.25) -- (-17.70,3.10cm);
  \draw[red, thick, ->] (-17.7,0.25) -- (-14.7,-2.9
  cm);
}
\tikz[baseline,remember picture,overlay]{
  \draw[blue,thick, ->] (-17.70,0.25) -- (-9,2.40cm);
  \draw[blue, thick, ->] (-17.7,0.25) -- (-0.1,0.25cm);
  \fill[red,opacity=0.3] (-17.85,0.75) circle[radius=5pt];
  \fill[red,opacity=0.3] (-16.65,-0.25) circle[radius=5pt];
  \fill[red,opacity=0.3] (-16.65,-1.25) circle[radius=5pt];
  \fill[red,opacity=0.3] (-16.65,1.72) circle[radius=5pt];
  \fill[blue,opacity=0.3] (-16.65,0.25) circle[radius=5pt];
  \fill[blue,opacity=0.3] (-16.65,0.75) circle[radius=5pt];
  \fill[blue,opacity=0.3] (-14.65,0.75) circle[radius=5pt];
}
\end{adjustwidth}
    \caption{The genus-0 GV invariants for curve classes $(q_1,q_2)\in H_2(X_{2,106},\mathbb{Z})$. Here, the red shaded region indicate the Mori cone, while the blue shaded region indicates the cone of movable curves. The red lines indicate the extremal rays of the Mori cone, while the blue lines indicate the extremal rays of the cone of movable curves. The dotted electric states are further studied in \cref{fig:X2106mass}.}
    \label{tab:GV X2106}
\end{table}

\begin{figure}
    \centering
    \includegraphics[width=0.5\linewidth]{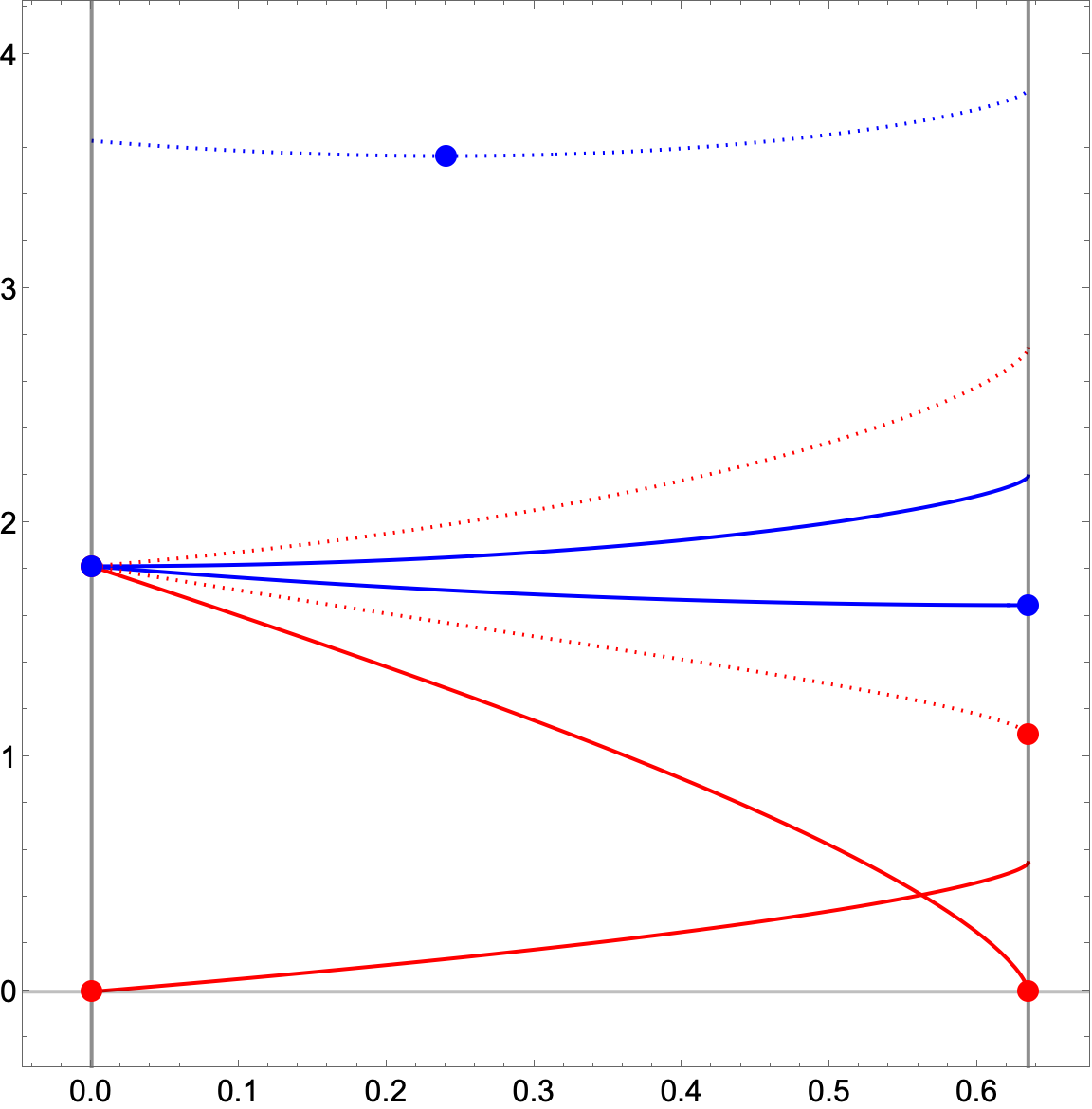}
    \begin{picture}(0,0)\vspace*{-1.2cm}
        \put(-265,120){$m$}
        \put(-120,-10){$\Delta$}
        \put(-27,250){\footnotesize $\Delta_{\rm CFT}$}
        \put(-230,250){\footnotesize $\Delta_{\rm flop}$}
        \put(-2,225){\footnotesize \textcolor{blue}{$(2,1)$}}
        \put(-2,110){\footnotesize \textcolor{blue}{$(1,0)$}}
        \put(-2,170){\footnotesize \textcolor{red}{$(1,3)$}}
        \put(-2,140){\footnotesize \textcolor{blue}{$(1,1)$}}
        \put(-2,80){\footnotesize \textcolor{red}{$(1,-1)$}}
        \put(-2,55){\footnotesize \textcolor{red}{$(0,1)$}}
        \put(-2,20){\footnotesize \textcolor{red}{$(1,-3)$}}
    \end{picture}
    \vspace*{0.25cm}
    \caption{The masses of BPS 0-branes in $\cC_{\rm BPS-B}^{(0)}$ for M-theory compactified on $X_{2,106}$, shown as functions of the canonically normalized scalar $\Delta$ parameterizing the constant-volume slice. The blue and red lines indicate particles in $\cC_{\rm BPS-BB}^{(0)}$ and $\cC_{\rm BPS-B}^{(0)}\backslash \cC_{\rm BPS-BB}^{(0)}$, respectively. The solid and dotted lines represent particles that are on the boundary of or in the interior of their respective cones, respectively. The solid dot indicates the location of their mass minima.}
    \label{fig:X2106mass}
\end{figure}
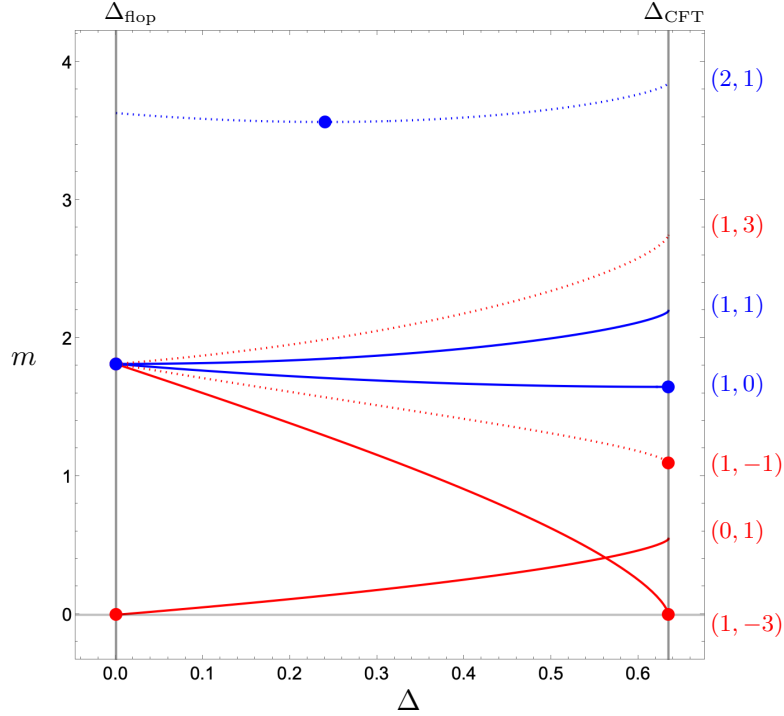

\begin{figure}[!tp]
    \centering
    \begin{subfigure}{.475\textwidth}
        \includegraphics[width=\linewidth]{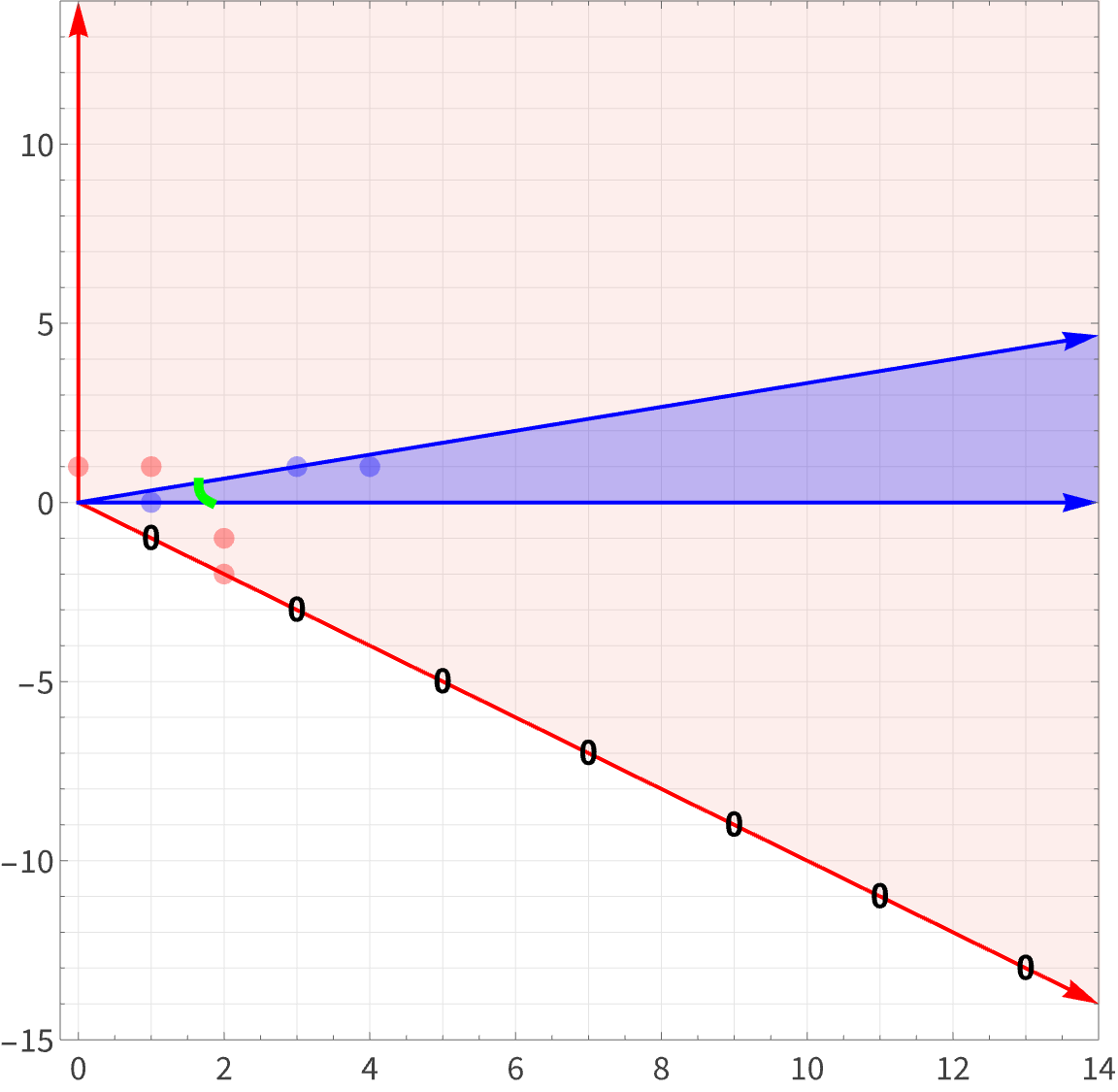}
        \begin{picture}(0,0)\vspace*{-1.2cm}
            \put(-15,125){$p^2$}
            \put(120,0){$p^1$}
        \end{picture}
        \vspace*{-0.25cm}
        \caption{$\cC_{\rm BPS-B}^{(1)}$ and $\cC_{\rm BPS-BB}^{(1)}$}
        \label{fig:X2106cones}
    \end{subfigure}
    \hfill
    \begin{subfigure}{.475\textwidth}
        \includegraphics[width=\linewidth]{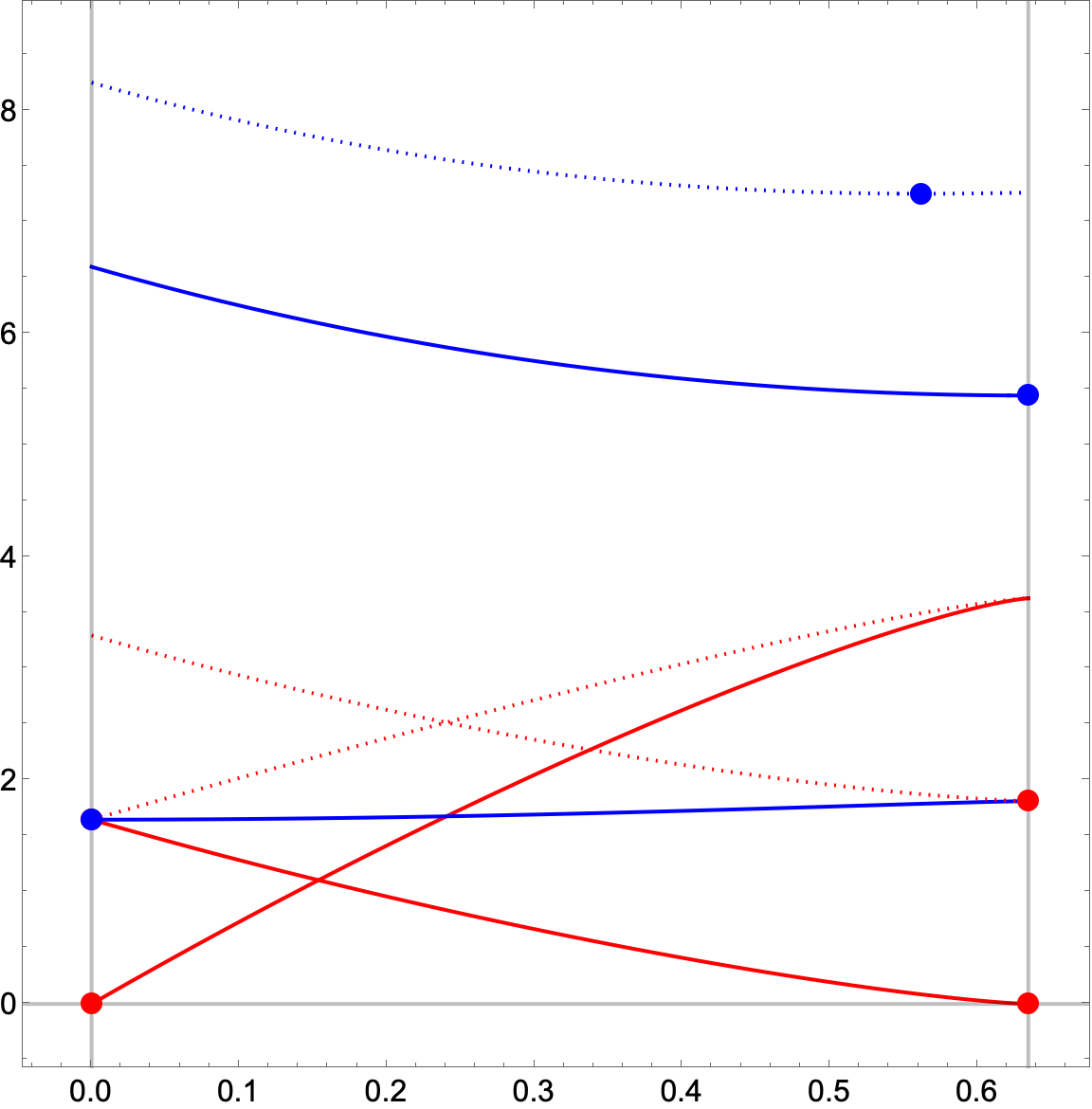}
        \begin{picture}(0,0)\vspace*{-1.2cm}
            \put(-15,125){$T$}
            \put(120,0){$\Delta$}
            \put(7,252){\footnotesize $\Delta_{\rm flop}$}
            \put(200,252){\footnotesize $\Delta_{\rm CFT}$}
            \put(190,45){\footnotesize \textcolor{red}{$(0,1)$}}
            \put(185,105){\footnotesize \textcolor{red}{$(2,-2)$}}
            \put(185,127){\footnotesize \textcolor{red}{$(2,-1)$}}
            \put(23,120){\footnotesize \textcolor{red}{$(1,1)$}}
            \put(23,85){\footnotesize \textcolor{blue}{$(1,0)$}}
            \put(190,175){\footnotesize \textcolor{blue}{$(3,1)$}}
            \put(190,220){\footnotesize \textcolor{blue}{$(4,1)$}}
        \end{picture}
        \vspace*{-0.25cm}
        \caption{Tensions of BPS 1-branes}
        \label{fig:X2106tensions}
    \end{subfigure}
    \caption{The relevant physics associated with 1-branes in M-theory compactified on $X_{2,106}$. In \cref{fig:X2106cones}, $\cC_{\rm BPS-B}^{(1)}$ is the red region and $\cC_{\rm BPS-BB}^{(1)}$ is the blue region. The sites with $h^0(X_{2,106},\cO_{X_{2,106}}(D))=0$ are labeled by 0. The remaining unlabeled integral sites in $\cC_{\rm BPS-B}^{(1)}$ are all effective. The constant-volume slice ($\cF=1$) in the K\"ahler cone is shown in green. 
    Additionally, the divisor classes chosen further analysis in \cref{fig:X2106tensions} are indicated by dots. In \cref{fig:X2106tensions}, the tensions of the divisor classes highlighted in dots in the four distinct regions across $\cC_{\rm BPS-B}^{(1)}$ are shown as a function of the canonically normalized scalar $\Delta$.
    The minima of the tensions for these four classes of BPS 1-branes are shown as solid dots. The dotted lines indicate divisor classes strictly in the interior of either $\cC_{\rm BPS-B}^{(1)}\backslash \cC_{\rm BPS-BB}^{(1)}$ or $\cC_{\rm BPS-BB}^{(1)}$, while solid lines indicate BPS 1-branes along the boundary of either $\cC_{\rm BPS-B}^{(1)}$ or $\cC_{\rm BPS-BB}^{(1)}$.}
    \label{fig:X2106}
\end{figure}

In general, determining the exact effective cone for a Calabi--Yau threefold is difficult. For toric Calabi--Yau threefolds, however, the toric effective cone is straightforward to identify. This cone is contained in the full effective cone. The extremal rays of the toric effective cone are the columns of an appropriate GLSM charge matrix. In this setting, the choice of GLSM shown in \cref{eq:GLSM h112} determines the basis of divisor classes in the toric effective cone. In this basis, the toric effective cone is
\begin{equation}
    \cE_{\rm toric}(X_{2,106})=\big\{(p_1,p_2)\in \mathbb{R}^2\,\big\vert\, p_1+p_2\geq 0\,, p_2\geq 0\big\}\,.
\end{equation}
The toric effective cone can be a proper subcone of the full effective cone of the toric Calabi--Yau threefold, due to the presence of autochthonous divisors. In the present example, however, there are no autochthonous divisors in $X_{2,106}$~\cite{Gendler:2026uux}. Therefore, $\cE_{\rm toric}(X_{2,106})=\cE(X_{2,106})$. Consequently, the BPS 1-brane cone is
\begin{equation}
    \cC^{(1)}_{\rm BPS-B}=\cE(X_{2,106})\,.
\end{equation}
The effective potential for black 1-branes with magnetic charge $p=(p^1,p^2)\in \cE(X_{2,106})$ is
\begin{multline}
\label{eq:X2106 BS potential}
    V_{\rm BS}=\frac{3}{2\big[(t^1)^3+3(t^1)^2t^2-9t^1(t^2)^2+9(t^2)^3\big]^2}\bigg[\big((p^1)^2+2p^1p^2+9(p^2)^2\big)(t^1)^4\\
    +4\big((p^1)^2-6p^1p^2-9(p^2)^2\big)(t^1)^3t^2+6(p^1+3p^2)^2(t^1)^2(t^2)^2\\
    -36\big((p^1)^2-2p^1p^2+3(p^2)^2\big)t^1(t^2)^3+9(p^1-3p^2)^2(t^2)^4\bigg]\,.
\end{multline}
The BPS black 1-brane attractor equation, $\partial_iZ_m=0$, restricted to the constant-volume slice is
\begin{equation}
    z=\frac{p^1}{p^2}-3\,.
\end{equation}
Therefore, imposing $z>0$, the BPS black 1-brane cone is
\begin{equation}
    \cC^{(1)}_{\rm BPS-BB}=\big\{(p^1,p^2)\in \mathbb{R}^2\,\big\vert\, p^1\geq 3p^2\geq 0\big\}\cong \Nef(X_{2,106})\,.
\end{equation}

Having identified the BPS-B cones and the BPS-BB cones for both the 0-branes and the 1-branes, we see that the following cone duality is realized geometrically
\begin{equation}
    \cC_{\rm BPS-B}^{(0)}=(\cC_{\rm BPS-BB}^{(1)})^\vee\,,\qquad \cC_{\rm BPS-BB}^{(0)}=(\cC_{\rm BPS-B}^{(1)})^\vee\,.
\end{equation}
Additionally, the non-BPS holes discussed in~\cite{Gendler:2026uux} and also illustrated in \cref{fig:X2106cones} are
\begin{equation}
    [D_{\rm non-BPS}]=(2n+1,-(2n+1))\,,
\end{equation}
where $n\in\mathbb{Z}_{\geq 0}$. These all lie outside of $\cC_{\rm BPS-BB}^{(1)}$ and also strictly on the boundary of $\cC_{\rm BPS-B}^{(1)}$.

Let us analyze the masses of BPS 0-branes in this theory. The four classes of BPS 0-branes are represented by the dots indicated in \cref{tab:GV X2106}. Their masses, as a function of the canonically normalized scalar parameterizing the constant-volume slice in the K\"ahler cone, are shown in \cref{fig:X2106mass}. The BPS 0-branes chosen to represent the general behavior of BPS 0-branes in $\cC_{\rm BPS-BB}^{(0)}$ has charge $(2,1)$ and is shown as the blue dotted curve in \cref{fig:X2106mass}. Its mass is indeed minimized at a positive value in the interior of the K\"ahler cone. For BPS 0-branes on the boundaries of $\cC_{\rm BPS-BB}^{(0)}$, the masses are represented by solid blue curves. Their respective minima are located on the boundaries of the K\"ahler moduli space and remain positive. A similar behavior occurs for BPS 0-branes in the strict interior of $\cC_{\rm BPS-B}^{(0)}\backslash \cC_{\rm BPS-BB}^{(0)}$, as shown by the red dotted curves. Their non-zero minima are also located on the boundary of the K\"ahler moduli space. For BPS 0-branes along the boundaries of $\cC_{\rm BPS-B}^{(0)}$, their masses vanish at one of the two boundaries of the K\"ahler moduli space.

Finally, let us analyze the tensions of the four distinct classes of divisors as outlined in \cref{sec:tension properties}. We begin with the BPS 1-branes whose charges lie strictly in the interior of $\cC_{\rm BPS-BB}^{(1)}$. An example is the divisor class $[D]=(4,1)$, shown as the blue dotted curve in \cref{fig:X2106tensions}. Its tension is minimized in the interior of the K\"ahler moduli space. 
Next, let us consider divisor classes on $\partial \cC_{\rm BPS-BB}^{(1)}$. The primitive divisors along the boundaries of $\cC_{\rm BPS-BB}^{(1)}$ are $[D]=(1,0)$ and $[D]=(3,1)$. The derivatives of their tensions also vanish at the corresponding K\"ahler cone boundaries, indicating that these boundary values are stable minima in moduli space. 
For divisors in $\cC_{\rm BPS-B}^{(1)}\backslash \cC_{\rm BPS-BB}^{(1)}$, we observe two distinct behaviors illustrated via the two red dotted lines in \cref{fig:X2106tensions}. While the tensions of both divisors are minimized along the boundaries of the K\"ahler moduli space, the derivative of the tension for the divisor $(1,1)$ vanishes at its minimum, while for the divisor $(2,-1)$ it is non-zero. Therefore, the divisor class $(1,1)$ has a stable minimum along the boundary of the K\"ahler cone, while the minimum tension of the divisor class $(2,-1)$ is not stable. As can be seen, their behaviors are directly correlated with those of the divisors along the boundaries of $\cC_{\rm BPS-B}^{(1)}$.
At last, the tension of divisors along the boundaries of $\cC_{\rm BPS-B}^{(1)}$ acquire their minimum value, which is zero, at the boundaries of the K\"ahler moduli space. 

In this example, the primitive generators of both $\cC_{\rm BPS-B}$ and $\cC_{\rm BPS-BB}$ for BPS 0-branes and BPS 1-branes all have minimum tensions $\lesssim 1$.

As a last remark, consider the non-BPS 1-brane with charge $p=(1,-1)$. The central charge associated to this 1-brane is
\begin{equation}
    Z_p=2t^2(2t^1-3t^2)=2(2z+3)\bigg(\frac{6}{z^3+12z^2+36z+36}\bigg)^{2/3}\,,
\end{equation}
and it vanishes along the flop wall. Hence, $Z_{\rm min}=0$. The effective black string potential given in \cref{eq:X2106 BS potential} associated to this charge is
\begin{equation}
    V_{\rm BS}(p)=2\cdot 6^{1/3}\frac{z^4+11z^3+48z^2+72z+18}{(z^3+12z^2+36z+36)^{4/3}}\,,
\end{equation}
where the minimum is located along the CFT wall. Hence, the tension of the would-be extremal black brane with charge $p$ in this theory is $T^{\rm Ext}=\sqrt{\mathrm{min}(V_{\rm BS})}=6^{-1/6}$. For this family of non-BPS 1-branes with charges $(2n+1,-(2n+1))$, where $n\in\mathbb{Z}_{\geq 0}$, the minimum central charges and tensions of would-be extremal black branes are $Z_{\rm min}=0$ and $T^{\rm Ext}=(2n+1)6^{-1/6}$. Hence, the non-BPS 1-branes in this theory satisfy the inequality
\begin{equation}
    Z_{\rm min}<T^{\rm Ext}\,.
\end{equation}
This is consistent with \cref{def:BPS occupied 2}.

\subsubsection{Toric Calabi--Yau threefold with interior holes}
\label{sec:Calabi--Yau threefold with interior holes}

As a last example, let us consider the toric Calabi--Yau threefold with $h^{1,1}=3$ from the Kreuzer--Skarke database~\cite{Kreuzer:2000xy} and also studied in \cite[Example 3]{Gendler:2026uux}.\footnote{We are grateful to Elijah Sheridan for relevant discussions on this particular Calabi--Yau threefold.} We denote this Calabi--Yau threefold as $X_{3,165}$. In this example, non-BPS holes appear both along the boundary and in the strict interior of $\cC_{\rm BPS-B}^{(1)}$. Again, working in the default charge basis for the divisor classes determined by \verb+CYTools+, the GLSM charge matrix associated with the toric fourfold is 
\begin{equation}
    \begin{pmatrix}
        1&1&1&3&0&6&0\\
        0&0&0&1&1&-2&0\\
        0&0&0&1&0&2&1
    \end{pmatrix}\,.
\end{equation}
Additionally, the toric effective cone is strictly contained in the effective cone of the Calabi--Yau threefold due to the presence of autochthonous divisors~\cite{Gendler:2026uux}. The extended K\"ahler cone of the birational family associated to this Calabi--Yau threefold has a single chamber. In the basis of the GLSM, the effective cone of $X_{3,165}$ is generated by the following extremal rays, represented by columns of the following matrix
\begin{equation}
    \cE(X_{3,165})=\mathrm{Cone}
    \begin{pmatrix}
        1&0&0&3&6\\
        0&1&0&-1&2\\
        0&0&1&1&-2
    \end{pmatrix}\,,
\end{equation}
where the last column is the primitive effective charge class of the autochthonous divisor. In this basis, the K\"ahler two-form is $J=t^1[D_1]+t^2[D_2]+t^3[D_3]$ and the prepotential is
\begin{equation}
    \cF=\frac12\bigg[(t^1)^2(t^2+t^3)-3t^1\big((t^2)^2+(t^3)^2\big)+3(t^2)^3+3(t^3)^3\bigg]\,,
\end{equation}
where the K\"ahler cone is
\begin{equation}
\label{eq:X3165 KahlerCone}
    K(X_{3,165})=\big\{(t^1,t^2,t^3)\in\mathbb{R}^3\,\vert\, t^2,t^3> 0\,,t^1-3t^2> 0\,, t^1-3t^3> 0\big\}\,.
\end{equation}
The K\"ahler cone has 8 boundaries associated to its four walls and four extremal rays shown in \cref{fig:X3165 1branes}. The dual primitive generators in the Mori cone are
\begin{equation}
    [C_2]=(0,1,0)\,,\qquad [C_3]=(0,0,1)\,,\qquad [C_4]=(1,-3,0)\,,\qquad [C_5]=(1,0,-3)\,.
\end{equation}
Their volumes correspond to the K\"ahler cone constraints in \cref{eq:X3165 KahlerCone} and vanish along the walls of the K\"ahler cone. Starting with $C_2$ and $C_3$, their volumes vanish along the walls $t^2=0$ and $t^3=0$, respectively. Furthermore, the genus-0 GV invariants are $n_{C_2}^0=108$ and $n_{C_3}^0=108$ while $n_{mC_2}^0=0$ and $n_{mC_3}^0=0$ for all $m>1$. Thus, the shrinking curves in these primitive classes are isolated rational curves. By adjunction, their normal bundles must have degree $-2$. As the curves are isolated, the normal bundles are
\begin{equation}
    N_{C_2/X_{3,165}}\cong \cO_{\mathbb{P}^1}(-1)\oplus \cO_{\mathbb{P}^1}(-1)\,,\qquad N_{C_3/X_{3,165}}\cong \cO_{\mathbb{P}^1}(-1)\oplus \cO_{\mathbb{P}^1}(-1)\,.
\end{equation}
Hence, these are flop curves.
Upon crossing through the K\"ahler wall $t^2=0$, the prepotential in the new phase of the extended K\"ahler cone is
\begin{equation}
    \cF^{(2)}-\cF=-\frac{108}{6}(t^2)^3\,.
\end{equation}
However, we can define the integral involution
\begin{equation}
    A_2=
    \begin{pmatrix}
        1&6&0\\
        0&-1&0\\
        0&2&1
    \end{pmatrix}\,,
\end{equation}
where $A_2^2=1$ and $\det A_2=-1$. Then, the above involution $(t^1,t^2,t^3)\mapsto A_2(t^1,t^2,t^3)$ realizes $\cF^{(2)}$. Therefore, the new chamber is isomorphic to the original K\"ahler cone and the wall $t^2=0$ is a symmetric flop wall. The same analysis applies when crossing through the $t^3=0$ wall, which is again a symmetric flop wall. Namely, the new birational equivalent Calabi--Yau threefold is isomorphic to $X_{3,165}$. The other two codimension-one walls in the K\"ahler cone involve shrinking divisors. The relevant divisor classes are $[D_1]=(1,0,0),[D_2]=(0,1,0),[D_3]=(0,0,1)$ and their volumes are
\begin{align}
\begin{split}
    \volume(D_1)=t^1(t^2+t^3)-\frac32\big((t^2)^2+(t^3)^2\big)\,,\quad
    \volume(D_2)=\frac12 (t^1-3t^2)^2\,,\quad
    \volume(D_3)=\frac12(t^1-3t^3)^2\,.
\end{split}
\end{align}
We see that $D_2$ collapses along the K\"ahler wall $t^1-3t^2=0$ where $D_2^3=9$. We can identify $D_2\cong \mathbb{P}^2$, with the normal bundle determined by adjunction as
\begin{equation}
    N_{D_2/X_{3,165}}\cong \cO_{\mathbb{P}^2}(-3)\,.
\end{equation}
Considering $L\subset D_2$ as a class of a line, then we also have $D_1.L=1,D_2.L=-3,D_3.L=0$ and it also vanishes along the same K\"ahler wall at which $D_2$ shrinks to zero size. Thus, $L$ is the curve class $C_4$. As $N_{L/D_2}\cong \cO_{\mathbb{P}^1}(1)$, we obtain
\begin{equation}
    N_{L/X_{3,165}}\cong \cO_{\mathbb{P}^1}(1)\oplus \cO_{\mathbb{P}^1}(-3)\,.
\end{equation}
Hence, this is a finite-distance CFT wall. The same analysis applies to $D_3$, the corresponding line class $C_5$, and the K\"ahler wall $t^1-3t^3=0$. There also exist higher-codimension walls of the K\"ahler cone. These are constructed as intersections between the codimension-one walls. In particular, when the symmetric flop walls intersect with the CFT walls, the codimensions-two ray in the K\"ahler cone signals a mixed degeneration of a symmetric flop along with a $\mathbb{P}^2$ contraction. Similarly, when the two CFT walls intersect, the codimension-two ray signals two $\mathbb{P}^2$ divisors simultaneously collapsing to a point. When the two symmetric flop walls intersect, we obtain a qualitatively distinct boundary of the K\"ahler cone. The intersection is $t^2=t^3=0$ with $t^1\geq 0$. The prepotential indicates that this boundary lies at infinite distance in moduli space. In particular, along the constant-volume slice, it is reached in the limit $t^1\sim (t^2)^{-1/2}\sim (t^3)^{-1/2}\to\infty$. This can also be seen in the green shaded region of \cref{fig:X3165 1branes} indicating the constant-volume slice in the K\"ahler cone where the infinite-distance limit lies along $t^2=t^3=0$. In this case, the relevant non-big nef divisor is $D_1$. We observe $D_1^3=0$ while $D_1^2\neq 0$. From the emergent string conjecture~\cite{Lee:2019wij}, this is the characteristic intersection data of an elliptic fibration limit. To see this from the volume perspective, the relevant curve to consider is $C=D_1^2$. The volume is $\volume(C)=t^2+t^3$. Along this boundary, $\volume(C)\to 0$. However, the volume of $D_1$ behaves as $\volume(D_1)\sim t_1^{-1}\to 0$. Therefore, although $\volume(D_1)$ vanishes, it does so parametrically slower than $\volume(C)$. Thus, the leading light tower of states are interpreted as the KK tower, signaling a decompactification limit.

\begin{figure}[!tp]
    \centering
    \begin{subfigure}{\textwidth}
        \includegraphics[width=.475\linewidth]{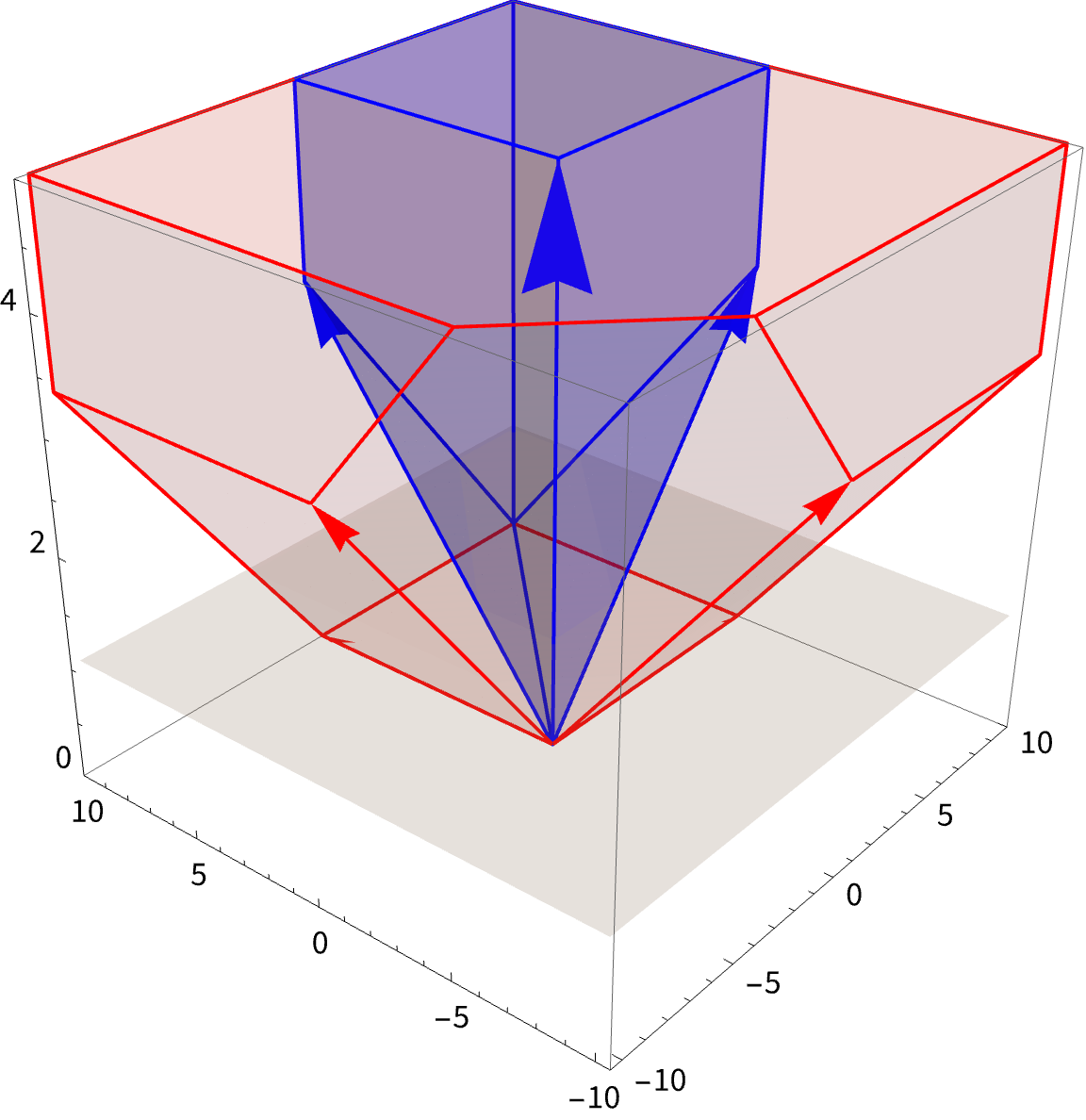}
        \hfill
        \includegraphics[width=.475\linewidth]{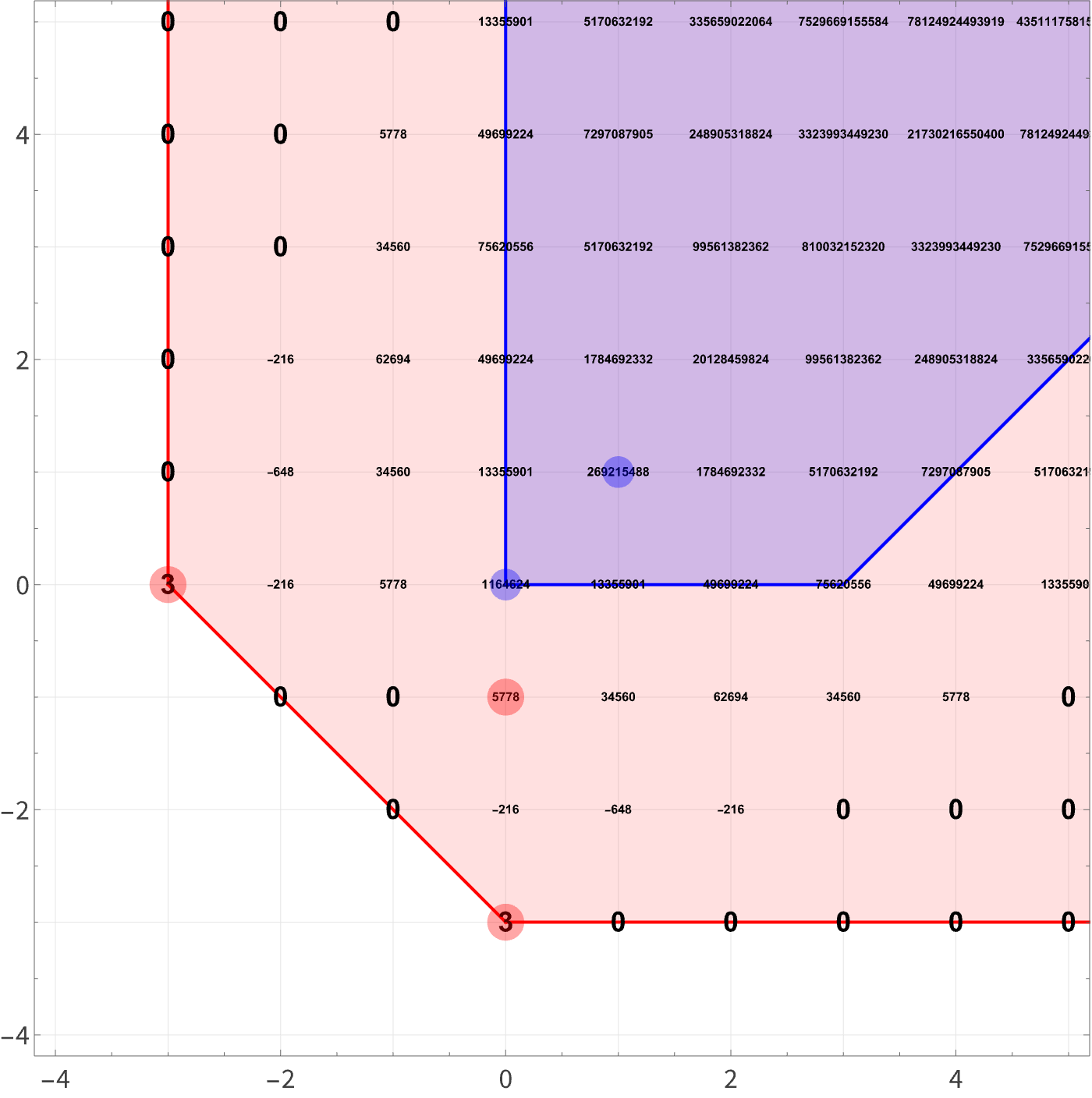}
        \begin{picture}(0,0)\vspace*{-1.2cm}
            \put(-500,120){$q_1$}
            \put(-440,30){$q_2$}
            \put(-300,30){$q_3$}
            \put(-250,120){$q_3$}
            \put(-120,-10){$q_2$}
        \end{picture}
        \caption{Cones of BPS 0-branes}
        \label{fig:X3165 0branes}
    \end{subfigure}
    \hfill
    \vspace*{0.25cm}
    \begin{subfigure}{\textwidth}
        \includegraphics[width=.475\linewidth]{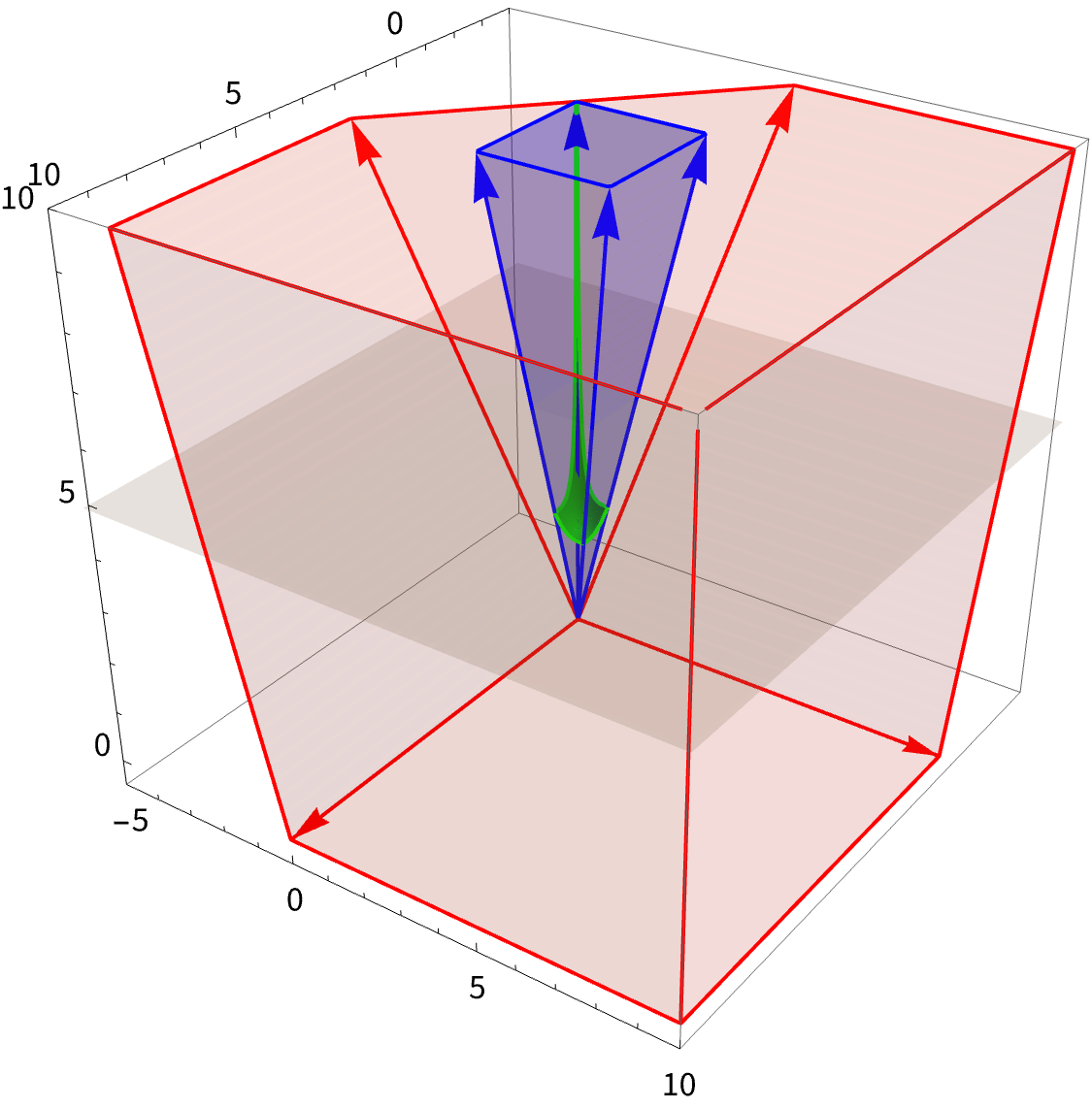}
        \hfill
        \includegraphics[width=.475\linewidth]{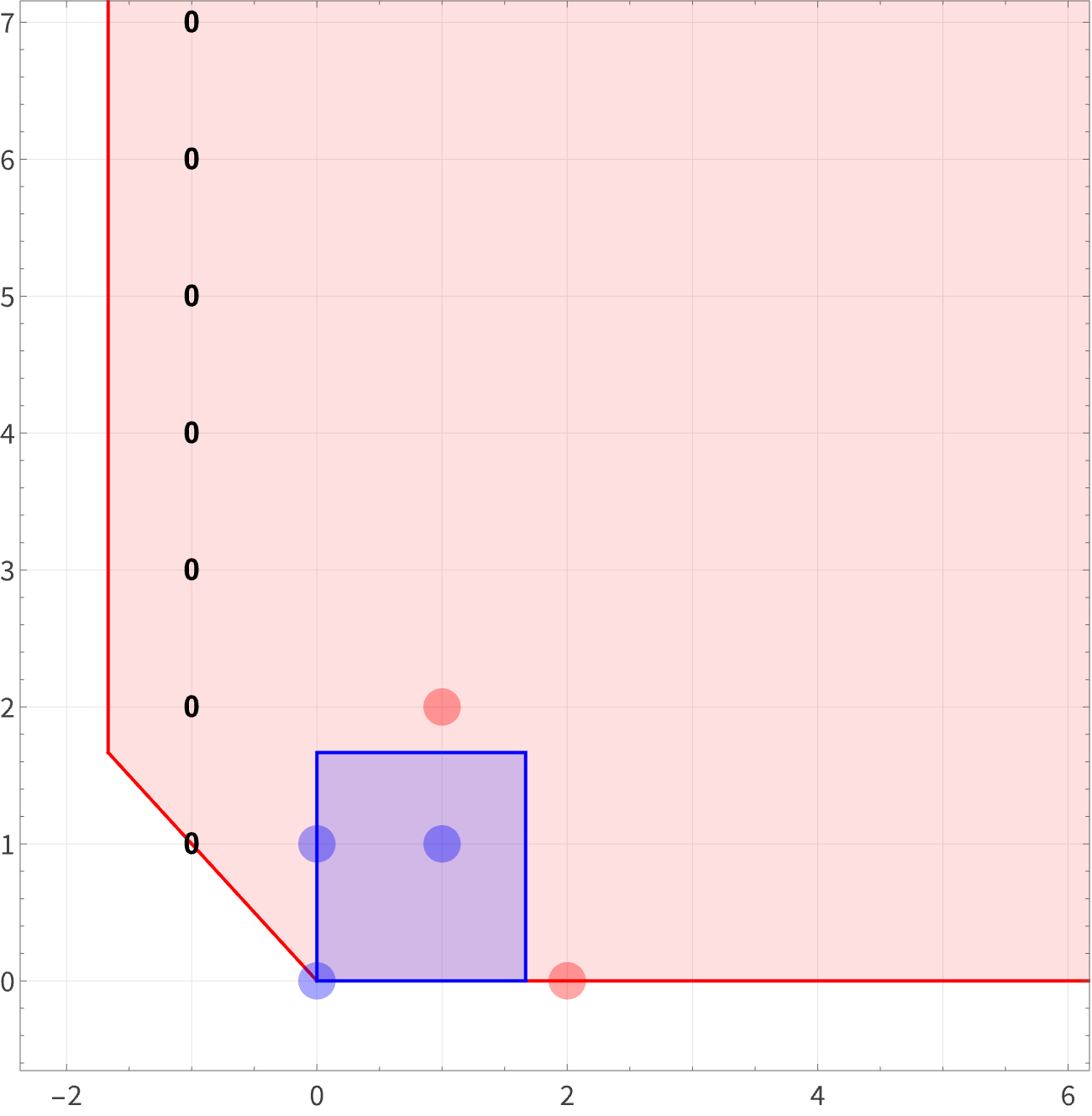}
        \begin{picture}(0,0)\vspace*{-1.2cm}
            \put(-500,105){$p^1$}
            \put(-405,20){$p^3$}
            \put(-445,220){$p^2$}
            \put(-250,120){$p^3$}
            \put(-120,-10){$p^2$}
        \end{picture}
        \vspace*{0.cm}
        \caption{Cones of BPS 1-branes}
        \label{fig:X3165 1branes}
    \end{subfigure}
    \caption{The $\cC_{\rm BPS-B}$ and $\cC_{\rm BPS-BB}$ for 0-branes and 1-branes in M-theory compactified on $X_{3,165}$. In the left figures, $\cC_{\rm BPS-B}$ and $\cC_{\rm BPS-BB}$ are illustrated as the red and blue shaded region, respectively. The arrows indicate the extremal rays of the cones. Additionally, the green two-dimensional region in \cref{fig:X3165 1branes} is the constant-volume slice in the K\"ahler cone. In the right figures, we present the invariants associated to curves and divisors along the two-dimensional gray shaded plane indicated in the corresponding left figures. In the case of 0-branes, the invariant shown is the genus-0 GV invariants along the $q_1=1$ plane. In the case of 1-branes, the invariant shown is $h^0(X_{3,165},\cO_{X_{3,165}}(D))$ along the $p^1=5$ plane. The unlabeled integral points have non-zero $h^0$ and hence correspond to effective divisor classes.}
    \label{fig:X3165}
\end{figure}

\begin{figure}[!tp]
    \centering
    \begin{subfigure}{\textwidth}
    \centering
    \includegraphics[width=\linewidth]{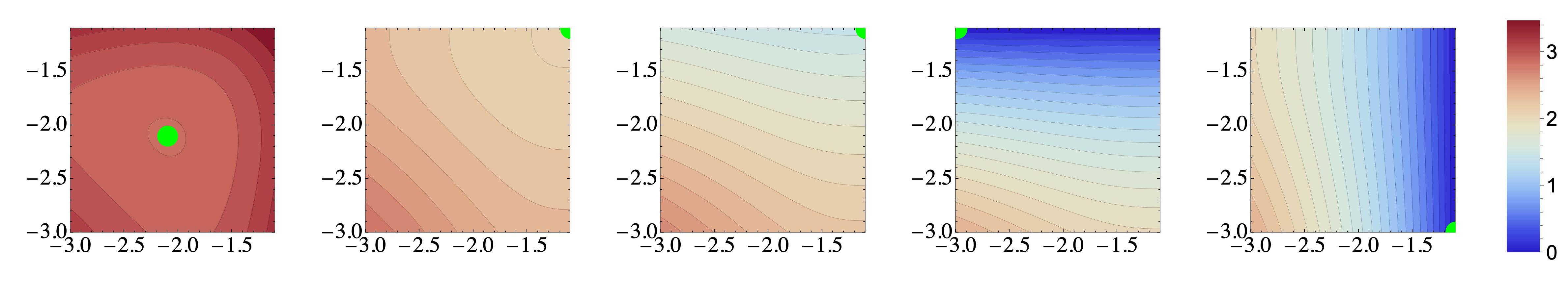}
    \begin{picture}(0,0)\vspace*{-1.2cm}
        \put(-20,10){\footnotesize $\log t^2/t^1$}
        \put(-275,60){\footnotesize $\log t^3/t^1$}
        \put(-200,110){\footnotesize $(1,1,1)$}
        \put(-115,110){\footnotesize $(1,0,0)$}
        \put(-25,110){\footnotesize $(1,0,-1)$}
        \put(67,110){\footnotesize $(1,0,-3)$}
        \put(157,110){\footnotesize $(1,-3,0)$}
        \put(225,110){\footnotesize $m$}
    \end{picture}
    \vspace*{-0.5cm}
    \caption{Mass of BPS 0-branes in $\cC_{\rm BPS-B}$}
    \label{fig:X3165Mass}
    \end{subfigure}
    \hfill
    \vspace*{0.5cm}
    \begin{subfigure}{\textwidth}
    \centering
    \includegraphics[width=\linewidth]{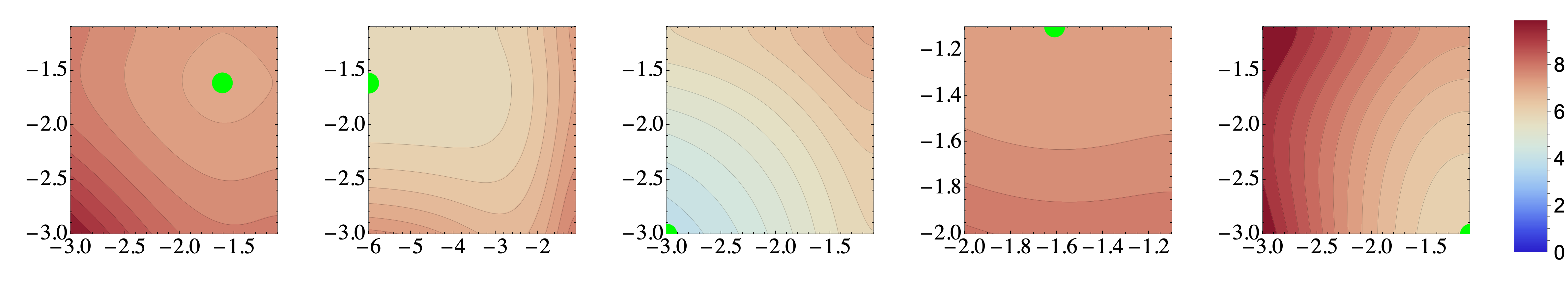}
    \begin{picture}(0,0)\vspace*{-1.2cm}
        \put(-20,10){\footnotesize $\log t^2/t^1$}
        \put(-275,60){\footnotesize $\log t^3/t^1$}
        \put(-200,110){\footnotesize $(5,1,1)$}
        \put(-115,110){\footnotesize $(5,0,1)$}
        \put(-25,110){\footnotesize $(5,0,0)$}
        \put(67,110){\footnotesize $(5,1,2)$}
        \put(157,110){\footnotesize $(5,2,0)$}
        \put(228,110){\footnotesize $T$}
    \end{picture}
    \vspace*{-0.5cm}
    \caption{Tension of BPS 1-branes in $\cC_{\rm BPS-B}$}
    \label{fig:X3165Tensions}
    \end{subfigure}
    \caption{The moduli dependence of BPS brane masses and tensions in $\cC_{\rm BPS-B}^{(0)}$ and $\cC_{\rm BPS-B}^{(1)}$, respectively. The BPS 0-branes and BPS 1-branes chosen here are highlighted in the right panel of \cref{fig:X3165}. The locations of the minima of their masses and tensions are shown as the green dot (restricted to the plotting region).}
    \label{fig:X3165 mass and tension}
\end{figure}

With the above physics along the boundaries of the moduli space established, we turn to the spectrum of BPS states in the theory. The BPS 0-brane cone and the BPS 1-brane cone can be identified with the Mori cone and the effective cone, respectively. These cones in the basis chosen above are
\begin{align}
    \cC_{\rm BPS-B}^{(0)}&=M(X_{3,165})=\big\{(q_1,q_2,q_3)\in\mathbb{R}^3\,\vert\,q_1\geq 0\,,3q_1\geq -q_2\,,3q_1\geq -q_3\,,3q_1\geq -q_2-q_3\big\}\,,\\
    \cC_{\rm BPS-B}^{(1)}&=\cE(X_{3,165})=\big\{(p^1,p^2,p^3)\in\mathbb{R}^3\,\vert\,p^2\geq 0\,,p^1\geq -3p^2\,,p^1\geq -3p^3\,,p^2\geq -p^3\big\}\,.
\end{align}
These are represented as the red shaded regions shown in the left panel of \cref{fig:X3165}.

Using the prepotential and knowing the K\"ahler cone, we can similarly solve for the attractor solutions in this case. Restricting to the constant-volume slice, the BPS black 0-brane cone is again identified with the cone of movable curves
\begin{equation}
    \cC_{\rm BPS-BB}^{(0)}=\mathrm{Mov}(X_{3,165})=\big\{(q_1,q_2,q_3)\in\mathbb{R}^3\,\vert\,q_2,q_3\geq 0\,,3q_1\geq q_2-q_3\,,3q_1\geq q_3-q_2\big\}\,,
\end{equation}
while the BPS black 1-brane cone is identified with the K\"ahler cone, $\cC_{\rm BPS-BB}^{(1)}=\overline{K(X_{3,165})}$. These four cones are illustrated in \cref{fig:X3165}, where we observe strict inclusions in this geometry, namely $\cC_{\rm BPS-BB}\subset \cC_{\rm BPS-B}$. 

With the BPS black brane cone identified for the 0-branes and 1-branes, we can proceed to compute the genus-0 GV invariants for curve classes and $h^0(X_{3,165},\cO_{X_{3,165}}(D))$ for divisor classes. Their values along a two-dimensional slice of the relevant cones are shown in the right panel of \cref{fig:X3165}. We can see that in the right panel of \cref{fig:X3165 0branes}, the genus-0 GV invariants are non-vanishing and strictly positive in the BPS black brane cone. Similarly, the divisor classes that are not effective lie strictly outside the BPS black 1-brane cone as can be seen in the right panel of \cref{fig:X3165 1branes}. Therefore, \cref{conj:BPS completeness} holds in the BPS-BB cones in M-theory compactified on $X_{3,165}$. The duality between cones preserving the same supersymmetry is manifested through the duality between cones of curves and divisors
\begin{equation}
    \cC_{\rm BPS-B}^{(0)}=(\cC_{\rm BPS-BB}^{(1)})^\vee\,,\qquad \cC_{\rm BPS-BB}^{(0)}=(\cC_{\rm BPS-B}^{(1)})^\vee\,.
\end{equation}
Hence, \cref{conj:cone dual conjecture} also remains satisfied in this theory.

Let us examine the behavior of the masses and tensions of BPS 0-branes and 1-branes along the constant-volume slice $\cF=1$ within the K\"{a}hler cone. For convenience, we introduce the coordinates $u=t^2/t^1$ and $v=t^3/t^1$, such that the K\"ahler cone constraints become $0\leq u,v\leq 1/3$. The constant-volume constraint fixes the following relation
\begin{equation}
    t^1=f(u,v)^{-1/3}\,,\qquad t^2=u f(u,v)^{-1/3}\,,\qquad t^3=v f(u,v)^{-1/3}\,,
\end{equation}
where $f(u,v)=\frac12\left(u+v-3u^2-3v^2+3u^3+3v^3\right)$. With this notation, the mass of a BPS 0-brane of electric charge $q=(q_1,q_2,q_3)$ is
\begin{equation}
    m_q(u,v)=\frac{q_1+q_2u+q_3v}{f(u,v)^{1/3}}\,.
\end{equation}
The tension of a BPS 1-brane with magnetic charge $p=(p^1,p^2,p^3)$ is 
\begin{equation}
    T_p(u,v)=p^i\tau_i(u,v)\,,
\end{equation}
where 
\begin{equation}
    \tau_1(u,v)=\frac{u+v-\frac32(u^2+v^2)}{f(u,v)^{2/3}}\,,\qquad \tau_2(u,v)=\frac{\frac12(1-3u)^2}{{f(u,v)^{2/3}}}\,,\qquad \tau_3(u,v)=\frac{{\frac12(1-3v)^2}}{{f(u,v)^{2/3}}}\,.
\end{equation}

The contour plots in \cref{fig:X3165 mass and tension} illustrate the moduli dependence of masses and tensions of the four types of BPS 0-branes and BPS 1-branes expected from the discussion in \cref{sec:tension properties}. The green dots denote the minima within the plotted domain.
For the BPS 0-branes, the charge $(1,1,1)$ provides an example of a BPS 0-brane within $\cC_{\rm BPS-BB}^{(0)}$ whose mass is minimized in the interior of the K\"{a}hler moduli space. In this case, the minimum of the mass is a stable minimum. We next consider the mass of the particle with electric charge $(1,0,0)$. The mass is $m_{(1,0,0)}=t^1=f(u,v)^{-1/3}$ and minimizing $m_{(1,0,0)}$ is equivalent to maximizing $f(u,v)$ on the allowed square $0\leq u,v\leq 1/3$. The minimum is therefore at the finite-distance boundary $u=v=1/3$, which is the intersection between the two walls $t^1-3t^2=0$ and $t^1-3t^3=0$. The mass remains strictly positive and is approximately $2.08$. The BPS 0-branes with charges $(1,0,-1)$ and $(1,0,-3)$ have minimum masses along the wall $v=1/3$, as their masses are $m_{(1,0,-1)}=t^1(1-v)$ and $m_{(1,0,-3)}=t^1-3t^3=t^1(1-3v)$. The former remains strictly positive throughout the Kähler cone and is minimized at the finite-distance corner $u=v=1/3$. By contrast, the latter vanishes along the entire wall $v=\frac13$. Similarly, the plotted charge $(1,-3,0)$ has $m_{(1,-3,0)}=t^1-3t^2=t^1(1-3u)$ and therefore becomes massless on the wall $u=\frac13$. 

Let us examine the BPS 1-brane tensions along the two-dimensional slice $p^1=5$. As above, this slice is a general choice and the following analysis remains general for all such integral charge sites in $\cC_{\rm BPS-B}^{(1)}$. The BPS 1-brane with magnetic charge $(5,1,1)$ is representative of a BPS black 1-brane whose tension is minimized in the interior of the Kähler moduli space. The minimum occurs approximately near $u=v\approx 0.20$. The tension of the divisor $(5,0,1)$ is minimized along the boundary $u\to 0$, namely $t^2\to 0$ with $v$ remaining finite. This is a finite-distance flop wall of the Kähler cone. The tension of the divisor $(5,0,0)$ is $T_{(5,0,0)}=5[{u+v-\frac32(u^2+v^2)}]/{f(u,v)^{2/3}}$. In the plotted finite window, its smallest value appears at the lower-left corner. In the limit of $u\to 0,v\to 0$, the tension behaves as $T_{(5,0,0)}\sim (u+v)^{1/3}$. Therefore, the tension of this BPS 1-brane vanishes in the infinite-distance limit of the K\"ahler moduli space. The BPS 1-brane with magnetic charge $(5,1,2)$ is minimized on the finite-distance wall $v=1/3$. Additionally, its two-dimensional gradient along the constant-volume slice does not vanish at this boundary minimum, so the minimum is not an interior critical point. Lastly, the divisor $(5,2,0)$ is minimized near the boundary $u=1/3$. As with the boundary class $(5,1,0)$, its tension is minimized along the same finite-distance wall of the K\"ahler cone. This is expected as the boundary ray of $\cC_{\rm BPS-B}$ spanned by $(5,2,0)$ is aligned with the boundary ray of $\cC_{\rm BPS-BB}$ spanned by $(5,1,0)$. 

Hence, the properties of tensions of BPS branes across cones for BPS 0-branes and BPS 1-branes are consistent with those discussed in \cref{sec:tension properties}.

The non-BPS 1-branes with charges on the boundary of $\cC_{\rm BPS-B}^{(1)}$ in \cref{sec:Calabi--Yau threefold with boundary holes} are consistent with \cref{def:BPS occupied 2}. In this example, let us focus on the family of non-BPS 1-branes with charges $p=(5,-1,2+n)$ and $n\geq 0$. These charges are all in the strict interior of $\cC_{\rm BPS-B}^{(1)}$, but lie outside of $\cC_{\rm BPS-BB}^{(1)}$. The central charge of 1-branes within this family of non-BPS 1-branes is
\begin{equation}
    Z_p=\frac{-24u^2+16u+(9n+3)v^2-(6n+2)v+n+1}{2f(u,v)^{2/3}}\,.
\end{equation}
The minimum of $Z_p$ along the constant-volume slice occurs at $u=0,v=1/3$ with $Z_{\rm min}=12^{1/3}$ for every choice of $n\geq 0$. This is the intersection of the flop wall $t^2=0$ and the CFT wall $t^1-3t^3=0$. 
The effective black string potential associated with this family of 1-branes is
\begin{multline}
    V_{\rm BS}(p)=\frac{1}{8f(u,v)^{4/3}}\bigg(
    192u^4-6u^3v-354u^3-144u^2v^2+102u^2v+50u^2\\
    -384uv^3+480uv^2-194uv+42u+3v^4+24v^3-22v^2+6v+1\\
    +n\big(-36u^3v+12u^3-432u^2v^2+324u^2v-60u^2+288uv^2-204uv+36u+18v^4-24v^3+24v^2-12v+2\big)\\
    +n^2\big(-54u^3v+18u^3+54u^2v-18u^2-18uv+6u+27v^4-36v^3+18v^2-6v+1\big)\bigg)\,.
\end{multline}
This potential is minimized at $u=v=1/3$ where the two CFT walls $t^1-3t^2=0$ and $t^1-3t^3=0$ intersect. The tension of the would-be extremal black brane of charge $p$ for all $n\geq 0$ is $T^{\rm Ext}=\frac{5}{\sqrt{2}\cdot 3^{1/6}}$. Thus, the charges in $\cC_{\rm BPS-B}^{(1)}$ associated with non-BPS 1-branes in the spectrum all satisfy
\begin{equation}
    Z_{\rm min}<T^{\rm Ext}\,.
\end{equation}
This inequality is satisfied for all non-BPS 1-branes in this theory and is consistent with \cref{def:BPS occupied 2}.

\section{Conclusions and future directions}
\label{sec:conclusions}

In this paper, we have considered supersymmetric theories and asked which brane charges support BPS states? We have conjectured (and in certain cases proven) that the following strategy gives an answer to this question.  We first find all the classical BPS black brane solutions.  We have conjectured in \cref{conj:BPS completeness} that all the quantized charges allowed by these solutions support BPS states.   The interior of these charge cones for BPS-BB are also charges for which positive minimum tension $T^{\rm min}_{\rm BPS-{BB}}>0$ is achieved in the interior of the moduli space.  To find all the rest of the BPS states (whose tensions fall into the remaining categories in \cref{sec:tension properties}) is less trivial.  However, in cases where the existence of BPS branes are moduli independent, if one considers the cone they generate (over $\mathbb{R}_+$) we have provided evidence that this is dual to the BPS-BB cone of the dual electric/magnetic brane.  This link can be used to deduce the BPS-B cone.

In the following, we will outline a few exciting directions that we believe are interesting to pursue.
\vspace*{-0.75cm}
\paragraph{4d $\cN=2$ theories.} 
Our \cref{conj:BPS completeness} about the completeness of the BPS-BB cone also applies to the cases where the BPS branes jump, such as charged BPS particles for 4d ${\cal N}=2$ supersymmetric theories.  However for those cases we have no general picture about how to identify the BPS branes (or the cones they generate) for given moduli. In 4d $\cN=2$ theories from type IIA compactified on Calabi--Yau threefolds, the relevant invariant is the Donaldson--Thomas invariant counting BPS bound states associated to the D0-D2-D4-D6 configuration wrapping appropriate cycles in the Calabi--Yau threefold. In these theories, BPS states experience a wall-crossing phenomenon where across chambers, BPS states can decay. Therefore, necessarily the BPS-B cone becomes moduli dependent. For type IIB on Calabi--Yau threefold, the relevant cycles are the special-Lagrangian cycles (SLags) of which the D3-branes wrap to form BPS 0-branes in 4d. The dynamics of SLags as we vary moduli or across possible topological transitions of Calabi--Yau manifolds are complicated and furthermore, their behaviors are understood only in limited cases, e.g., in local Calabi--Yau threefolds~\cite{Shapere:1999xr}. It would be interesting to further develop techniques to address these issues.  It is natural to expect that our moduli-independent requirement for the duality between electric/magnetic BPS-B and BPS-BB cones should have generalization to cases such as 4d ${\cal N}=2$ where the BPS states depend both on moduli and the choice of supersymmetry parameterized by an angle.

\vspace*{-0.25cm}
\paragraph{BPS compatibility and vectorial central charges.}
The examples discussed in the present paper represent an albeit diverse, yet small sample within the vast landscape of quantum gravity theories. The above 4d $\cN=2$ theories, which can be thought of as type II string theories compactified on Calabi--Yau threefolds, represent one of the most complicated classes of quantum gravity theories for which both the central charges of BPS branes are complex and the BPS-B cone is moduli dependent due to the presence of walls of marginal stability. However, there are also more theories in which the choice of preserving the same supersymmetry and identifying the BPS-B cone remains challenging, such as heterotic string theories on $T^n$ of which the BPS-B cone depends on moduli. For example, in the case of heterotic string theory compactified on $T^3$, preserving the same supersymmetry among BPS states amounts to a fixed unit vector $\hat{p}_R\in S^2\subset \mathbb{R}^3$. This alone does not fix the BPS-B cone (which is not empty) as we vary the three-plane within $\Gamma^{3,19}\otimes \mathbb{R}$. Moreover, there are no extremal black brane solutions in this theory with a smooth horizon. Therefore, the electric-magnetic pairing between cones becomes ambiguous. On the other hand, via duality, we can study M-theory compactified on K3. The various cones that have been featured in \cref{sec:ECalabi--Yau threefolds} are identical in this scenario. Namely, via geometry, a duality between the ample cone and the Mori cone along with the duality between the cone of movable curves and the effective cone of divisors are manifested. Therefore, it would be interesting to extend our conjectures to these cases where the central charges and choice of supersymmetry among BPS branes become more complicated.

\vspace*{-0.25cm}
\paragraph{Refinement of BPS completeness.}
Among the BPS black 0-brane cones in 5d $\cN=1$ quantum gravity theories, we have already seen a refinement of our BPS completeness \cref{conj:BPS completeness} in the form of the positivity of genus-0 GV invariants in $\cC_{\rm BPS-BB}^{(0)}$ as stated in \cref{conj:GV BPS-BB}. Additionally, in the relevant examples discussed in \cref{sec:QG with 32}, the branes form (threshold) bound states~\cite{Witten:1995im,Sethi:1997pa}. Therefore, all of these seem to suggest that in addition to BPS occupancy at each integral sites within the BPS-BB cone, there is a BPS bound state at each integral charge in the BPS-BB cone. Upon establishing the aforementioned future directions, it would be interesting to examine a wider variety of examples in which upon studying the relevant BPS invariants, we can arrive at a better understanding of which consistent set of BPS states occupy the integral sites in the BPS-BB cone.

\section*{Acknowledgments}
We would like to thank Naomi Gendler, Damian van de Heisteeg, Elijah Sheridan, Mike Stillman, and Max Wiesner for helpful conversations. This work is supported in part by a grant from the Simons Foundation (602883, CV) and the DellaPietra Foundation.

\appendix

\bibliography{papers}

@article{Long:2021lon,
    author = "Long, Cody and Sheshmani, Artan and Vafa, Cumrun and Yau, Shing-Tung",
    title = "{Non-Holomorphic Cycles and Non-BPS Black Branes}",
    eprint = "2104.06420",
    archivePrefix = "arXiv",
    primaryClass = "hep-th",
    reportNumber = "CIMP-D-21-00618R0",
    doi = "10.1007/s00220-022-04587-4",
    journal = "Commun. Math. Phys.",
    volume = "399",
    number = "3",
    pages = "1991--2043",
    year = "2023"
}

@article{skauli,
  title={Curve classes on Calabi--Yau complete intersections in toric varieties},
  author={Skauli, Bj{\o}rn},
  journal={Bulletin of the London Mathematical Society},
  volume={55},
  number={2},
  pages={811--825},
  year={2023},
  publisher={Wiley Online Library},
  eprint = "1911.03146",
  archivePrefix = "arXiv",
  primaryClass = "math.AG",
}

@article{Heidenreich:2016aqi,
    author = "Heidenreich, Ben and Reece, Matthew and Rudelius, Tom",
    title = "{Evidence for a sublattice weak gravity conjecture}",
    eprint = "1606.08437",
    archivePrefix = "arXiv",
    primaryClass = "hep-th",
    doi = "10.1007/JHEP08(2017)025",
    journal = "JHEP",
    volume = "08",
    pages = "025",
    year = "2017"
}

@article{reece2025co,
  title={Co-scaling and alignment of electric and magnetic towers},
  author={Reece, Matthew and Rudelius, Tom and Tudball, Christopher},
  journal={Journal of High Energy Physics},
  volume={2025},
  number={9},
  pages={1--69},
  year={2025},
  publisher={Springer}
}

@article{Nevoa:2025xiq,
    author = "Nevoa, Vinicius and Raman, Sanjay and Vafa, Cumrun",
    title = "{Elementary Constituents Conjecture}",
    eprint = "2511.13813",
    archivePrefix = "arXiv",
    primaryClass = "hep-th",
    month = "11",
    year = "2025"
}

@article{Alim:2021vhs,
    author = "Alim, Murad and Heidenreich, Ben and Rudelius, Tom",
    title = "{The Weak Gravity Conjecture and BPS Particles}",
    eprint = "2108.08309",
    archivePrefix = "arXiv",
    primaryClass = "hep-th",
    reportNumber = "ACFI-T21-09",
    doi = "10.1002/prop.202100125",
    journal = "Fortsch. Phys.",
    volume = "69",
    number = "11-12",
    pages = "2100125",
    year = "2021"
}

@article{Gendler:2026uux,
    author = "Gendler, Naomi and Sheridan, Elijah and Stillman, Michael and Wu, David H.",
    title = "{Holes in Calabi-Yau Effective Cones}",
    eprint = "2603.11173",
    archivePrefix = "arXiv",
    primaryClass = "hep-th",
    month = "3",
    year = "2026"
}

@article{Arkani-Hamed:2006emk,
    author = "Arkani-Hamed, Nima and Motl, Lubos and Nicolis, Alberto and Vafa, Cumrun",
    title = "{The String landscape, black holes and gravity as the weakest force}",
    eprint = "hep-th/0601001",
    archivePrefix = "arXiv",
    reportNumber = "HUTP-05-A0057",
    doi = "10.1088/1126-6708/2007/06/060",
    journal = "JHEP",
    volume = "06",
    pages = "060",
    year = "2007"
}

@article{Ooguri:2016pdq,
    author = "Ooguri, Hirosi and Vafa, Cumrun",
    title = "{Non-supersymmetric AdS and the Swampland}",
    eprint = "1610.01533",
    archivePrefix = "arXiv",
    primaryClass = "hep-th",
    reportNumber = "CALT-TH-2016-027, IPMU16-0139",
    doi = "10.4310/ATMP.2017.v21.n7.a8",
    journal = "Adv. Theor. Math. Phys.",
    volume = "21",
    pages = "1787--1801",
    year = "2017"
}

@article{Demirtas:2018akl,
    author = "Demirtas, Mehmet and Long, Cody and McAllister, Liam and Stillman, Mike",
    title = "{The Kreuzer-Skarke Axiverse}",
    eprint = "1808.01282",
    archivePrefix = "arXiv",
    primaryClass = "hep-th",
    doi = "10.1007/JHEP04(2020)138",
    journal = "JHEP",
    volume = "04",
    pages = "138",
    year = "2020"
}

@article{Banks:2010zn,
    author = "Banks, Tom and Seiberg, Nathan",
    title = "{Symmetries and Strings in Field Theory and Gravity}",
    eprint = "1011.5120",
    archivePrefix = "arXiv",
    primaryClass = "hep-th",
    doi = "10.1103/PhysRevD.83.084019",
    journal = "Phys. Rev. D",
    volume = "83",
    pages = "084019",
    year = "2011"
}

@article{Polchinski:2003bq,
    author = "Polchinski, Joseph",
    editor = "Baer, H. and Belyaev, A.",
    title = "{Monopoles, duality, and string theory}",
    eprint = "hep-th/0304042",
    archivePrefix = "arXiv",
    doi = "10.1142/S0217751X0401866X",
    journal = "Int. J. Mod. Phys. A",
    volume = "19S1",
    pages = "145--156",
    year = "2004"
}

@article{Harlow:2018tng,
    author = "Harlow, Daniel and Ooguri, Hirosi",
    title = "{Symmetries in quantum field theory and quantum gravity}",
    eprint = "1810.05338",
    archivePrefix = "arXiv",
    primaryClass = "hep-th",
    doi = "10.1007/s00220-021-04040-y",
    journal = "Commun. Math. Phys.",
    volume = "383",
    number = "3",
    pages = "1669--1804",
    year = "2021"
}

@article{ottem2015birational,
  title={Birational geometry of hypersurfaces in products of projective spaces},
  author={Ottem, John Christian},
  journal={Mathematische Zeitschrift},
  volume={280},
  number={1},
  pages={135--148},
  year={2015},
  publisher={Springer},
  eprint = "1305.0537",
  archivePrefix = "arXiv",
  primaryClass = "math-AG",
}

@article{Kreuzer:2000xy,
    author = "Kreuzer, Maximilian and Skarke, Harald",
    title = "{Complete classification of reflexive polyhedra in four-dimensions}",
    eprint = "hep-th/0002240",
    archivePrefix = "arXiv",
    reportNumber = "HUB-EP-00-13, TUW-00-07",
    doi = "10.4310/ATMP.2000.v4.n6.a2",
    journal = "Adv. Theor. Math. Phys.",
    volume = "4",
    pages = "1209--1230",
    year = "2000"
}

@article{Hosono:1994ax,
    author = "Hosono, S. and Klemm, A. and Theisen, S. and Yau, Shing-Tung",
    editor = "Greene, B. and Yau, Shing-Tung",
    title = "{Mirror symmetry, mirror map and applications to complete intersection Calabi-Yau spaces}",
    eprint = "hep-th/9406055",
    archivePrefix = "arXiv",
    reportNumber = "HUTMP-94-02, CERN-TH-7303-94, LMU-TPW-94-03",
    doi = "10.1016/0550-3213(94)00440-P",
    journal = "Nucl. Phys. B",
    volume = "433",
    pages = "501--554",
    year = "1995"
}

@article{Maldacena:1997de,
    author = "Maldacena, Juan Martin and Strominger, Andrew and Witten, Edward",
    title = "{Black hole entropy in M theory}",
    eprint = "hep-th/9711053",
    archivePrefix = "arXiv",
    doi = "10.1088/1126-6708/1997/12/002",
    journal = "JHEP",
    volume = "12",
    pages = "002",
    year = "1997"
}

@article{Gopakumar:1998jq,
    author = "Gopakumar, Rajesh and Vafa, Cumrun",
    title = "{M theory and topological strings. 2.}",
    eprint = "hep-th/9812127",
    archivePrefix = "arXiv",
    reportNumber = "HUTP-98-A070",
    month = "12",
    year = "1998"
}

@article{Katz:2020ewz,
    author = "Katz, Sheldon and Kim, Hee-Cheol and Tarazi, Houri-Christina and Vafa, Cumrun",
    title = "{Swampland Constraints on 5d $\mathcal{N}=1$ Supergravity}",
    eprint = "2004.14401",
    archivePrefix = "arXiv",
    primaryClass = "hep-th",
    doi = "10.1007/JHEP07(2020)080",
    journal = "JHEP",
    volume = "07",
    pages = "080",
    year = "2020"
}

@article{Gendler:2022ztv,
    author = "Gendler, Naomi and Heidenreich, Ben and McAllister, Liam and Moritz, Jakob and Rudelius, Tom",
    title = "{Moduli space reconstruction and Weak Gravity}",
    eprint = "2212.10573",
    archivePrefix = "arXiv",
    primaryClass = "hep-th",
    reportNumber = "ACFI-T22-10",
    doi = "10.1007/JHEP12(2023)134",
    journal = "JHEP",
    volume = "12",
    pages = "134",
    year = "2023"
}

@article{McNamara:2019rup,
    author = "McNamara, Jacob and Vafa, Cumrun",
    title = "{Cobordism Classes and the Swampland}",
    eprint = "1909.10355",
    archivePrefix = "arXiv",
    primaryClass = "hep-th",
    month = "9",
    year = "2019"
}

@article{Demirtas:2022hqf,
    author = "Demirtas, Mehmet and Rios-Tascon, Andres and McAllister, Liam",
    title = "{CYTools: A Software Package for Analyzing Calabi-Yau Manifolds}",
    eprint = "2211.03823",
    archivePrefix = "arXiv",
    primaryClass = "hep-th",
    month = "11",
    year = "2022"
}

@article{Witten:1996qb,
    author = "Witten, Edward",
    title = "{Phase transitions in M theory and F theory}",
    eprint = "hep-th/9603150",
    archivePrefix = "arXiv",
    reportNumber = "IASSNS-HEP-96-26",
    doi = "10.1016/0550-3213(96)00212-X",
    journal = "Nucl. Phys. B",
    volume = "471",
    pages = "195--216",
    year = "1996"
}

@article{vandeHeisteeg:2023dlw,
    author = "van de Heisteeg, Damian and Vafa, Cumrun and Wiesner, Max and Wu, David H.",
    title = "{Species scale in diverse dimensions}",
    eprint = "2310.07213",
    archivePrefix = "arXiv",
    primaryClass = "hep-th",
    doi = "10.1007/JHEP05(2024)112",
    journal = "JHEP",
    volume = "05",
    pages = "112",
    year = "2024"
}

@article{Lee:2019wij,
    author = "Lee, Seung-Joo and Lerche, Wolfgang and Weigand, Timo",
    title = "{Emergent strings from infinite distance limits}",
    eprint = "1910.01135",
    archivePrefix = "arXiv",
    primaryClass = "hep-th",
    reportNumber = "CERN-TH-2019-159",
    doi = "10.1007/JHEP02(2022)190",
    journal = "JHEP",
    volume = "02",
    pages = "190",
    year = "2022"
}

@book{lazarsfeld2017positivity,
  title={Positivity in Algebraic Geometry II: Positivity for Vector Bundles, and Multiplier Ideals},
  author={Lazarsfeld, Robert K},
  volume={49},
  year={2004},
  publisher={Springer-Verlag Berlin Heidelberg}
}

@article{Greene:1995hu,
    author = "Greene, Brian R. and Morrison, David R. and Strominger, Andrew",
    title = "{Black hole condensation and the unification of string vacua}",
    eprint = "hep-th/9504145",
    archivePrefix = "arXiv",
    reportNumber = "CLNS-95-1335",
    doi = "10.1016/0550-3213(95)00371-X",
    journal = "Nucl. Phys. B",
    volume = "451",
    pages = "109--120",
    year = "1995"
}

@article{Greene:1996dh,
    author = "Greene, Brian R. and Morrison, David R. and Vafa, Cumrun",
    title = "{A Geometric realization of confinement}",
    eprint = "hep-th/9608039",
    archivePrefix = "arXiv",
    reportNumber = "CU-TP-769, DUKE-TH-96-125, HUTP-96-A033",
    doi = "10.1016/S0550-3213(96)00465-8",
    journal = "Nucl. Phys. B",
    volume = "481",
    pages = "513--538",
    year = "1996"
}

@Inbook{Hirzebruch1966,
    author="Hirzebruch, F.",
    title="The Riemann-Roch theorem for algebraic manifolds",
    bookTitle="Topological Methods in Algebraic Geometry",
    year="1966",
    publisher="Springer Berlin Heidelberg",
    address="Berlin, Heidelberg",
    pages="114--158",
    isbn="978-3-662-30697-0",
    doi="10.1007/978-3-662-30697-0_5",
    url="https://doi.org/10.1007/978-3-662-30697-0_5"
}

@incollection{miyaoka1987chern,
      title={The Chern classes and Kodaira dimension of a minimal variety},
      author={Miyaoka, Yoichi},
      booktitle={Algebraic geometry, Sendai, 1985},
      volume={10},
      pages={449--477},
      year={1987},
      publisher={Mathematical Society of Japan}
}

@article{kawamata1982generalization,
  title={A generalization of Kodaira-Ramanujam's vanishing theorem},
  author={Kawamata, Yujiro},
  journal={Mathematische Annalen},
  volume={261},
  number={1},
  pages={43--46},
  year={1982},
  publisher={Springer}
}

@article{viehweg1982vanishing,
  title={Vanishing theorems.},
  author={Viehweg, Eckart},
  journal = "Journal für die reine und angewandte Mathematik",
  volume = "335",
  pages = "1--8",
  year={1982},
  publisher={Walter de Gruyter, Berlin/New York Berlin, New York}
}

@book{CossecDolgachevLiedtke2025,
  author    = {Fran\c{c}ois Cossec and Igor Dolgachev and Christian Liedtke},
  title     = {Enriques Surfaces I},
  publisher = {Springer Singapore},
  year      = {2025},
  isbn      = {978-981-96-1213-0},
  doi       = {10.1007/978-981-96-1214-7},
  url       = {https://doi.org/10.1007/978-981-96-1214-7}
}

@article{Kim:2024eoa,
    author = "Kim, Hee-Cheol and Vafa, Cumrun and Xu, Kai",
    title = "{Finite landscape of 6d N=(1,0) supergravity}",
    eprint = "2411.19155",
    archivePrefix = "arXiv",
    primaryClass = "hep-th",
    doi = "10.21468/SciPostPhys.20.1.016",
    journal = "SciPost Phys.",
    volume = "20",
    number = "1",
    pages = "016",
    year = "2026"
}

@article{Morrison:1996pp,
    author = "Morrison, David R. and Vafa, Cumrun",
    title = "{Compactifications of F theory on Calabi-Yau threefolds. 2.}",
    eprint = "hep-th/9603161",
    archivePrefix = "arXiv",
    reportNumber = "DUKE-TH-96-107, HUTP-96-A012",
    doi = "10.1016/0550-3213(96)00369-0",
    journal = "Nucl. Phys. B",
    volume = "476",
    pages = "437--469",
    year = "1996"
}

@book{Hartshorne1977,
  author    = {Robin Hartshorne},
  title     = {Algebraic Geometry},
  series     = {Graduate Texts in Mathematics},
  volume     = {52},
  publisher = {Springer},
  address   = {New York},
  year      = {1977},
  isbn      = {978-0-387-90244-9},
  doi       = {10.1007/978-1-4757-3849-0},
  url       = {https://doi.org/10.1007/978-1-4757-3849-0}
}

@article{Ooguri:2006in,
    author = "Ooguri, Hirosi and Vafa, Cumrun",
    title = "{On the Geometry of the String Landscape and the Swampland}",
    eprint = "hep-th/0605264",
    archivePrefix = "arXiv",
    reportNumber = "CALT-68-2600, HUTP-06-A017",
    doi = "10.1016/j.nuclphysb.2006.10.033",
    journal = "Nucl. Phys. B",
    volume = "766",
    pages = "21--33",
    year = "2007"
}

@misc{cygv,
  author       = {Rios Tascon, Andres},
  title        = {{cygv}: Compute GV and GW invariants of Calabi--Yau manifolds},
  year         = {2026},
  howpublished = {\url{https://github.com/ariostas/cygv}},
  note         = {Version 0.2.2}
}

@article{Witten:1995im,
    author = "Witten, Edward",
    title = "{Bound states of strings and p-branes}",
    eprint = "hep-th/9510135",
    archivePrefix = "arXiv",
    reportNumber = "IASSNS-HEP-95-83",
    doi = "10.1016/0550-3213(95)00610-9",
    journal = "Nucl. Phys. B",
    volume = "460",
    pages = "335--350",
    year = "1996"
}

@book{Lazarsfeld2004,
  author    = {Robert Lazarsfeld},
  title     = {Positivity in Algebraic Geometry I: Classical Setting: Line Bundles and Linear Series},
  series     = {Ergebnisse der Mathematik und ihrer Grenzgebiete. 3. Folge / A Series of Modern Surveys in Mathematics},
  volume     = {48},
  publisher = {Springer},
  address   = {Berlin, Heidelberg},
  year      = {2004},
  isbn      = {978-3-540-22533-1},
  doi       = {10.1007/978-3-642-18808-4},
  url       = {https://doi.org/10.1007/978-3-642-18808-4}
}

@book{Huybrechts2005,
  author    = {Daniel Huybrechts},
  title     = {Complex Geometry: An Introduction},
  series     = {Universitext},
  publisher = {Springer},
  address   = {Berlin, Heidelberg},
  year      = {2005},
  isbn      = {978-3-540-21290-5},
  doi       = {10.1007/b137952},
  url       = {https://doi.org/10.1007/b137952}
}

@article{Sethi:1997pa,
    author = "Sethi, Savdeep and Stern, Mark",
    title = "{D-brane bound states redux}",
    eprint = "hep-th/9705046",
    archivePrefix = "arXiv",
    reportNumber = "IASSNS-HEP-97-45, DUK-M-97-5",
    doi = "10.1007/s002200050374",
    journal = "Commun. Math. Phys.",
    volume = "194",
    pages = "675--705",
    year = "1998"
}

@article{lazic2020generalised,
  title={On generalised abundance, I},
  author={Lazi{\'c}, Vladimir and Peternell, Thomas},
  journal={Publications of the Research Institute for Mathematical Sciences},
  volume={56},
  number={2},
  pages={353--389},
  year={2020}
}

@article{Seiberg:1996bd,
    author = "Seiberg, Nathan",
    title = "{Five-dimensional SUSY field theories, nontrivial fixed points and string dynamics}",
    eprint = "hep-th/9608111",
    archivePrefix = "arXiv",
    reportNumber = "RU-96-69",
    doi = "10.1016/S0370-2693(96)01215-4",
    journal = "Phys. Lett. B",
    volume = "388",
    pages = "753--760",
    year = "1996"
}

@article{Shapere:1999xr,
    author = "Shapere, Alfred D. and Vafa, Cumrun",
    title = "{BPS structure of Argyres-Douglas superconformal theories}",
    eprint = "hep-th/9910182",
    archivePrefix = "arXiv",
    reportNumber = "HUTP-99-A057, UKHEP-99-15",
    month = "10",
    year = "1999"
}

@article{Larfors:2019sie,
    author = "Larfors, Magdalena and Schneider, Robin",
    title = "{Line bundle cohomologies on CICYs with Picard number two}",
    eprint = "1906.00392",
    archivePrefix = "arXiv",
    primaryClass = "hep-th",
    reportNumber = "UUITP-18/19",
    doi = "10.1002/prop.201900083",
    journal = "Fortsch. Phys.",
    volume = "67",
    number = "12",
    pages = "1900083",
    year = "2019"
}
\bibliographystyle{JHEP}

\end{document}